\documentclass[11pt]{article}

\usepackage{amsmath,amsthm,amsbsy,amssymb,relsize,cases,ulem,bm,latexsym,bbm,dsfont,stmaryrd,wasysym,esint,cmll,graphicx,psfrag,epsf,enumerate,ragged2e,geometry,colortbl,caption,longtable,float,booktabs,multirow,fancyhdr,setspace,titlesec,graphicx,rotating,abstract,pdflscape,arydshln,dashrule,subcaption,hyperref,cleveref,mathtools,enumitem,xstring,threeparttable}

\usepackage{upgreek}

\usepackage{amsfonts}

\usepackage{mathrsfs}
\usepackage{xr-hyper}

\allowdisplaybreaks[4]
\usepackage[dvipsnames]{xcolor}
\usepackage[many]{tcolorbox}
\usepackage{fancyhdr}
\newcommand{\RNum}[1]{\uppercase\expandafter{\romannumeral #1\relax}}
\titleformat{\section}{\normalfont\large\bfseries}{\thesection.}{1em}{}
\titleformat{\subsection}{\normalfont\large\bfseries}{\thesubsection}{1em}{}
\newcommand{\rn}{\romannumeral}
\titleformat{\subsubsection}{\normalfont\bfseries}{\thesubsubsection}{1em}{}

\usepackage{snaptodo}
\snaptodoset{block rise=2em}
\snaptodoset{margin block/.style={font=\scriptsize}} 
\snaptodoset{chain bias=5cm} 
\newcommand{\ag}[1]{%
	\snaptodo[margin block/.append style=green!50!black]{%
		{\sloppy\textbf{Abhi}: #1}}}
\makeatletter
\renewcommand\@biblabel[1]{}

\makeatother
\makeatletter
\newcommand{\labeltext}[1]{%
	#1%
	\protected@edef\@currentlabel{#1}%
}
\makeatother
\newcommand{\propref}[2]{%
	\hyperref[#1]{\ref*{#1} {#2}}%
}
\DeclareMathAlphabet{\mathsuet} {T1} {wesu}{bx}{sl}
\makeatletter
\def\pmb@#1#2{\setbox8\hbox{$\m@th#1{#2}$}%
	\setboxz@h{$\m@th#1\mkern.9mu$}\pmbraise@\wdz@
	\binrel@{#2}%
	\dimen@-\wd8 %
	\binrel@@{%
		\mkern-.1mu\copy8 %
		\kern\dimen@\mkern-.3mu\copy8 %
		\kern\dimen@\mkern.3mu\copy8 %
	}
}
\makeatother
\newtheorem{assumption}{Assumption}%[section]
\newtheorem{theorem}{Theorem}[section]

\newtheorem{definition}{Definition}%[section]

\newtheorem{lemma}{Lemma}%[section]
\newtheorem{notation}{Notation}

\newtheorem{remark}{Remark}%[section]
\renewenvironment{proof}{{\noindent\it Proof.}\quad}{\hfill$\blacksquare$\par}
\newcommand{\tr}{\operatorname{tr}}
\newcommand{\E}{\operatorname{E}}
\newcommand{\Var}{\operatorname{Var}}

\newcommand{\bnorm}{\big\Vert}

\newcommand{\bc}{\mathbf{c}}
\newcommand{\kwy}{j_{1}}

\newcommand{\kwu}{j_{3}}
\newcommand{\Diag}{\mathtt{Diag}}
\newcommand{\diagM}{\mathtt{diagM}}

\newcommand{\bZ}{\mathbf{Z}}

\newcommand{\Varm}{\mathbf{\Omega}}%{\pmb{\mathcal{V}}}

\newcommand{\bQ}{\mathbf{Q}}
\newcommand{\bJ}{\mathbf{J}}

\newcommand{\lin}{\texttt{line}}
\newcommand{\qua}{\texttt{quad}}
\newcommand{\bY}{\mathbf{Y}}
\newcommand{\bV}{\mathbf{V}}
\newcommand{\bX}{\mathbf{X}}

\newcommand{\bphi}{\pmb{\phi}}
\newcommand{\bvarPhi}{\pmb{\Phi}}

\newcommand{\bS}{\mathbf{S}}

\newcommand{\lwy}{\ell_{1}}
\newcommand{\lwyl}{\ell_{2}}
\newcommand{\lwu}{\ell_{3}}
\newcommand{\dwy}{\lwy}

\newcommand{\dwu}{\lwu}
\newcommand{\Sigep}{\mathbf{\Sigma}}

\newcommand{\bzero}{\mathbf{0}}
\newcommand{\bP}{\mathbf{P}}

\newcommand{\lpn}{\ell_{p}}
\newcommand{\lmn}{\ell_{m}}
\newcommand{\cI}[1]{\mathbf{1}[#1]}
\newcommand{\bestY}{\mathbb{Y}}

\newcommand{\bard}{\bar{d}}
\newcommand{\cB}{\mathcal{B}}

\newcommand{\cA}{\mathscr{A}}
\newcommand{\bC}{\mathbf{C}}

\newcommand{\lqn}{\ell_{q}}
\newcommand{\bE}{\mathbf{E}}
\newcommand{\bL}{\mathbf{L}}

\newcommand{\bM}{\mathbf{M}}

\newcommand{\rank}[1]{\mathrm{rank}(#1)}
\newcommand{\spnorm}[1]{\|#1\|_{\mathrm{sp}}}

\newcommand{\Fnorm}[1]{\|#1\|}%_{\mathrm{F}}
\newcommand{\bFnorm}[1]{\Big\Vert#1\Big\Vert}%_{\mathrm{F}}

\newcommand{\colnorm}[1]{\|#1\|_{1}}
\newcommand{\rc}{\mathrm{rc}}
\newcommand{\rcnorm}[1]{\|#1\|_{\rc}}
\newcommand{\rownorm}[1]{\|#1\|_{\infty}}

\newcommand{\approwy}{\varsigma_{1}}
\newcommand{\approwyl}{\varsigma_{2}}
\newcommand{\approwu}{\varsigma_{3}}
\newcommand{\pk}{p_{k}}
\newcommand{\gk}{g_{k}}

\newcommand{\approgk}{\varsigma_{k}}

\newcommand{\Bt}{\mathcal{B}_{t}}
\newcommand{\bH}{\mathbf{H}}
\newcommand{\hTt}{h_{Tt}}

\newcommand{\dlim}{\stackrel{d}{ \to}}
\newcommand{\mNL}{m_{N}^{\mathtt{line}}}
\newcommand{\mNQ}{m_{N}^{\mathtt{quad}}}
\newcommand{\mNQk}{m_{Nk}^{\mathtt{quad}}}
\newcommand{\invN}{\frac{1}{n(T-1)}}
\newcommand{\invNsq}{\frac{1}{\sqrt{n(T-1)}}}

\newcommand{\cJ}{\mathcal{J}}
\newcommand{\mC}{\mathfrak{C}}
\newcommand{\abs}{\mathrm{abs}}
\newcommand{\QED}{\\\noindent\makebox[\linewidth][r]{$\blacksquare$}\\}
\newcommand{\te}{\mathtt{e}}
\newcommand{\mn}{\ell_n}
\usepackage{xparse,etoolbox}
\newcommand{\errorD}[2]{\ell_{#2}^{#1\varsigma_{#2}}}

\NewDocumentCommand{\dzn}{o}{%
	\ensuremath{%
		\IfValueTF{#1}{\ell_{z#1}}{\ell_{z}}
	}%
}

\NewDocumentCommand{\maxerrorD}{o}{%
	\ensuremath{%
		\max_{1\leq k\leq 3}
		\{\ell_{k}^{%
			\IfValueTF{#1}{#1\varsigma_{k}}{\varsigma_{k}}%
		}\}%
	}%
}

\newcommand{\elln}{\ell_{n}}
\newcommand{\bEL}{\pmb{E}}
\NewDocumentCommand{\YEL}{m o}{%
	\ifstrequal{#1}{N}{%
		B_{3}^{-1}\pmb{\mathcal{E}}_{\!#1}% 
	}{%
		B_{3}^{-1}\mathcal{E}_{#1}%         % 
	}%
	\IfValueT{#2}{^{#2}}%
}

\newcommand{\lchi}{\scalebox{1.2}{$\chi$}} 
\NewDocumentCommand{\chiL}{m o}{%
	\ifstrequal{#1}{N}{%
		\pmb{\lchi}_{\!#1}% 
	}{%
		{\lchi}_{#1}%         % 
	}%
	\IfValueT{#2}{^{#2}}%
}

\NewDocumentCommand{\UL}{m o}{%
	\ifstrequal{#1}{N}{\mathbf{U}}{U}_{#1}%
	\IfValueT{#2}{^{#2}}%
}
\newcommand{\bVL}{\mathbf{V}}
\NewDocumentCommand{\JL}{m o}{%
	\ifstrequal{#1}{N}{\mathbf{J}}{J}_{#1}%
	\IfValueT{#2}{^{#2}}%
}
\NewDocumentCommand{\YL}{m o}{%
	\ifstrequal{#1}{N}{\mathbf{Y}}{Y}_{#1}%
	\IfValueT{#2}{^{#2}}%
}
\NewDocumentCommand{\dYL}{m o}{%
	\ifstrequal{#1}{N}{\ddot{\mathbf{Y}}}{\ddot{Y}}_{#1}%
	\IfValueT{#2}{^{#2}}%
}

\NewDocumentCommand{\ZL}{m o}{%
	\ifstrequal{#1}{N}{\mathbf{Z}}{Z}_{#1}%
	\IfValueT{#2}{^{#2}}%
}
\NewDocumentCommand{\XL}{m o}{%
	\ifstrequal{#1}{N}{\mathbf{X}}{X}_{#1}%
	\IfValueT{#2}{^{#2}}%
}

\newcommand{\brL}{\mathbf{r}}
\NewDocumentCommand{\KL}{m o}{%
	\ifstrequal{#1}{N}{\mathbf{K}}{K}_{#1}%
	\IfValueT{#2}{^{#2}}%
}
\newcommand{\MSL}[1]{%
	\ifstrequal{#1}{N}{%
		\pmb{M}({\bm{\Sigma}_{#1}}) %
	}{%
		{M}({\Sigma_{#1}}) %
	}%
}

\newcommand{\MSLhat}[1]{%
	\ifstrequal{#1}{N}{%
		\pmb{M}(\tilde{\bm{\Sigma}}_{#1})%
	}{%
		M(\tilde{\Sigma}_{#1}) %
	}%
}

\newcommand{\JSL}[1]{%
	\ifstrequal{#1}{N}{%
		\pmb{J}({\bm{\Sigma}_{#1}}) %
	}{%
		{J}({\Sigma_{#1}}) %
	}%
}

\newcommand{\JSLhat}[1]{%
	\ifstrequal{#1}{N}{%
		\pmb{J}({\tilde{\bm{\Sigma}}_{#1}}) %
	}{%
		{J}({\tilde{\Sigma}_{#1}}) %
	}%
}

\NewDocumentCommand{\QL}{m o}{%
	\IfSubStr{#1}{N}{%
		\mathbf{Q}_{\!#1}%       %
	}{%
		Q_{#1}%                % 
	}%
	\IfValueT{#2}{^{#2}}%
}

\NewDocumentCommand{\QLhat}{m o}{%
	\IfSubStr{#1}{N}{%
		\tilde{\mathbf{Q}}_{\!#1}%       %
	}{%
		\tilde{Q}_{#1}%                % 
	}%
	\IfValueT{#2}{^{#2}}%
}

\NewDocumentCommand{\PL}{m o}{%
	\IfSubStr{#1}{N}{%
		\mathbf{P}_{\!#1}%       %
	}{%
		P_{#1}%                % 
	}%
	\IfValueT{#2}{^{#2}}%
}

\NewDocumentCommand{\PLhat}{m o}{%
	\IfSubStr{#1}{N}{%
		\tilde{\mathbf{P}}_{\!#1}%       %
	}{%
		\tilde{P}_{#1}%                % 
	}%
	\IfValueT{#2}{^{#2}}%
}

\NewDocumentCommand{\hatSigma}{m o}{%
	\ifstrequal{#1}{N}{\widehat{\mathbf{\Sigma}}}{\widehat{\Sigma}}_{#1}%
	\IfValueT{#2}{^{#2}}%
}

\newcommand{\invB}[1]{B_{#1}^{-#1}}
\newcommand{\invS}[1]{S_{#1}^{-#1}}
\newcommand{\invR}[1]{R_{#1}^{-#1}}

\newcommand{\cQ}{\mathscr{Q}}
\newcommand{\cQtilde}{\tilde{\cQ}}
\newcommand{\sigmaI}[1]{\mathcal{I}_{#1}}
\newcommand{\fe}{\mathbf{c}_{n}+\alpha_{t}l_n}
\newcommand{\lnT}{\ln{T}}
\newcommand{\vsinf}{\underline{\varsigma}}

\newcommand{\mumin}{\mu_{\min}}
\newcommand{\mat}[1]{\{#1\}}
\newcommand{\vecD}{\mathtt{vec}}
\newcommand{\Tau}{\mathfrak{s}}
\newcommand{\calX}{\mathcal{X}}
\newcommand{\cF}{\mathcal{F}}
\newcommand{\sU}{\mathscr{U}}

\newcommand{\sG}{\mathscr{G}}
\newcommand{\sGN}[1]{m_{N}'(#1)\Omega_{N}^{-1}m_{N}(#1)}
\newcommand{\EsGN}[1]{\bigl[\E\bigl(m_{N}(#1)\bigl)\bigr]'\Omega_{N}^{-1}\E\bigl(m_{N}(#1)\bigl)}
\newcommand{\suma}[2]{\sum_{#1}^{#2}}
\newcommand{\Op}{O_{p}}
\newcommand{\op}{o_{p}}
\newcommand{\half}{\frac{1}{2}}

\newcommand{\stl}{\setlength{\tabcolsep}}
\newcommand{\stla}{\setlength{\abovecaptionskip}}

\usepackage{tocloft}
\usepackage[round]{natbib}
\usepackage{multibib}
\newcites{mmc}{Supplement Reference}
\usepackage[commandnameprefix=always]{changes}%[markup=default,commandnameprefix=always]
\setaddedmarkup{\textcolor{magenta}{#1}}
\usepackage{soul}
\setdeletedmarkup{\textcolor{gray}{\sout{#1}}}
\providecommand{\keywords}[1]
{
	\small
	\\
	\textit{Keywords:} \textit{#1}
}
\providecommand{\JEL}[1]
{
	\small
	\\~\\
	\textit{JEL classification: } \textit{#1}
}
\newcommand{\titlename}{Semi-nonparametric estimation of spatial dynamic panel data models with nonparametric spatial weights}
\title{\titlename\thanks{Gupta's research was supported in part by the Social Science and Humanities Research Council of Canada. Qu's research was supported by the National Natural Science Foundation of China via grant 72595872. Zhang's research was supported by the National Natural Science Foundation of China via grant 72503134.}}
\author{Abhimanyu Gupta\thanks{Department of Economics, Queen's University, Dunning Hall, 94 University Avenue, Kingston, Ontario K7L 3N6, Canada, and Department of Economics, University of Essex, Wivenhoe Park, Colchester CO4 3SQ, UK. Email: abhimanyu.g@queensu.ca.} 
	\and Xi Qu\thanks{Department of Economics, Antai College of Economics and Management, Shanghai Jiao Tong University, 1954 Huashan Road, Shanghai, 200030, China PRC. Email: xiqu@sjtu.edu.cn.}
	\and Jiajun Zhang\thanks{International Business School, Shanghai University of International Business and Economics, 201620, China PRC. Email: jiajun30@suibe.edu.cn.}}
\begin{document}
	\maketitle 
	\begin{abstract}
		We develop a semi-nonparametric framework for spatial dynamic panel data (SDPD) models with two-way fixed effects when the spatial interaction structure is unknown beyond a distance measure. This is accomplished by modelling spatial weights in the outcome, lagged-outcome, and disturbance channels as unknown functions of underlying economic distances. These enter the SDPD system through matrix-function operators, providing a unified approach that accommodates both spatial autoregressive and matrix exponential spatial specifications. Allowing for unknown heteroskedasticity, we propose sieve GMM estimators based on a stacked set of linear and quadratic moment conditions, and derive a feasible optimal GMM estimator and a more efficient feasible best GMM estimator. As $(n,T)\to\infty$, the parametric component is $\sqrt{n(T-1)}$-consistent and asymptotically normal, echoing classical semi-nonparametric results. Monte Carlo experiments indicate excellent finite-sample performance. We apply the method to `witch' killings as studied by \citet{miguel2005poverty}, and find that economic-geography proximity rather than cultural-geography proximity between communities significantly amplifies spatial dependence in these economic murders.
		\JEL{C13, C23}
		\keywords{Spatial dynamic panel data models, Unknown spatial weights, Sieve-based GMM estimation, Unknown heteroskedasticity, semiparametric, nonparametric}
	\end{abstract}
	%%%%%%%%%%%%%%%%%
	% \begingroup
	% \renewcommand{\addcontentsline}[3]{}
	%%%%%%%%%%%%%%%	
	\newpage
	\section{Introduction}
	This paper develops a general framework for \textit{semi-nonparametric} estimation of and inference on spatial dynamic panel data (SDPD) models with two-way fixed effects (FE), in the sense that the model contains both finite dimensional and infinite dimensional parameters of interest. SDPD models provide a workhorse framework for quantifying spillovers across units and over time, with applications in trade, public finance, housing, and regional economics. We contribute to the spatial autoregressive (SAR) literature that continues to develop extensively to accommodate various data structures and estimation challenges. Foundational studies established quasi-maximum likelihood (QML) and generalized method of moments (GMM) approaches for basic SDPD specifications \citep{yu2008quasi-maximum, lee2010, lee2014efficient}. This literature has subsequently expanded to address specific panel dynamics, such as spatial cointegration and short panels \citep{yu2012estimation, su2015qml}, as well as unobserved heterogeneity, including interactive fixed effects and latent group structures \citep{shi2017spatial, kuersteiner2020dynamic, bai2021dynamic, su2023identifying}.
	% SDPD models provide a workhorse framework for quantifying spillovers across units and over time, with applications in trade, public finance, housing, and regional economics; see, e.g.
	% \citet{yu2008quasi-maximum,lee2010,yu2012estimation,lee2014efficient,su2015qml,shi2017spatial,yang2018unified,kuersteiner2020dynamic,li2021spatial,bai2021dynamic,su2023identifying}.
	The SDPD framework is particularly versatile because it simultaneously incorporates three channels of spatial spillovers: the outcome, the lagged outcome, and the disturbance. A leading specification with two-way FE allows for contemporaneous spatial-lag dependence, space-time-lag dependence, and spatially correlated disturbances \citep{yang2018unified}. 
	% In these standard models, the individual and time fixed effects are typically combined with spatial weights matrices that are treated as known.
	\par 
	A key maintained assumption in such models is that the spatial interaction matrices are prespecified. However, this assumption may be too strict in empirical applications. Spatial weights are typically constructed from geographic or economic proximity, and misspecification may bias estimated spillovers and distort counterfactual policy conclusions. Even within a single channel, a single pre-specified matrix can be restrictive when spillovers operate through multiple distance components or higher-order paths. Moreover, linear spatial-lag terms may fail to capture nonlinear interaction patterns generated by matrix function operators, such as matrix exponential spatial specifications (MESS); see, e.g. \citet{lesage2007matrix,han2013model,debarsy2015large,yang2022unified}. In many empirical settings, researchers only know that the interaction depends on exogenous distances, but a specific parametric link is unavailable \citep{pinkse2002spatial}.
	\par 
	Motivated by these considerations, we estimate the spatial interaction structure directly from the data. We allow the weight between units in each channel to be an unknown function of observable, exogenous, and time-invariant distances. These unknown distance link functions are mapped into interaction matrices through matrix functions. This formulation, which we formalize in Section \ref{sec:model}, preserves the three-channel SDPD structure while allowing each channel to feature its own data-driven spatial interaction pattern, including SAR and MESS specifications.
	\par 
	Directly estimating these weights without structural assumptions or a penalized approach (see e.g. \citealp{lam2020estimation}) is infeasible because the number of unknown links grows quadratically with the cross-sectional dimension $n$. To accommodate settings where $n$ is comparable to or larger than the time dimension $T$, we impose structure through an observed exogenous distance measure, and approximate each unknown link function by a sieve expansion. This reduces the number of unknown links from quadratic in $n$ to a much smaller set of sieve coefficients that grows only slowly with sample size $nT$.
	\par 
	There is a growing literature on estimating unknown spatial links. For cross-sectional models, \citet{pinkse2002spatial} and \citet{sun2016functional} estimate a distance-based weights function via series approximation. For static spatial panels, \citet{lam2020estimation} and \citet{de2025identifying} exploit sparsity using LASSO-type methods, though their theory requires $T$ to be large relative to $n$. In contrast, \citet{chen2025npspatial} adopts a series approximation that accommodates large $n$ when $T$ is finite or moderately large. In this paper, we also employ a series-based approach, but extend the scope in two significant ways. First, whereas existing research typically allows spatial interaction only through the outcome variable, our framework simultaneously incorporates all three SDPD channels: the outcome, the lagged outcome, and the disturbance. Second, based on the SDPD model, we specify the unknown spatial weights themselves as varying coefficients that depend on an exogenous distance measure. This marks a departure from the existing varying-coefficient SAR literature, which typically treats the spatial weights matrix as pre-specified and instead allows the regression coefficients to vary with covariates (see, e.g. \citealp{su2012sieve, sun2018estimation, hoshino2022sieve, gupta2025wald}).
	\par 
	In this paper, we propose the sieve generalized method of moments estimator (GMME) for the SDPD models with unknown spatial weights matrices. Using sieve-based GMMEs in the case of unknown heteroskedasticity, we construct a feasible optimal GMME (OGMME) and a more efficient feasible best GMME (BGMME) based on quadratic and linear moments. Our study is closely related to \citet{pinkse2002spatial}, \citet{sun2016functional}, and \citet{chen2025npspatial}. The first two studies consider cross-sectional SAR models and derive the two-stage least squares (2SLS) using linear moment conditions. The latter study analyzes a static SAR spatial panel model with individual FE and derives the efficient 2SLS estimator (2SLSE), again based on linear moments. 
	
	Our framework differs from these studies in two important aspects. 
	First, we develop a unified SDPD framework that accommodates both SAR and MESS specifications, contrasting with \citet{pinkse2002spatial}, \citet{sun2016functional}, and \citet{chen2025spatial}, who restrict their attention to SAR models in cross-sectional or static panel settings. This is non-trivial because the interaction between the time-lagged dependent variable and the sieve approximation error introduces complex bias terms that are absent in static settings and must be controlled.
	Second, our semi-nonparametric SDPD model considers unknown heteroskedastic disturbances and develops GMM estimators based on quadratic and linear moments. This setting has not been previously explored.%
	\footnote{Our model can be viewed as a semi-nonparametric counterpart of \citet{lee2014efficient}, delivering a more flexible specification of the spatial interaction structure. %Related work has examined SDPD models with parametric endogenous weights, for example \citet{qu2017qml} and \citet{kuersteiner2020dynamic}. Extending our semiparametric framework to SDPD models with functionally unknown endogenous weights is a natural but technically more involved direction and is left for future work.
		
	}
	Under certain rate conditions, the estimator of the finite-dimensional parameters is $\sqrt{n(T-1)}$-consistent and asymptotically normal as $(n,T)\to\infty$.
	Monte Carlo experiments demonstrate that the proposed estimators exhibit excellent finite-sample performance.
	\par 
	We then apply our methods to investigate the spillover effects of so-called `witch' killing. Many studies have noted that poverty and violence occur simultaneously \citep{huang2004crime,mehlum2006poverty,Khanna20190547,mcguirk2020economic,HEILMANN2021104408}. There is a strong negative
	relationship between economic growth and crime across countries \citep{miguel2005poverty} . \cite{miguel2005poverty} empirically finds that exogenous extreme rainfall causes lower income in the region, and leads to a large increase in the murder of `witches', who are nearly all elderly women. The results show that income shocks are the driver of the increase in witch killings, instead of a scapegoat culture. Based on \cite{miguel2005poverty}, we apply our model to investigate the mechanism of spillover effects in the study of witch killings. Our results show that the economic-geography proximity rather than the cultural-geography proximity between communities significantly amplifies spatial dependence in witch killings. 
	\par 
	The rest of this paper is organized as follows. Section \ref{sec:model} introduces the formal model specification, Section \ref{sec:gmm} provides the sieve-based GMM estimation methods, and Section \ref{sec:asy} investigates the large sample properties of GMMEs. This includes a discussion of the asymptotic efficiency of the proposed estimators under unknown heteroskedasticity. Monte Carlo results for various estimators are provided in Section \ref{sec:mc}. Section \ref{sec:witch} presents our empirical applications to the witch killings. The appendix contains the regularity assumptions, key lemmas, and proofs of the theorems, while the online supplementary material provides additional lemmas and simulation results.
	\\
	%%     -----------------------------         %%
	\textit{Notation.} For a real symmetric matrix $H$, let $\mu_{\max}(H)$ and $\mu_{\min}(H)$ denote its largest and smallest eigenvalues, and let $\rho(H)$ denote its spectral radius, which is defined as the maximum absolute value of its eigenvalues. For a generic real matrix $H$, we write $\Fnorm{H}$ and $\spnorm{H}$ for the Frobenius and spectral norms, respectively. Let $\rcnorm{H}\equiv\max\mat{{\colnorm{H},\rownorm{H}}}$, where $\colnorm{H}$ stands for its column sum norm and $\rownorm{H}$ stands for its row sum norm. When $H$ is square, let $H^{s}=H+H'$, let $\diagM(H)$ denote the vector of diagonal entries of $H$, and let $\Diag_{i=1}^r(H_{i})$ denote the block diagonal matrix whose diagonal blocks are $H_1, \cdots, H_r$. 
	% For the sieve approximation, we define $\elln=\max_{1\leq k\leq 3}\ell_{k}$ and $\vsinf=\min_{1\leq k\leq 3}\varsigma_{k}$.
	\color{black}
	%%         ----------------------------------         %%
	\section{Model specification}\label{sec:model}	
	%%%                 Modifications              %%%
	%%%  -----------------------------------------------  %%%
	For $i=1,\dots,n$ and $t=1,\dots,T$, a standard parametric SDPD model typically assumes that an observable outcome $y_{it}$ is a linear function of $\sum_{j=1}^{n}w_{1ij}y_{jt}$ and $\sum_{j=1}^n w_{2ij}y_{j,t-1}$, for some known `spatial weights' $w_{1ij}$ and $w_{2ij}$. 
	We extend this to a semi-nonparametric setting by modeling the spatial links as unknown functions of observable exogenous distances. The scalar form of the model for unit $i$ at time $t$ can be expressed as:
	\[\setlength{\abovedisplayskip}{3pt}
	y_{it} = \sum_{j=1}^n g_1(d_{ij})y_{jt}+\gamma y_{i,t-1} + \sum_{j=1}^n g_2(d_{ij})y_{j,t-1}+x_{it}'\beta+c_i+\alpha_{t}+ u_{it},
	\setlength{\belowdisplayskip}{3pt}
	\]
	with $u_{it}=\sum_{j=1}^n g_3(d_{ij})u_{jt}+\epsilon_{it}$, for unknown functions $g_k(\cdot), k=1,2,3$. Here, $d_{ij}$ represents some observed time-invariant exogenous distance measures between units $i$ and $j$, $x_{it}$ is an $\ell_x\times 1$ vector of exogenous regressors, $c_i$ and $\alpha_t$ are unobserved individual and time fixed effects, respectively, and $\epsilon_{it}$ is an unobserved error term. 
	By stacking these observations across all $n$ units, we generalize these interactions through matrix-function operators, which allows us to provide a unified approach that accommodates both SAR and MESS specifications. 
	\color{black}
	%%%  -------------------------------------------------------  %%%
	Specifically, the matrix form of our model is:
	\begin{equation} \label{eq:sdpd_np}
		B_{1}Y_{t}=(\gamma I_n+B_2)Y_{t-1}+X_{t}\beta+\bc_{n}+\alpha_{t}l_{n}+U_{t}, \quad B_{3}U_{t}=E_{t},
	\end{equation}
	where the elements of the $n\times n$ matrices $B_k$ are determined by the unknown distance link functions $g_k(d_{ij})$.  These link functions form the spatial weight matrices $G_k=\mat{g_k(d_{ij})}$, which in turn characterize the operators $B_k$ as follows:
	\begin{flalign}\label{eq:matrix function}
		\begin{split}
			\text{SAR:} \
			B_1&=I_n-G_1, \ B_2=G_2, \ B_3=I_n-G_3; \\
			\text{MESS:} \ 
			B_k&=e^{G_k} \ \text{for $k=1,2,3$},
		\end{split}    
	\end{flalign}
	with $I_{n}$ the $n \times n$ identity matrix. Moreover, $Y_{t}=(y_{1t},y_{2t},...,y_{nt})'$, $E_{t}=(\epsilon_{1t},\epsilon_{2t},...,\epsilon_{nt})'$ are $n \times 1$ column vectors and the $\epsilon_{it}$ are independent with zero means and homoskedastic or unknown heteroskedastic variances across $i$ and $t$, and $U_{t}=(u_{1t},...,u_{nt})'$, which is a SAR/MESS process.
	$X_{t}$ is an $n\times \ell_x$ matrix of time varying regressors, $\bc_{n}=(c_{1},...,c_{n})'$ denotes the vector of unobserved individual FE, $a_{t}$ is a scalar unobservable time FE, and $l_{n}$ is the $n\times 1$ vector of ones. The initial condition $Y_0$ is exogenously given. We impose the normalization $l_{n}'\bc_{n}=0$ to avoid the loss of identification of $c_{i}$ and $\alpha_{t}$ because $c_{i}+\alpha_{t}=(c_{i}-\nu)+(\alpha_{t}+\nu)$ for an arbitrary $\nu$. 
	For both specifications, we ensure system stability by imposing the spectral-radius condition $\rho(A)<1$, where $A=B_{1}^{-1}(\gamma I_n+B_2)$.
	\begin{remark}\label{remark:MESS}
		For the MESS specification, setting $\gamma=\varphi-1$ allows the model in \eqref{eq:sdpd_np} to be interpreted in the intuitive form $e^{G_1}Y_{t} = (\varphi I_n + (e^{G_2}-I_n)) Y_{t-1} + X_{t}\beta+\bc_{n}+\alpha_{t}l_{n}+U_{t}$ with $e^{G_3}U_{t}=E_{t}$. In this representation, $\varphi Y_{t-1}$ captures the pure time autoregressive persistence, while the operator $(e^{G_2}-I_n)$ acting on $Y_{t-1}$ isolates the spatio-temporal spillover effect. By the standard matrix exponential expansion $e^{G_2}-I_n=\sum_{l=1}^{\infty}\frac{G_2^l}{l!}$, the MESS specification captures higher-order spillovers with factorially decaying weights. This stands in distinct contrast to the SAR specification, where the spatio-temporal effect is driven solely by $G_{2}Y_{t-1}$, restricting the immediate spillover to first-order neighbors. 
		Unlike the SAR specification which requires $\rho(B_1)<1$ and $\rho(B_3)<1$ for invertibility, the MESS ensures $B_1$ and $B_3$ are invertible for any finite $G_1$ and $G_3$ \citep{debarsy2015large}.
	\end{remark}
	Throughout the paper, the subscript $k=1, 2, 3$ refers to the contemporaneous outcome, lagged outcome, and disturbance channels, respectively. For estimation, we proceed by approximating $g_{k}(\cdot)$ via series, i.e., $g_{k}(d)=\sum_{\pk=1}^{\ell_{k}}\lambda_{\pk}\phi_{k\pk}(d)+\delta_{k}(d)$, where $\phi_{k\pk}(\cdot)$ are basis functions and $\delta_{k}(d)$ are suitably decaying approximation errors and the series lengths $\ell_{k}$ diverge to infinity as $(n,T)\to\infty$. We then denote the sieve approximation and relevant matrices as $\xi_{k}(d)=\sum_{\pk=1}^{\ell_{k}}\lambda_{\pk}\phi_{k\pk}(d)$, $\Xi_{k}=\mat{\xi_{k}(d_{ij})}$, and $\varPhi_{k\pk}=\mat{\phi_{k\pk}(d_{ij})}$. By combining \eqref{eq:matrix function} with these approximations and defining $\Delta_{k}=\mat{\delta_{k}(d_{ij})}$, we obtain the decomposition $B_{k}=S_{k}+R_{k}$, where $S_{k}$ depends on the sieve matrix $\Xi_{k}$ and $R_{k}$ collects the approximation error matrix. Specifically, for the SAR, $S_{1}=I_{n}-\Xi_{1}, S_{2}=\Xi_{2}, S_{3}=I_{n}-\Xi_{3}$ with $R_{1}=-\Delta_{1}, R_{2}=\Delta_{2}, R_{3}=-\Delta_{3}$; for the MESS, $S_{k}=e^{\Xi_{k}}$ and $R_{k}=S_{k}(e^{\Delta_{k}}-I_{n})$.
	\par 
	Substituting these into \eqref{eq:sdpd_np}, the relationship between $Y_{t}$ and $E_{t}$ can be reformulated as:
	\begin{equation}\label{eq:mess_appro_Et}
		S_{3}S_{1}Y_{t}=S_{3}\bigl((\gamma I_{n}+S_{2})Y_{t-1}+X_{t}\beta+\mathbf{c}_{n}+\alpha_{t}l_{n}\bigr)+V_t,
	\end{equation}
	where the error term is $V_{t}=r_{t}+E_{t}$ with $r_{t}=\sum_{k=1}^3 r_{kt}$, and $r_{1t}=-S_{3}R_{1}Y_{t}$, $r_{2t}=S_{3}R_{2}Y_{t-1}$, and $r_{3t}=-R_{3}U_{t}$. These components capture errors arising from contemporaneous outcome interactions, lagged outcome interactions, and disturbance feedback, respectively. 
	%While $r_{1t}$ is discussed in \citet{chen2025npspatial} without $S_3$, the components $r_{2t}$ and $r_{3t}$ are typically not treated explicitly in existing studies.
	\begin{remark}
		Our setup formally covers `higher-order' SAR models, see e.g. \cite{lee2010efficient}. These would simply have approximation error equal to zero. Our focus though is on the truly nonparametric case with a non-zero approximation error.
	\end{remark}
	\section{GMM estimation}\label{sec:gmm}
	To eliminate individual FE, we employ the forward orthogonal difference (FOD) transformation \citep{arellano1995another, alvarez2003time, lee2014efficient}, rather than the first difference (FD) approach used in studies such as \citet{sun2018estimation}. 
	Specifically, for any $n \times T$ matrix of variables $[K_1, \dots, K_T]$, we obtain the $n\times (T-1)$ transformed matrix $[K_1^*, \dots, K_{T-1}^*]$ by applying an orthogonal transformation matrix $F_{T,T-1}$, where the algebraic details are in the appendix. 
	The merit of the FOD transformation is that it yields uncorrelated transformed idiosyncratic terms $E_t^*$, preserving their serial uncorrelatedness, whereas FD induces an MA(1) process even when the original errors are serially uncorrelated.% 
	\footnote{In our semi-nonparametric setting, the transformed error term is $V_t^{*} = E_t^{*} + r_t^{*}$, where $r_t^{*}$ represents the transformed sieve approximation error. As shown in Lemma \ref{lem:Rr} in the supplement, under certain rate conditions, $r_t^{*}$ is asymptotically negligible.}
	Time FE can be eliminated by $J_{n}=I_{n}-\frac{1}{n}l_nl_n'$, recalling that $\alpha_{t}J_{n}l_{n}=0$.
	After these transformations, the reformulated model \eqref{eq:mess_appro_Et} for $t=1,\dots,T-1$ becomes:%
	\footnote{We approximate the unknown weight functions by a sieve space of dimension $\elln=\max_{1\leq k\leq 3}\ell_{k}$, where $\ell_{k}\to\infty$ but grows sufficiently slowly relative to $n,T$. This contrasts with \citet{sun2018estimation}, who nonparametrically estimate coefficient functions of a $T$-dimensional vector of covariates. Hence, the first issue in \citet[Section~5.1]{sun2018estimation} does not arise in our setting.}
	\begin{flalign}\label{eq:mess nofe}
		J_{n}S_{3}S_{1}Y_{t}^{*}=\gamma J_{n}S_{3}Y_{t-1}^{(*,-1)}+J_{n}S_{3}S_{2}Y_{t-1}^{*}+J_{n}S_{3}X_{t}^{*}\beta +J_{n}V_{t}^{*},
	\end{flalign}
	where $Y_{t}^{*}$ and $Y_{t-1}^{(*,-1)}$ represent the FOD-transformed contemporaneous and lagged outcome vectors, respectively. The transformed regressor matrix $X_{t}^{*}$ and the error term $V_{t}^{*}$ are defined analogously. These variables are constructed such that the transformation for a period $t$ depends only on current and future observations, and the algebraic construction of these transformation matrices is shown in Section \ref{sec:notation} of the appendix. We now introduce some useful notation for the estimates. Denote $N=n(T-1)$ and $\theta=(\pi',\lambda')'$, where $\pi=(\gamma,\beta')'$ and $\lambda=(\lambda_{1}^{\prime},\lambda_{2}',\lambda_{3}')'$ with $\lambda_{k}=(\lambda_{k1},...,\lambda_{k\ell_{k}})'$ for $k=1,2,3$, and $\dim(\theta)=\ell_{\theta}$. Furthermore, $\hat{g}_{k}(d)=\bphi_{k}'(d)\hat\lambda_{k}$, where $\bphi_{k}(d)=[\phi_{k1}(d),...,\phi_{k\ell_{k}}(d)]$ and $\hat\lambda_{k}$ is a generic estimator of $\lambda_{k}$.
	% Furthermore, let $g(d)=(g_{1}(d),g_{2}(d),g_{3}(d))'$. 

	\subsection{Moment conditions}
	For the estimation of equation \eqref{eq:mess nofe}, we define $\bQ_{N}=\Diag_{t=1}^{T-1}(Q_{t})$ as the block-diagonal instrument variables (IV) matrix, where $Q_{t}$ is an $n\times\lqn$ matrix of instrumental variables. A typical construction of $Q_t$ includes $Y_{t-1}$, $X_{t}^*$, and their interactions with the sieve basis functions $\varPhi_{k\pk}$. 
	For example,
	\begin{flalign*}
		Q_{t}=\bigl(Y_{t-1}, X_{t}^{*},[\varPhi_{21},\ldots,\varPhi_{2\lwyl},\varPhi_{31}\varPhi_{21},\ldots,\varPhi_{3\lwu}\varPhi_{2\lwyl}]Y_{t-1},
		[\varPhi_{11},..., \varPhi_{1\lwy}, \varPhi_{31}\varPhi_{11},
		\ldots,\varPhi_{3\lwu}\varPhi_{1\lwy}]X_{t}^{*}\bigr).
	\end{flalign*}
	We that $\lqn\sim \elln$, where the symbol `$\sim$' denotes the same rate of growth, and $\elln=\max_{1\leq k\leq 3}\ell_k$. The linear moment condition is 
	$
	m^{\lin}_{N}(\theta)=\bQ_{N}'\bJ_{N}\mathbf{V}_{N}^{*}(\theta),
	$
	where $\bJ_{N}=I_{T-1}\otimes J_{n}$, $\mathbf{V}_{N}^{*}(\theta)=\left(\bV_{1}^{*\prime}(\theta),\bV_{2}^{*\prime}(\theta),...,\bV_{T-1}^{*\prime}(\theta)\right)'$ is the vector of transformed residuals and $\otimes$ denotes the Kronecker product. Note that $Y_{t-1}$ serves as an asymptotically valid instrument for $Y_{t-1}^{(*,-1)}$ since, as can be derived from Lemma \ref{lem:VBV}(\rn4) in the supplement, the terms involving approximation errors $Y_{t-1}'r_t^{*}$ are asymptotically negligible. 
	While theoretically valid, the potential for finite sample bias is explored further in our Monte Carlo simulations.
	%%%%%%%%%%%%%%%%%%%%%%%%%%%%%%%%%%%%%%%%%%%%
	\subsection*{An initial estimator}
	Motivated by \citet{pinkse2002spatial}, \citet{sun2016functional} and \citet{chen2025npspatial}, we consider the two-stage least square estimator (2SLSE) derived from the linear moment $m^{\lin}_{N}(\theta)$, where
	%\begin{flalign*}%\label{eq:nls}
	$
	\hat\theta_{ts}=\arg\min_{\theta\in\Theta}m^{\lin\prime}_{N}(\theta)\bm{M}_{Q}m^{\lin}_{N}(\theta)
	\ \text{with} \
	\bm{M}_{Q}=\bJ_{N}\bQ_{N}(\bQ_{N}'\bJ_{N}\bQ_{N})^{-1}\bQ_{N}'\bJ_{N}.
	$
	%\end{flalign*}
	% Consequently, the sieve estimation for $g_{k}(d)$ becomes 
	% %\begin{flalign*}%\label{eq:phik}
	% $
	% \hat{g}_{k,nls}(d)=\hat{\lambda}_{k,nls}'\bphi_{k}(d).
	% $
	%\end{flalign*}
	\subsection{Optimal GMM estimation}
	We now consider GMM estimation, for which we can construct additional quadratic moment conditions. 
	We assume that the variance structure of $E_{t}$, denoted by $\Sigma_{t}$, has three cases:
	\labeltext{(V0)}\label{var:homo}: $\sigma^2 I_{n}$, \labeltext{(V1)}\label{var:hetei}: $\Diag_{i=1}^{n}(\sigma_{i}^2)$, and
	\labeltext{(V2)}\label{var:hetet}: $\sigma_{t}^2I_{n}$. There is a sense in the extant literature that combining both the linear and quadratic moments could improve GMM estimation efficiency. For related parametric SDPD models, see e.g. \citet{lee2014efficient} and \cite{kuersteiner2020dynamic}. For related nonparametric spatial panel data models, see e.g. \citet{sun2018estimation} and \citet{yang2025estimation}.
	Specifically, \sloppy
	%\begin{flalign}%\label{eq:moment_quad}
	$m_{N}^{\texttt{quad}}(\theta)=\Big(\bV_{N}^{*\prime}(\theta)\bJ_{N}\bP_{N1}'\bJ_{N}\bV_{N}^{*}(\theta),...,\bV_{N}^{*\prime}(\theta)\bJ_{N}\bP_{N\lpn}'\bJ_{N}\bV_{N}^{*}(\theta)\Big)'.$
	%\end{flalign}
	Here, $\bP_{Nk}=\Diag_{t=1}^{T-1}(P_{kt})$, $P_{kt}$ is some $n\times n$ square matrix satisfying $\diagM(J_{n}P_{kt}J_{n})=\bzero_{n\times 1}$. Lemma \ref{lem:transJ}(\rn1) in the appendix provides a way of choosing $P_{kt}$.
	\par 
	Then, the stacked moment condition combining linear and quadratic moments is
	\begin{flalign}\label{eq:optimal_moments}
		m_{N}(\theta)=\invN(m_{N}^{\texttt{quad}\prime}(\theta),m_{N}^{\texttt{line}\prime}(\theta))'.    
	\end{flalign}
	According to \citet{hansen1982large}, the optimal weights of the moments can be approximated by
	\begin{flalign}\label{eq:Var_moment}
		\Omega_{N}=\invN\Diag\Big(\Varm_{N1},\Varm_{N2}\Big),
	\end{flalign}
	% by Lemma \ref{lem:VBV} in the supplement,
	where 
	$[\Varm_{N1}]_{ij}=\tr(\Sigep_{N}\bJ_{N}\bP_{Ni}\bJ_{N}\Sigep_{N}\bJ_{N}\bP_{Nj}^{s}\bJ_{N})$, for $i,j=1,...,\lpn$,
	$\Varm_{N2}=\bQ_{N}'\bJ_{N}\Sigep_{N}\bJ_{N}\bQ_{N}$, $\Sigep_{N}=\Diag_{t=1}^{T-1}(\Sigma_{t})$, and $\bP_{Nj}^{s}=\bP_{Nj}+\bP_{Nj}'$.
	The feasible optimal GMME (OGMME) is 
	%\[
	$\hat{\theta}_{ogmm}=\arg\min_{\theta\in\Theta}m_{N}(\theta)\tilde{\Omega}_{N}^{-1}m_{N}(\theta),
	$
	%\]
	where $\tilde{\Omega}_{N}$ is a consistent estimator of $\Omega_{N}$.
	%%%%%%%%%%%%%%%%%%%%%%%%%%%%%%%%%%%%%%%%%%%%%%%%
	%%%%%%%%%%%%%%%%%%%%%%%%%%%%%%%%%%%%%%%%%%%%%%%%%%%%%%%%
	%%%%%%%%%%%%%%%%%    Best GMM    %%%%%%%%%%%%%%%%%%%%%%%
	%%%%%%%%%%%%%%%%%%%%%%%%%%%%%%%%%%%%%%%%%%%%%%%%%%%%%%%%
	\subsection{Best GMM estimation}
	In this section, we consider the case in which the OGMME attains the smallest asymptotic variance, referred to as the BGMME \citep{lee2010efficient,lee2014efficient}. 
	To the best of our knowledge, there is no such work on the semi-nonparametric SDPD model.
	For the parametric model, \cite{lee2014efficient} considers only the homoskedastic case, which cannot be directly extended to the heteroskedastic setting. Furthermore, when heteroskedasticity is unknown, we must impose additional structure on $\Sigma_{t}$ to construct a consistent estimator $\hat{\Sigma}_{t}$, as we do in our cases \ref{var:hetei} and \ref{var:hetet}.%
	\footnote{It is difficult to treat the fully two-dimensional unknown heteroskedastic case $\sigma_{it}^2$ because we cannot obtain feasible best moment conditions to construct the BGMME. A common way to restore feasibility is to impose a separable structure, for example $\sigma_{it}^2 = \sigma_{i}^2 \sigma_{t}^2$ or $\sigma_{it}^2 = \sigma_{i}^2 + \sigma_{t}^2$, which reduces the number of variance parameters and makes consistent estimation possible. Similar discussions can be found in \citet{chen2025spatial}.} To eliminate time FE under unknown heteroskedasticity and to derive the best moment conditions, we introduce the following $n \times n$ projection matrix:
	\begin{flalign}\label{eq:MSL}
		\MSL{t}=I_{n}-\Sigma_{t}^{-\half}l_{n}\bigl(l_{n}'\Sigma_{t}^{-1}l_{n}\bigr)^{-1}l_{n}'\Sigma_{t}^{-\half},
	\end{flalign}
	which is idempotent and satisfies $\MSL{t}\Sigma_{t}^{-\half}\,l_{n} = 0$. As discussed in Section \ref{sec:estimate_variance} of the supplement, the projection matrix $\MSL{t}$ is well-behaved.
	Thus, we can employ $\MSL{t}$ to eliminate the time FE even when $\Sigma_{t}$ is heteroskedastic and unknown. We then pre-multiply $\MSL{t}\Sigma_{t}^{-\half}$ and employ the FOD transformation on both sides of \eqref{eq:mess_appro_Et} to obtain the transformed model:
	\begin{align}
		\MSL{t}\Sigma_{t}^{-\half}S_{3}S_{1}Y_{t}^{*}
		&=\gamma \MSL{t}\Sigma_{t}^{-\half}S_{3}Y_{t-1}^{(*,-1)}
		+\MSL{t}\Sigma_{t}^{-\half}S_{3}S_{2}Y_{t-1}^{*}
		+\MSL{t}\Sigma_{t}^{-\half}S_{3}X_{t}^{*}\beta
		\notag
		\\
		&\quad+\MSL{t}\Sigma_{t}^{-\half}V_{t}^{*},
		\qquad t=1,\dots,T-1,
		\label{eq:trans_gmm}
	\end{align}
	with $V_{t}^{*}=r_{t}^{*}+E_{t}^{*}$. Furthermore, denote
	$\JSL{t}=\Sigma_{t}^{-\half}\MSL{t}\Sigma_{t}^{-\half},$
	and thus, the linear and quadratic moment conditions for equation \eqref{eq:trans_gmm} are now
	\begin{align}\notag
		&\invN\mNL(\theta)=\invN\bQ_{N}^{\prime}\Sigep_{N}^{-\half}\MSL{N}\Sigep_{N}^{-\half}\bVL_{N}^{*}(\theta)=\invN\bQ_{N}^{\prime}\JSL{N}\bVL_{N}^{*}(\theta),\\ \label{eq:best_moments}
		%\intertext{\centerline{\text{and}}}
		&\invN\mNQk(\theta)= \invN\bVL_{N}^{*\prime}(\theta)\JSL{N}\bP_{Nj}\JSL{N}\bVL_{N}^{*}(\theta),
	\end{align}
	for $j=1,...,\lpn$,
	where $\MSL{N}=\Diag_{t=1}^{T-1}\MSL{t}$ and  $\JSL{N}=\Diag_{t=1}^{T-1}\JSL{t}$. Therefore, we use $J(\Sigma_{t})$ in BGMM, and $J_{n}$ in OGMM and 2SLS.%
	\footnote{When the variance of $\epsilon_{it}$ is cross-sectionally homogeneous, i.e., $\sigma_{it}^2=\sigma_{t}^2$ or $\sigma^2$, we have $\MSL{t}=J_{n}$ and $\JSL{t}=\sigma_{t}^{-1}J_{n}$ or $\JSL{t}=\sigma^{-1}J_{n}$. Also, note that $\JSL{t}\Sigma_{t}\JSL{t}=\Sigma_{t}^{-1/2}\MSL{t}\Sigma_{t}^{-1/2}\Sigma_{t}\Sigma_{t}^{-1/2}\MSL{t}\Sigma_{t}^{-1/2}=\Sigma_{t}^{-1/2}\MSL{t}\Sigma_{t}^{-1/2}=\JSL{t}$.}%
	\par 
	After some algebra, for $t=1,...,T-1,$ we have
	\begin{flalign}\label{eq:best Qt}
		%	\begin{split}
			Q_{t}=S_{3}\left[[\tfrac{\partial\bS_{1N}}{\partial\lambda_{11}},\cdots,\tfrac{\partial\bS_{1N}}{\partial\lambda_{1\lwy}}]S_{1}^{-1}\mathbb{W}_{t}, [\tfrac{\partial\bS_{2N}}{\partial\lambda_{21}},\cdots,\tfrac{\partial\bS_{2N}}{\partial\lambda_{2\lwyl}}]\bestY_{t},\bestY_{t},X_{t}^{*}\right],
			%	\end{split}
	\end{flalign}
	where $\bestY_{t}$ is defined in \eqref{eq:bestylag} in the appendix, and $\mathbb{W}_{t}=(\gamma I_{n}+S_{2})\bestY_{t}+X_{t}^{*}\beta+\alpha_{t}^{*}l_{n}$.
	As for the quadratic moments, 
	\begin{flalign}\label{eq:best Pt}
		\begin{split}
			& P_{t}=[P_{1t},\cdots,P_{\lwy t},P_{\lwy+1,t},\cdots,P_{\lwu t}],
			\ \text{where} \
			P_{\kwy,t}=\left[S_{3}\tfrac{\partial S_{1}}{\partial\lambda_{1\kwy}}S_{1}^{-1}S_{3}^{-1}\Sigma_{t}\right]^{\diamond} 
			\ \text{and} \ 
			\\& 
			P_{\lwy+\kwu,t}=\left[\tfrac{\partial S_{3}}{\partial\lambda_{3\kwu}}S_{3}^{-1}\Sigma_{t}\right]^{\diamond},
		\end{split}
	\end{flalign}
	for $\kwy=1,...,\lwy$ and $\kwu=1,...,\lwu$, where 
	`$^{\diamond}$' denotes the diagonal-adjustment operator in Lemma \ref{lem:transJ}(\rn2) in the appendix: For any $n\times n$ time-varying matrix $H_t$, it modifies only the diagonal entries so that the transformed matrix has a zero diagonal, i.e., $\diagM(\JSL{t}H_t^{\diamond}\JSL{t})=\bzero_{n\times 1}$ by construction. 
	Moreover, now
	$[\Varm_{N1}]_{ij}=\tr(\JSL{N}\bP_{Ni}\JSL{N}\bP_{Nj}^{s}\JSL{N})$ for $i,j=1,...,\ell_1+\ell_3$, 
	and
	$\Varm_{N2}=\bQ_{N}'\JSL{N}\bQ_{N}$ by using $\JSL{N}\Sigep_{N}\JSL{N}=\JSL{N}$.
	We can then give the algorithm to obtain the proposed estimators:
	% \newline~\newline
	\\[15pt]
	\rule{0pt}{0.1pt}\textbf{\large Algorithm 1: Estimation Procedure}%\newline\rule{\textwidth}{1pt}\\
	\hrule\noindent
	\textbf{Step 1:} Obtain the initial 2SLSE
	$\hat\theta_{ts}=\min_{\theta\in\Theta}m^{\lin\prime}_{N}(\theta)\bm{M}_{Q}m^{\lin}_{N}(\theta)$ using linear moments $m^{\lin}_{N}(\theta)=\bQ_{N}'\bJ_{N}\mathbf{V}_{N}^{*}(\theta)$. Evaluate $\hat{g}_{k,ts}(d)=\bphi_{k}'(d)\hat\lambda_{k,ts}$ for $k=1,2,3$. Compute the residuals $\tilde{V}_{t}=J_{n}\bigl(V_t(\hat\theta_{ts})-\frac{1}{T}\sum_{h=1}^{T}V_h(\hat\theta_{ts})\bigr)$ and $\tilde{\Sigma}_{t}=\Sigma_{t}(\hat\theta_{ts})$ based on $\hat\theta_{ts}$.
	\\
	\textbf{Step 2:} Given $\tilde{\Sigma}_{t}$, construct $\tilde{\Omega}_{N}$ based on equation \eqref{eq:Var_moment}. Obtain the feasible OGMME 
	$\hat{\theta}_{ogmm}=\arg\min_{\theta\in\Theta}m_{N}'(\theta)\tilde{\Omega}_{N}^{-1}m_{N}(\theta)$ where $m_{N}(\theta)$ is given in \eqref{eq:optimal_moments}. Evaluate $\hat{g}_{k,ogmm}(d)=\bphi_{k}'(d)\hat\lambda_{k,ogmm}$.
	Update the residuals $\tilde{V}_{t}=J_{n}\bigl(V_t(\hat\theta_{ogmm})-\frac{1}{T}\sum_{h=1}^{T}V_h(\hat\theta_{ogmm})\bigr)$, and $\tilde{\Sigma}_{t}=\Sigma_{t}(\hat\theta_{ogmm})$.
	\\
	\textbf{Step 3:} Given $\hat\theta_{ogmm}$ and $\tilde{\Sigma}_{t}$, estimate $\MSLhat{t}$, $\QLhat{t}$ and $\PLhat{t}$ from equations \eqref{eq:MSL}, \eqref{eq:best Qt} and 
	\eqref{eq:best Pt}, respectively. Estimate $\JSLhat{t}=\tilde{\Sigma}_{t}^{-1/2}\MSLhat{t}\tilde{\Sigma}_{t}^{-1/2}$. Construct $\tilde{\Omega}_{N}$ based on equation \eqref{eq:Var_moment}, with $\bJ_{N}$ replaced by $\JSLhat{N}$. Obtain the feasible BGMME 
	$\hat{\theta}_{bgmm}=\arg\min_{\theta\in\Theta}m_{N}'(\theta)\tilde{\Omega}_{N}^{-1}m_{N}(\theta)$, where $m_{N}(\theta)$ is defined as \eqref{eq:best_moments}. Evaluate $\hat{g}_{k,bgmm}(d)=\bphi_{k}'(d)\hat\lambda_{k,bgmm}$.
	% \\[-2ex]
	% \rule{\textwidth}{1pt} \indent
	\hrule\medskip
	
	\section{Asymptotic properties}\label{sec:asy}
	We collect the full set of regularity assumptions in Section \ref{sec:assumption} in the appendix. In this section, we highlight two conditions that drive the asymptotics: the joint growth requirement on $(n,T)$ and the sieve dimension requirement on $\elln$.
	\begin{assumption}[Sample size]\label{ass:nT}
		As $(n,T)\to\infty$, $n=o\left(T^{p/2}\right)$ for some $p>2$.
	\end{assumption}
	\begin{assumption}[Sieve function]\label{ass:sieve-basis}
		(\rn1) The approximation errors satisfy $\delta_{k}(d)=\Op(\ell_{k}^{-\varsigma_{k}})$, $k=1,2,3$. Let $\elln=\max_{1\leq k\leq 3}\ell_{k}$ denote the sieve dimension, with $\elln\to\infty$ as $(n,T)\to\infty$. Moreover, $\sqrt{n(T-1)}\elln^{-\vsinf}+\sqrt{\frac{\elln}{n(T-1)}}\to 0$ as $(n,T)\to\infty$, where $\vsinf>2$ with $\vsinf=\min_{1\leq k\leq 3}\varsigma_{k}$.
		\\
		(\rn2) The basis functions $\mat{\phi_{kl}(\cdot), k=1,2,3,\ l=1,\cdots,\elln}$ are uniformly bounded (UB) on the compact domain $D$, and
		$\bphi_k(d)=[\phi_{k1}(d),...,\phi_{k\ell_{k}}(d)]$ satisfies
		$\max\limits_{1\leq k\leq 3}\sup\limits_{d\in D}\spnorm{\bphi_{k}(d)}\leq C_{\ell}\sqrt{\ell_{n}}$ for some constant $C_{\ell}$.
	\end{assumption}
	Assumption \ref{ass:nT} allows a larger growth rate of $n$ compared to $T$. This is a mild assumption as discussed in \citet[Assumption~3.6]{su2023identifying}.
	%This assumption can also help us control the growth rate of sieve expansions. 
	In Assumption \ref{ass:sieve-basis}, (\rn1) ensures that the sieve approximation error is asymptotically negligible, and, moreover, that the asymptotic correlation between $\bV_{N}^{*}$ and $\bX_{N}^{*}$ vanishes.\footnote{In the cross-sectional setting for each $t$, we also need $\sqrt{n}\elln^{-\vsinf}\to 0$, which is implied by $\sqrt{n(T-1)}\elln^{-\vsinf}\to 0$ as $(n,T)\to\infty$.} 
	Condition (\rn2) imposes growth restrictions on the series expansion terms.
	\subsection{Properties of estimators of the finite-dimensional parameter}{\label{subsec:fd_inf}}
	For $k=1,2,3$, 
	write $\hat{\pi}_{gmm}$ for either $\hat{\pi}_{ogmm}$ or $\hat{\pi}_{bgmm}$, and the true estimand is $\pi_0$.
	% , and 
	% $\hat{g}_{k,gmm}$ for either $\hat{g}_{k,gmm}$ or $\hat{g}_{k,bgmm}$. 
	Then, we have the following theorem.
	
	\begin{theorem}\label{thm:gm_pi}
		Suppose that $\tilde{\Omega}_{N}$ is evaluated by the initial consistent estimator $\tilde{\theta}$ based on equation \eqref{eq:Var_moment}:\\ 
		(\rn1) Under Assumptions \ref{ass:nT}-\ref{ass:sieve-basis}, and \ref{ass:disterbance}-\ref{ass:id} in the appendix, when $\elln^{3/2-\vsinf}+\frac{\elln^{3/2}}{n}\to 0$ as $(n,T)\to\infty$:%\frac{\elln^{3/2}\lnT}{n(T-1)}
		\[
		\hat\pi_{gmm}-\pi_0=\op(1).
		\]
		(\rn2) Under Assumptions \ref{ass:nT}-\ref{ass:sieve-basis}, and \ref{ass:disterbance}-\ref{ass:id} in the appendix, when $\sqrt{\frac{(T-1)\elln^{3}}{n}}+\sqrt{n(T-1)}\elln^{3/2-\vsinf}\to 0$ as $(n,T)\to\infty$:%\frac{\elln^{3/2}\lnT}{\sqrt{n(T-1)}}
		\[
		\sqrt{n(T-1)}(\hat\pi_{gmm}-\pi_0)\dlim N(0,\bm{\Sigma}_{\pi_0,gmm}),
		\]
		where $\Sigep_{\pi_0,gmm}=\lim_{n,T\to\infty}K_{\pi}\Omega_{N}K_{\pi}'=O(1)$ and $K_{\pi}$ is defined in \eqref{eq:H_pi} in the appendix.
	\end{theorem}
	
	The rate condition in Theorem \ref{thm:gm_pi}(ii) implies the one in (i). Indeed,
	$\sqrt{n(T-1)}\elln^{3/2-\vsinf} \to 0\Rightarrow 
	\elln^{3/2-\vsinf} \to 0,$ and
	$\sqrt{\frac{(T-1)\elln^{3}}{n}}=\frac{\elln^{3/2}}{n} \sqrt{n(T-1)}\to 0\Rightarrow \frac{\elln^{3/2}}{n}\to 0$.
	\par 
	We next establish the consistency of the covariance estimator $\bm\Sigma_{\pi_0,gmm}$ and characterize the asymptotic efficiency of the BGMME.
	\begin{theorem}\label{thm:Var_pi}
		% Let $\hat\theta$ be either $\hat\theta_{ogmm}$ or $\hat\theta_{bgmm}$. 
		Define
		$\hat{\bm\Sigma}_{\pi,gmm}= \hat{K}_{\pi}\hat{\Omega}_{N}\hat{K}_{\pi}^{\prime}$,
		where $\hat{K}_{\pi}$ and $\hat\Omega_{N}$ are the feasible counterparts of
		$K_\pi$ and $\Omega_{N}$ evaluated at $\hat\pi_{gmm}$.
		Under the conditions of Theorem \ref{thm:gm_pi}(\rn2), we have:
		\\
		(\rn1)
		$
		\spnorm{\hat{\bm\Sigma}_{\pi,gmm}-\bm\Sigma_{\pi_0,gmm}}=\op(1).
		$
		\\
		(\rn2) The asymptotic variance of the BGMME attains the efficiency lower bound $\Sigep_{b_\pi}$ by the generalized Cauchy-Schwarz inequality on $\lim_{n,T\to\infty}K_{\pi}\Omega_{N}K_{\pi}$, where $\Sigep_{b_\pi}$ is defined in equation \eqref{eq:Var_bg_pi} in the appendix. Accordingly, it can be consistently estimated by $\hat\Sigep_{b_\pi}$.
	\end{theorem}
	\subsection{Properties of estimators of the infinite-dimensional parameter}{\label{subsec:infd_inf}}
	Having established the asymptotic properties of the finite-dimensional parameters, we now turn to the inference on the unknown spatial weight functions $g_{k}(\cdot)$. The following theorem establishes the uniform consistency and asymptotic normality of the sieve estimators $\hat{g}_{k,gmm}(\cdot)$.
	\begin{theorem}\label{thm:gm_gk}
		Let $\hat{g}_{k,gmm}$ be either $\hat{g}_{k,gmm}$ or $\hat{g}_{k,bgmm}$, and let the true $g_{k}$ be $g_{k0}$, $k=1,2,3$. Let $\hat{g}_{k,gmm}(d)=\bphi_{k}'(d)\hat\lambda_{k,gmm}$. Then:\\
		(\rn1) Under Assumptions \ref{ass:nT}-\ref{ass:sieve-basis}, and \ref{ass:disterbance}-\ref{ass:id} in the appendix, when  $\elln^{2-\vsinf}+\frac{\elln^{2}}{n}+\frac{\elln}{\sqrt{n(T-1)}}\to 0$ as $(n,T)\to\infty$:%\frac{\elln^{2}\lnT}{n(T-1)}
		\[
		\sup_{d\in D}|\hat{g}_{k,gmm}(d)-g_{k0}(d)|=\op(1).
		\]
		(\rn2) Under Assumptions \ref{ass:nT}-\ref{ass:sieve-basis}, and \ref{ass:disterbance}-\ref{ass:id} in the appendix, when  $\sqrt{\frac{(T-1)\elln^{3}}{n}}+\sqrt{n(T-1)}\elln^{3/2-\vsinf}\to 0$ as $(n,T)\to\infty$: 
		\[
		\sqrt{\frac{n(T-1)}{\elln}}(\hat{g}_{k,gmm}(d)-g_{k0}(d))\stackrel{d}{\to}N(0,\Sigep_{g_{k0},gmm}),
		\]
		where $\Sigep_{g_{k0},gmm}=O(1)$ is defined in equation \eqref{eq:Var_g} in the appendix.
	\end{theorem}
	The asymptotic results in Theorem \ref{thm:gm_gk} highlight two key distinctions between the nonparametric sieve estimator and the parametric case in Theorem \ref{thm:gm_pi}. First, the rate conditions required to establish uniform consistency are slightly stronger than those for $\hat{\pi}$. Specifically, the condition involves terms of order $n^{-1}\elln^{2}$ rather than $n^{-1}\elln^{3/2}$. This stricter requirement reflects the difficulty of achieving uniform convergence of the function $\hat{g}_k(d)$ over the domain $d\in D$, as opposed to the standard convergence of a finite parameter vector. Second, the convergence rate for asymptotic normality is slower. While the finite-dimensional estimators achieve the standard $\sqrt{n(T-1)}$ rate, the nonparametric estimators converge at the slower rate $\sqrt{n(T-1)/\elln}$.%
	\footnote{\label{foot:rate} For illustration, let $T\to\infty$ and set $n=O(T^{13/5})$ and $\elln=O(n^{1/5})$.
		With $\vsinf=p=6$, $\frac{n}{T^{p/2}}=O(T^{-2/5})$ so Assumption \ref{ass:nT} holds. Note that $\elln^{2-\vsinf}\to 0$ as $\vsinf>2$, $\frac{\elln^{2}}{n}=O(n^{-3/5})=O(T^{-39/25})$, $\frac{\elln}{\sqrt{n(T-1)}}=O(T^{-32/25})$, $\sqrt{\frac{(T-1)\elln^{3}}{n}}=O(T^{-1/50})$ and $\sqrt{n(T-1)}\elln^{3/2-\vsinf}=O(T^{-27/50})$.
		Thus, this choice of $(n,T,\elln)$ satisfies all necessary rate conditions.}
	\begin{remark}\label{remark:finite T}
		(\rn1) While our model assumes $(n,T) \to \infty$, the proposed estimators are also applicable in panels with large $n$ and finite $T$. Specifically, the 2SLSE and OGMME remain consistent under large $n$ and fixed $T$ for the variance structures specified in \ref{var:homo} and \ref{var:hetet}. However, the implementation of the BGMME requires $(n,T)\to\infty$. This is because the individual FE, which are required to construct $\bestY_{t}$ as defined in \eqref{eq:bestylag} in the appendix, cannot be consistently estimated with finite $T$, and thus the best IVs might not be available (see \citealp[footnote 24]{lee2014efficient}).
	\end{remark}
	\section{Monte Carlo simulations}\label{sec:mc}
	We conduct Monte Carlo experiments to evaluate the finite sample performance of our proposed estimators. To conserve space, we report the key baseline results for the MESS specification in the main text, while the results for the SAR specifications and additional results are provided in the supplement.
	\par
	For the DGP in equation \eqref{eq:sdpd_np}, $X_{t}$, $\bc_{n}$, $\alpha_{t}$, and $E_{t}$ are generated from independent standard normal distributions. We consider $\pi_0=(\gamma_0,\beta_0)'=(-0.7,1)'$ for the MESS and $\pi_0=(0.3,1)'$ for the SAR. We set $n\in\mat{100, 200, 400}$ and $T\in\mat{10, 25}$, generate the spatial panel data with $500+T$ periods, and then use the last $T$ periods as our sample. The initial value is generated as $N(0, I_n)$, and we consider the case of \ref{var:hetei} with 1,000 replications. 
	\par 
	To generate the unknown $G_{k}$, we randomly generate the distances $d_{ij}$ between individuals and set $G_{k}=G$, constructed as $g^{*}(d_{ij})=\texttt{Normal}(-d_{ij})$ if $i \neq j$, and $g_{ii}=0$, where $\texttt{Normal}(\cdot)$ is the standard normal cdf. Here, $d_{ij}$ is the Euclidean distance between the coordinates of units $i$ and $j$, and we set $d_{ij}=0$ if $d_{ij}$ exceeds a threshold $\bar{d}_{0}$, where $\bar{d}_{0}$ is the 10th percentile of all $d_{ij}$. The two-dimensional coordinates are generated from a uniform distribution $U[0,1]$.
	% To generate the unknown $G_{k}$, we randomly generate the distances $d_{ij}$ between individuals and set $G_{k}=G$, constructed as  
	% $g^{*}(d_{ij})=\texttt{Normal}(-d_{ij})\cI{d_{ij}\leq\text{10th percentile of all $d_{ij}$}}$ if $i \neq j$, and $g_{ii}=0$, where $\texttt{Normal}(\cdot)$ is the standard normal cdf, and 
	% $d_{ij}$ is the Euclidean distance between the coordinates of the units $i$ and $j$, where the two-dimensional coordinates are generated from $U[0,1]$.
	Then, $G=G^{*}/(1.2\spnorm{G^{*}})$. From the construction, $\rho(A)=0.6949$. 
	In estimation, we use the row-normalization version of $\varPhi_{kl}=\varPhi_{l}=\mat{\phi_{l}(d_{ij})}$ for $l=1,...,\elln$.
	We consider two different choices for the sieve dimension: $\ell_{k}=\elln=2$, and $\ell_{k}=\elln=[n^{1/5}]+2$, a choice motivated by the theoretical rate conditions derived in Section \ref{sec:asy} and Footnote \ref{foot:rate}.
	\par 
	To generate the heteroskedastic variances $\sigma_{\epsilon_i}^2$, we proceed in three steps. First, partition the $n$ individuals into $K$ groups and assign initial group labels $k_i\in\{1,\dots,K\}$ accordingly; then randomly permute these labels to obtain $\tilde k_i$, so that group sizes remain fixed but membership is randomized. Second, define 
	$u_h = u(h)=(h-1)/(K-1), h=1,\dots,K,$
	and interpret $u_{\tilde{h}_i}=u(\tilde{h}_i)$ as the value associated with the randomized group of observation $i$. Finally, set
	$\sigma_{\epsilon_i}^2 = \bigl(1 + {u_{\tilde{h}_i}}/{K}\bigr)^2$.
	In our design, we set $K=3$. Thus, $u=\mat{0,{1}/{2},1}$. 
	% If an observation is assigned to group $\tilde{h}_i=2$ after random permutation, then $u_{\tilde{h}_i}=1/2$ and $\sigma_{\epsilon_i}^2=\bigl(1+{1/6}\bigr)^2$.
	\par 
	For the estimation of $\pi$, Table \ref{sim:hei} shows the following: First, as $n$, $T$, and $\elln$ increase, both OGMME and BGMME exhibit lower root mean squared error (RMSE) and smaller bias compared with the 2SLSE. This indicates that the GMME can improve the finite sample performance of the 2SLSE, especially when $n$, $T$, and $\elln$ are large. Second, within the GMME framework, BGMME is more efficient than OGMME when $n$ and $T$ increase and $\elln/(nT)\to 0$, as its empirical standard deviation (ESD) and RMSE are smaller. However, this efficiency gain is less pronounced when $T$ and $\elln$ are both small, which is consistent with the theoretical rate conditions. Third, as $(n, T)$ increases, the coverage probability (CP) approaches the nominal 0.95 level, suggesting that the variance approximation used for inference is reliable.
	\par 
	The results regarding $\hat{G}_k$ and $\rho(\hat{A})$ are as follows. First, as shown in Table \ref{sim:rhoG_mess}, the estimates for $\rho(\hat{A})$ are consistently less than 1 under both the baseline specification with $\bar{d}_0=10\%$ in Panel A and the alternative threshold with $\bar{d}_0=15\%$ in Panel B.% 
	\footnote{If the estimated $\rho(\hat{A})\geq 1$, following the approach of 
		\citet{hall2001nonparametric}, we can employ an Euclidean minimum distance estimator:
		$\hat{\theta}_{c}=\arg\min_{\theta} \sum_{k=1}^{\ell_{\theta}}%
		w_k(\hat{\theta}_{unc, k}-\theta_k)^2$ subject to 
		$\rho(A(\theta))\leq 1-\epsilon$, where $\hat{\theta}_{unc, k}$ is the unconstrained estimator, and the weights satisfy $\sum_{k}w_k=1$. 
		For example, 
		$w_{k}=\Var(\hat{\theta}_{unc, k})/\sum_{l=1}^{\ell_{\theta}}\Var(\hat{\theta}_{unc, l})$. This approach is also discussed in nonparametric SAR models \citep{malikov2017semiparametric,sun2018estimation}.} 
	The approximation of $\rho(\hat{A})$ improves with increasing $(n,T)$: Both bias and RMSE decline; even with $\elln=2$, larger $n$ and $T$ reduce bias, while $\elln=[n^{1/5}]+2$ produces the best performance in finite samples. This confirms that the estimated spatial dynamic system remains stable,  satisfying Assumption \ref{ass:appro_G} in the appendix regardless of the chosen distance cutoff.
	Second, from Table \ref{sim:heteG}, the bias and RMSE of the components in $\hat{\tilde{g}}_{k}$ decrease as $(n,T,\elln)$ increases. Analogous to $\pi$, BGMME yields the smallest bias and RMSE in most settings, followed by OGMME and then 2SLSE. However, this advantage may disappear when $T$ and $\elln$ are small. 
	In addition, the MAE also declines as $(n,T,\elln)$ increase, with the BGMME consistently yielding the smallest MAE. 
	\par 
	Moreover, we assess robustness to functional-form misspecification by simulating data from a SAR DGP and estimating a MESS model, and vice versa. 
	To facilitate comparability across parameterizations, Remark \ref{remark:MESS} implies a mapping between the two spatial parameters: 
	under a SAR DGP, $\gamma_0=0.3$ corresponds to $\gamma_0-1$ under the MESS parameterization; under a MESS DGP, $\gamma_0=-0.7$ corresponds to $\gamma_0+1$ under the SAR parameterization. We set $\beta_0=1$, keep all other design features unchanged, and maintain the stability condition $\rho(\hat{A})<1$ in estimation. Table \ref{sim:misspecification} shows that under misspecification, although the finite-dimensional $\pi$ has smaller biases, the CPs fall below the nominal 0.95 level and deteriorate as $(n,T)$ increases.
	%%%%%%%%%%%%%%%%%%%%%%%%%%%%%%%%%%%%%%%%%%%%%%%%%
	\begin{table}[htbp]
		\centering \small \stla{0cm} \stl{0.8mm} 
		\caption{Finite sample performance of estimators of $\pi$ for the MESS, $\Var(\epsilon_{it})=\sigma_i^2$.}
		\begin{spacing}{0.9}%\hspace*{-5mm}
			\begin{tabular}{lcccccclcccccc}
				\toprule
				\multirow{2}[4]{*}{} & \multicolumn{2}{c}{2SLS} & \multicolumn{2}{c}{OGMM} & \multicolumn{2}{c}{BGMM} & \multirow{2}[4]{*}{} & \multicolumn{2}{c}{2SLS} & \multicolumn{2}{c}{OGMM} & \multicolumn{2}{c}{BGMM} \\
				\cmidrule{2-7}\cmidrule{9-14}          & $\gamma$ & $\beta$  & $\gamma$ & $\beta$  & $\gamma$ & $\beta$  &       & $\gamma$ & $\beta$  & $\gamma$ & $\beta$  & $\gamma$ & $\beta$ \\
				\cmidrule{2-7}\cmidrule{9-14}
				\multicolumn{7}{c}{$(n,T,\elln)=(100,10,2)$}                         & \multicolumn{7}{c}{$(n,T,\elln)=(100,10,[n^{1/5}]+2)$} \\
				\cmidrule{2-7}\cmidrule{9-14}
				Bias  & -0.0237  & 0.0045  & -0.0255  & 0.0001  & -0.0278  & -0.0032  &  & -0.0229  & 0.0042  & -0.0252  & 0.0005  & -0.0277  & 0.0015  \\
				ESD   & 0.0381  & 0.0344  & 0.0380  & 0.0321  & 0.0412  & 0.0349  &   & 0.0360  & 0.0344  & 0.0366  & 0.0317  & 0.0396  & 0.0340  \\
				RMSE  & 0.0449  & 0.0347  & 0.0457  & 0.0321  & 0.0497  & 0.0350  &  & 0.0427  & 0.0347  & 0.0444  & 0.0317  & 0.0483  & 0.0340  \\
				CP    & 0.7970  & 0.9070  & 0.7950  & 0.9280  & 0.7850  & 0.9250  &   & 0.8930  & 0.9030  & 0.8830  & 0.9100  & 0.8720  & 0.8990  \\  
				\cmidrule{2-7}\cmidrule{9-14}
				\multicolumn{7}{c}{$(n,T,\elln)=(100,25,2)$}                         & \multicolumn{7}{c}{$(n,T,\elln)=(100,25,[n^{1/5}]+2)$} \\
				\cmidrule{2-7}\cmidrule{9-14}
				\cmidrule{2-7}\cmidrule{9-14}    Bias  & -0.0151  & 0.0023  & -0.0178  & -0.0004  & -0.0218  & -0.0072  &       & -0.0149  & 0.0014  & -0.0170  & -0.0007  & -0.0180  & 0.0020  \\
				ESD   & 0.0248  & 0.0222  & 0.0247  & 0.0221  & 0.0268  & 0.0246  &       & 0.0242  & 0.0216  & 0.0245  & 0.0212  & 0.0266  & 0.0235  \\
				RMSE  & 0.0290  & 0.0223  & 0.0304  & 0.0221  & 0.0345  & 0.0256  &       & 0.0284  & 0.0216  & 0.0298  & 0.0212  & 0.0321  & 0.0236  \\
				CP    & 0.8960  & 0.9070  & 0.8840  & 0.9270  & 0.8730  & 0.9160  &    & 0.8980  & 0.9140  & 0.8990  & 0.9180  & 0.8880  & 0.9070  \\
				\cmidrule{2-7}\cmidrule{9-14}
				\multicolumn{7}{c}{$(n,T,\elln)=(200,10,2)$}                         & \multicolumn{7}{c}{$(n,T,\elln)=(200,10,[n^{1/5}]+2)$} \\
				\cmidrule{2-7}\cmidrule{9-14}    Bias  & -0.0201  & 0.0055  & -0.0230  & 0.0038  & -0.0240  & 0.0003  &       & -0.0170  & 0.0041  & -0.0200  & 0.0028  & -0.0198  & 0.0038  \\
				ESD   & 0.0193  & 0.0199  & 0.0193  & 0.0183  & 0.0191  & 0.0182  &       & 0.0193  & 0.0195  & 0.0186  & 0.0187  & 0.0184  & 0.0186  \\
				RMSE  & 0.0279  & 0.0206  & 0.0301  & 0.0187  & 0.0307  & 0.0182  &       & 0.0257  & 0.0199  & 0.0273  & 0.0189  & 0.0271  & 0.0189  \\
				CP    & 0.8000  & 0.8990  & 0.8060  & 0.9260  & 0.7960  & 0.9150  &    & 0.8430  & 0.8880  & 0.8140  & 0.9200  & 0.8040  & 0.9090  \\
				\cmidrule{2-7}\cmidrule{9-14}
				\multicolumn{7}{c}{$(n,T,\elln)=(200,25,2)$}                         & \multicolumn{7}{c}{$(n,T,\elln)=(200,25,[n^{1/5}]+2)$}
				\\
				\cmidrule{2-7}\cmidrule{9-14}
				\cmidrule{2-7}\cmidrule{9-14}    Bias  & 0.0069  & 0.0046  & 0.0069  & 0.0046  & 0.0070  & 0.0045  &       & -0.0066  & 0.0036  & 0.0072  & 0.0025  & -0.0081  & 0.0035  \\
				ESD   & 0.0142  & 0.0131  & 0.0136  & 0.0128  & 0.0135  & 0.0126  &       & 0.0136  & 0.0129  & 0.0135  & 0.0124  & 0.0132  & 0.0122  \\
				RMSE  & 0.0158  & 0.0138  & 0.0153  & 0.0136  & 0.0152  & 0.0134  &       & 0.0151  & 0.0134  & 0.0153  & 0.0127  & 0.0155  & 0.0127  \\
				CP    & 0.9260  & 0.9400  & 0.9360  & 0.9330  & 0.9250  & 0.9220  &    & 0.9400  & 0.9350  & 0.9420  & 0.9420  & 0.9350  & 0.9400  \\
				\cmidrule{2-7}\cmidrule{9-14}
				\multicolumn{7}{c}{$(n,T,\elln)=(400,10,2)$}                         & \multicolumn{7}{c}{$(n,T,\elln)=(400,10,[n^{1/5}]+2)$} \\
				\cmidrule{2-7}\cmidrule{9-14}
				\cmidrule{2-7}\cmidrule{9-14}    Bias  & -0.0105  & 0.0065  & -0.0075  & 0.0050  & -0.0074  & -0.0047  &       & -0.0101  & 0.0053  & -0.0070  & 0.0043  & -0.0063  & 0.0045  \\
				ESD   & 0.0169  & 0.0146  & 0.0169  & 0.0146  & 0.0167  & 0.0145  &       & 0.0124  & 0.0116  & 0.0125  & 0.0112  & 0.0121  & 0.0112  \\
				RMSE  & 0.0199  & 0.0160  & 0.0185  & 0.0154  & 0.0182  & 0.0152  &       & 0.0160  & 0.0127  & 0.0143  & 0.0120  & 0.0136  & 0.0120  \\
				CP    & 0.9120  & 0.9220  & 0.9270  & 0.9300  & 0.9190  & 0.9320  &  & 0.9130  & 0.9300  & 0.9320  & 0.9360  & 0.9240  & 0.9330  \\
				\cmidrule{2-7}\cmidrule{9-14}
				\multicolumn{7}{c}{$(n,T,\elln)=(400,25,2)$}                         & \multicolumn{7}{c}{$(n,T,\elln)=(400,25,[n^{1/5}]+2)$}
				\\
				\cmidrule{2-7}\cmidrule{9-14}
				Bias  & -0.0053  & 0.0016  & -0.0033  & -0.0005  & -0.0036  & -0.0044  &   & -0.0046  & 0.0012  & -0.0026  & -0.0009  & -0.0029  & -0.0032  \\
				ESD   & 0.0094  & 0.0086  & 0.0094  & 0.0086  & 0.0093  & 0.0086  &    & 0.0090  & 0.0085  & 0.0090  & 0.0086  & 0.0089  & 0.0086  \\
				RMSE  & 0.0108  & 0.0088  & 0.0100  & 0.0086  & 0.0100  & 0.0097  &   & 0.0101  & 0.0086  & 0.0094  & 0.0087  & 0.0094  & 0.0092  \\
				CP    & 0.9250  & 0.9360  & 0.9320  & 0.9380  & 0.9250  & 0.9360  &     & 0.9340  & 0.9420  & 0.9400  & 0.9450  & 0.9410  & 0.9420  \\
				\bottomrule
			\end{tabular}%
			\hspace*{-1cm}
			\begin{tablenotes}
				\footnotesize
				\item \textbf{Note:} The true parameters are set to $\pi_0=(-0.7,1)'$, and $\bard_0=10\%$. The results are based on 1,000 Monte Carlo replications. Bias denotes the mean bias of the estimates, ESD denotes the standard deviation, RMSE denotes the root mean squared error, and CP denotes the 95\% coverage probability. The 2SLS refers to the two-stage least square estimator, OGMM refers to the feasible optimal GMM estimator, and BGMM refers to the feasible best GMM estimator. 
			\end{tablenotes}
		\end{spacing}
		\label{sim:hei}%
	\end{table}%
	
	% \clearpage
	% \thispagestyle{empty}
	\begin{table}[htbp]
		\centering 
		\small
		\stla{0cm} \stl{0.8mm} 
		\caption{Estimated $\rho(\hat{A})$ for the MESS, $\mathrm{Var}(\epsilon_{it})=\sigma_i^2$, $\rho(A)=0.6949$.}
		
		\begin{spacing}{0.9}\hspace*{-7mm}
			\begin{minipage}{1.1\textwidth} 
				\begin{tabular}{lcccccclcccccc}
					\toprule
					% ---------------- PANEL A ----------------
					\multicolumn{14}{c}{Panel A: Baseline results, $\bard_0=10\%$.} \\
					\cmidrule{2-7}\cmidrule{9-14}
					\multicolumn{7}{c}{$(n,T)=(100,10)$}             & \multicolumn{7}{c}{$(n,T,\elln)=(100,25)$} \\
					\cmidrule{2-7}\cmidrule{9-14}
					\multirow{2}[4]{*}{} & \multicolumn{3}{c}{$\elln=2$} & \multicolumn{3}{c}{$\elln=[n^{1/5}]+2$} & \multirow{2}[4]{*}{} & \multicolumn{3}{c}{$\elln=2$} & \multicolumn{3}{c}{$\elln=[n^{1/5}]+2$} \\
					\cmidrule{2-7}\cmidrule{9-14}          & 2SLS &OGMM & BGMM & 2SLS &OGMM & BGMM  &        & 2SLS &OGMM & BGMM & 2SLS &OGMM & BGMM \\              
					\cmidrule{2-7}\cmidrule{9-14}
					\cmidrule{2-7}\cmidrule{9-14}    Mean  & 0.5664  & 0.5781  & 0.5856  & 0.5784  & 0.5896  & 0.5925  &       & 0.6022  & 0.5961  & 0.5999  & 0.6178  & 0.6067  & 0.6112  \\
					% Bias  & -0.1285  & -0.1168  & -0.1093  & -0.1165  & -0.1053  & -0.1024  &       & -0.0927  & -0.0988  & -0.0950  & -0.0771  & -0.0882  & -0.0837  \\
					SD    & 0.1812  & 0.1797  & 0.2167  & 0.1890  & 0.1886  & 0.1866  &       & 0.1132  & 0.1135  & 0.1180  & 0.1017  & 0.1023  & 0.1021  \\
					RMSE  & 0.2221  & 0.2143  & 0.2427  & 0.2220  & 0.2160  & 0.2129  &       & 0.1463  & 0.1505  & 0.1515  & 0.1276  & 0.1351  & 0.1320  \\
					\cmidrule{2-7}\cmidrule{9-14}
					\multicolumn{7}{c}{$(n,T)=(200,10)$}     & \multicolumn{7}{c}{$(n,T,\elln)=(200,25)$} \\
					\cmidrule{2-7}\cmidrule{9-14}
					\multirow{2}[4]{*}{} & \multicolumn{3}{c}{$\elln=2$} & \multicolumn{3}{c}{$\elln=[n^{1/5}]+2$} & \multirow{2}[4]{*}{} & \multicolumn{3}{c}{$\elln=2$} & \multicolumn{3}{c}{$\elln=[n^{1/5}]+2$} \\
					\cmidrule{2-7}\cmidrule{9-14}          & 2SLS &OGMM & BGMM & 2SLS &OGMM & BGMM  &        & 2SLS &OGMM & BGMM & 2SLS &OGMM & BGMM \\
					\cmidrule{2-7}\cmidrule{9-14}
					\cmidrule{2-7}\cmidrule{9-14}    Mean  & 0.5827  & 0.5907  & 0.5899  & 0.5796  & 0.5867  & 0.5969  &       & 0.6198  & 0.6152  & 0.6084  & 0.6212  & 0.6214  & 0.6228  \\
					% Bias  & -0.1122  & -0.1042  & -0.1050  & -0.1153  & -0.1082  & -0.0980  &       & -0.0751  & -0.0797  & -0.0865  & -0.0737  & -0.0735  & -0.0721  \\
					SD    & 0.1631  & 0.1617  & 0.1950  & 0.1491  & 0.1493  & 0.1494  &       & 0.0925  & 0.0948  & 0.0923  & 0.0865  & 0.0865  & 0.0865  \\
					RMSE  & 0.1980  & 0.1924  & 0.2215  & 0.1885  & 0.1844  & 0.1787  &       & 0.1191  & 0.1239  & 0.1265  & 0.1136  & 0.1135  & 0.1126  \\
					\cmidrule{2-7}\cmidrule{9-14}
					\multicolumn{7}{c}{$(n,T)=(400,10)$}                           & \multicolumn{7}{c}{$(n,T,\elln)=(400,25)$} \\
					\cmidrule{2-7}\cmidrule{9-14}
					\multirow{2}[4]{*}{} & \multicolumn{3}{c}{$\elln=2$} & \multicolumn{3}{c}{$\elln=[n^{1/5}]+2$} & \multirow{2}[4]{*}{} & \multicolumn{3}{c}{$\elln=2$} & \multicolumn{3}{c}{$\elln=[n^{1/5}]+2$} \\
					\cmidrule{2-7}\cmidrule{9-14}          & 2SLS &OGMM & BGMM & 2SLS &OGMM & BGMM  &        & 2SLS &OGMM & BGMM & 2SLS &OGMM & BGMM \\
					\cmidrule{2-7}\cmidrule{9-14}
					\cmidrule{2-7}\cmidrule{9-14}    Mean  & 0.6321  & 0.6354  & 0.6351  & 0.6366  & 0.6362  & 0.6388  &       & 0.6458  & 0.6482  & 0.6485  & 0.6485  & 0.6486  & 0.6491  \\
					% Bias  & -0.0628  & -0.0595  & -0.0598  & -0.0583  & -0.0587  & -0.0561  &       & -0.0491  & -0.0467  & -0.0465  & -0.0464  & -0.0463  & -0.0458  \\
					SD    & 0.1059  & 0.1061  & 0.1061  & 0.0955  & 0.0956  & 0.0953  &       & 0.0754  & 0.0755  & 0.0756  & 0.0753  & 0.0731  & 0.0723  \\
					RMSE  & 0.1231  & 0.1217  & 0.1218  & 0.1119  & 0.1122  & 0.1106  &       & 0.0900  & 0.0888  & 0.0887  & 0.0885  & 0.0865  & 0.0856  \\
					% ---------------- PANEL B ----------------
					\cmidrule{2-7}\cmidrule{9-14}
					\multicolumn{14}{c}{Panel B: Additional results, $\bar{d}_0=15\%$.} \\
					\cmidrule{2-7}\cmidrule{9-14}
					\multicolumn{7}{c}{$(n,T)=(100,10)$}                           & \multicolumn{7}{c}{$(n,T,\elln)=(100,25)$} \\
					\cmidrule{2-7}\cmidrule{9-14}
					\multirow{2}[4]{*}{} & \multicolumn{3}{c}{$\elln=2$} & \multicolumn{3}{c}{$\elln=[n^{1/5}]+2$} & \multirow{2}[4]{*}{} & \multicolumn{3}{c}{$\elln=2$} & \multicolumn{3}{c}{$\elln=[n^{1/5}]+2$} \\
					\cmidrule{2-7}\cmidrule{9-14}         & 
					\multicolumn{1}{c}{2SLS} & \multicolumn{1}{c}{OGMM} & \multicolumn{1}{c}{BGMM} & \multicolumn{1}{c}{2SLS} & \multicolumn{1}{c}{OGMM} & \multicolumn{1}{c}{BGMM} &       & \multicolumn{1}{c}{2SLS} & \multicolumn{1}{c}{OGMM} & \multicolumn{1}{c}{BGMM} & \multicolumn{1}{c}{2SLS} & \multicolumn{1}{c}{OGMM} & \multicolumn{1}{c}{BGMM} \\
					\cmidrule{2-7}\cmidrule{9-14}    Mean  & 0.5618  & 0.5574  & 0.5603  & 0.5709  & 0.5682  & 0.5732  &       & 0.5830  & 0.5819  & 0.5820  & 0.5821  & 0.5921  & 0.5921  \\
					% Bias  & -0.1331  & -0.1375  & -0.1346  & -0.1240  & -0.1267  & -0.1217  &       & -0.1119  & -0.1130  & -0.1129  & -0.1128  & -0.1028  & -0.1028  \\
					SD    & 0.1107  & 0.1095  & 0.1089  & 0.0927  & 0.0904  & 0.0916  &       & 0.0955  & 0.0964  & 0.0990  & 0.0553  & 0.0520  & 0.0534  \\
					RMSE  & 0.1731  & 0.1758  & 0.1731  & 0.1548  & 0.1556  & 0.1523  &       & 0.1471  & 0.1485  & 0.1502  & 0.1256  & 0.1152  & 0.1159  \\
					\cmidrule{2-7}\cmidrule{9-14}
					\multicolumn{7}{c}{$(n,T)=(200,10)$}                           & \multicolumn{7}{c}{$(n,T,\elln)=(200,25)$} \\
					\cmidrule{2-7}\cmidrule{9-14}
					\multirow{2}[4]{*}{} & \multicolumn{3}{c}{$\elln=2$} & \multicolumn{3}{c}{$\elln=[n^{1/5}]+2$} & \multirow{2}[4]{*}{} & \multicolumn{3}{c}{$\elln=2$} & \multicolumn{3}{c}{$\elln=[n^{1/5}]+2$} \\
					\cmidrule{2-7}\cmidrule{9-14}          & 2SLS &OGMM & BGMM & 2SLS &OGMM & BGMM  &        & 2SLS &OGMM & BGMM & 2SLS &OGMM & BGMM \\
					\cmidrule{2-7}\cmidrule{9-14}
					\cmidrule{2-7}\cmidrule{9-14}    Mean  & 0.6032  & 0.6017  & 0.6017  & 0.6017  & 0.6017  & 0.6018  &       & 0.6179  & 0.6215  & 0.6233  & 0.6233  & 0.6227  & 0.6226  \\
					% Bias  & -0.0917  & -0.0932  & -0.0932  & -0.0932  & -0.0932  & -0.0931  &       & -0.0770  & -0.0734  & -0.0716  & -0.0716  & -0.0722  & -0.0723  \\
					SD    & 0.1061  & 0.1077  & 0.1063  & 0.0989  & 0.0958  & 0.0963  &       & 0.0861  & 0.0858  & 0.0858  & 0.0840  & 0.0837  & 0.0840  \\
					RMSE  & 0.1403  & 0.1424  & 0.1414  & 0.1359  & 0.1336  & 0.1340  &       & 0.1155  & 0.1129  & 0.1117  & 0.1104  & 0.1105  & 0.1109  \\
					\cmidrule{2-7}\cmidrule{9-14}
					\multicolumn{7}{c}{$(n,T)=(400,10)$}                           & \multicolumn{7}{c}{$(n,T,\elln)=(400,25)$} \\
					\cmidrule{2-7}\cmidrule{9-14}
					\multirow{2}[4]{*}{} & \multicolumn{3}{c}{$\elln=2$} & \multicolumn{3}{c}{$\elln=[n^{1/5}]+2$} & \multirow{2}[4]{*}{} & \multicolumn{3}{c}{$\elln=2$} & \multicolumn{3}{c}{$\elln=[n^{1/5}]+2$} \\
					\cmidrule{2-7}\cmidrule{9-14}          & 2SLS &OGMM & BGMM & 2SLS &OGMM & BGMM  &        & 2SLS &OGMM & BGMM & 2SLS &OGMM & BGMM \\
					\cmidrule{2-7}\cmidrule{9-14}
					\cmidrule{2-7}\cmidrule{9-14}
					Mean  & 0.6056  & 0.6071  & 0.6059  & 0.6075  & 0.6090  & 0.6111  &       & 0.6289  & 0.6331  & 0.6318  & 0.6311  & 0.6325  & 0.6321  \\
					% Bias  & -0.0893  & -0.0878  & -0.0890  & -0.0874  & -0.0859  & -0.0838  &       & -0.0660  & -0.0618  & -0.0631  & -0.0638  & -0.0624  & -0.0628  \\
					SD    & 0.1094  & 0.1025  & 0.1106  & 0.0955  & 0.0964  & 0.0990  &       & 0.0760  & 0.0756  & 0.0753  & 0.0739  & 0.0736  & 0.0740  \\
					RMSE  & 0.1412  & 0.1350  & 0.1420  & 0.1295  & 0.1291  & 0.1297  &       & 0.1006  & 0.0977  & 0.0982  & 0.0976  & 0.0965  & 0.0971  \\
					\bottomrule
				\end{tabular}%
				\hspace*{-1cm}
				\begin{tablenotes}
					\footnotesize
					\item \textbf{Note:} The results are based on 1,000 Monte Carlo replications. Mean denotes the mean of $\rho(\hat{A})$'s. ESD denotes the standard deviation of $\rho(\hat{A})$'s, and RMSE denotes the root mean squared error of $\rho(\hat{A})$'s. The 2SLS refers to the two-stage least square estimator, OGMM refers to the feasible optimal GMM estimator, and BGMM refers to the feasible best GMM estimator.
				\end{tablenotes}	
			\end{minipage}
		\end{spacing}	
		\label{sim:rhoG_mess}%
	\end{table}%
	% \clearpage

	\begin{sidewaystable}[htbp]
		\centering\small\stla{0cm} \stl{0.8mm} \vspace*{-5mm}
		\caption{Finite sample performance of estimators of $G_{1}, G_{2}$ and $G_{3}$ for the MESS. $\Var(\epsilon_{it})=\sigma_i^2$.}
		\begin{spacing}{0.5}\hspace*{-5mm}
			\begin{tabular}{lrrrrrrrrrlrrrrrrrrr}
				\toprule
				\multirow{4}[3]{*}{}	& \multicolumn{9}{c}{$(n,T,\elln)=(100,10,2)$} &   & \multicolumn{9}{c}{$(n,T,\elln)=(100,10,[n^{1/5}]+2)$} \\
				\cmidrule{2-10}\cmidrule{12-20}   
				\multirow{2}[4]{*}{} & \multicolumn{3}{c}{2SLS} & \multicolumn{3}{c}{OGMM} & \multicolumn{3}{c}{BGMM} & \multirow{2}[4]{*}{} & \multicolumn{3}{c}{2SLS} & \multicolumn{3}{c}{OGMM} & \multicolumn{3}{c}{BGMM} \\        
				\cmidrule{2-10}\cmidrule{12-20}        & \multicolumn{1}{c}{$\tilde{g}_1$} & \multicolumn{1}{c}{$\tilde{g}_2$} & \multicolumn{1}{c}{$\tilde{g}_3$} & \multicolumn{1}{c}{$\tilde{g}_{1}$} & \multicolumn{1}{c}{$\tilde{g}_{2}$} & \multicolumn{1}{c}{$\tilde{g}_{3}$} & \multicolumn{1}{c}{$\tilde{g}_{1}$} & \multicolumn{1}{c}{$\tilde{g}_{2}$} & \multicolumn{1}{c}{$\tilde{g}_{3}$} &       & \multicolumn{1}{c}{$\tilde{g}_{1}$} & \multicolumn{1}{c}{$\tilde{g}_{2}$} & \multicolumn{1}{c}{$\tilde{g}_{3}$} & \multicolumn{1}{c}{$\tilde{g}_{1}$} & \multicolumn{1}{c}{$\tilde{g}_{2}$} & \multicolumn{1}{c}{$\tilde{g}_{3}$} & \multicolumn{1}{c}{$\tilde{g}_{1}$} & \multicolumn{1}{c}{$\tilde{g}_{2}$} & \multicolumn{1}{c}{$\tilde{g}_{3}$} \\
				\cmidrule{2-10}\cmidrule{12-20}    MAE   & 0.0346  & 0.0396  & 0.0733  & 0.0342  & 0.0393  & 0.0712  & 0.0342  & 0.0394  & 0.0693  &       & 0.0344  & 0.0401  & 0.0728  & 0.0340  & 0.0397  & 0.0697  & 0.0343  & 0.0395  & 0.0679  \\
				Bias  & -0.0221  & -0.0326  & -0.0664  & -0.0206  & -0.0321  & -0.0648  & -0.0189  & -0.0315  & -0.0630  &       & -0.0219  & -0.0328  & -0.0661  & -0.0196  & -0.0320  & -0.0635  & -0.0182  & -0.0311  & -0.0618  \\
				RMSE  & 0.0266  & 0.0364  & 0.0774  & 0.0254  & 0.0359  & 0.0749  & 0.0246  & 0.0358  & 0.0727  &       & 0.0269  & 0.0374  & 0.0771  & 0.0253  & 0.0368  & 0.0734  & 0.0250  & 0.0362  & 0.0711  \\
				\cmidrule{2-10}\cmidrule{12-20}   
				& \multicolumn{9}{c}{$(n,T,\elln)=(100,25,2)$} &   & \multicolumn{9}{c}{$(n,T,\elln)=(100,25,[n^{1/5}]+2)$} \\
				\cmidrule{2-10}\cmidrule{12-20}
				\multirow{2}[4]{*}{} & \multicolumn{3}{c}{2SLS} & \multicolumn{3}{c}{OGMM} & \multicolumn{3}{c}{BGMM} & \multirow{2}[4]{*}{} & \multicolumn{3}{c}{2SLS} & \multicolumn{3}{c}{OGMM} & \multicolumn{3}{c}{BGMM} \\        
				\cmidrule{2-10}\cmidrule{12-20}          & \multicolumn{1}{c}{$\tilde{g}_{1}$} & \multicolumn{1}{c}{$\tilde{g}_{2}$} & \multicolumn{1}{c}{$\tilde{g}_{3}$} & \multicolumn{1}{c}{$\tilde{g}_{1}$} & \multicolumn{1}{c}{$\tilde{g}_{2}$} & \multicolumn{1}{c}{$\tilde{g}_{3}$} & \multicolumn{1}{c}{$\tilde{g}_{1}$} & \multicolumn{1}{c}{$\tilde{g}_{2}$} & \multicolumn{1}{c}{$\tilde{g}_{3}$} &       & \multicolumn{1}{c}{$\tilde{g}_{1}$} & \multicolumn{1}{c}{$\tilde{g}_{2}$} & \multicolumn{1}{c}{$\tilde{g}_{3}$} & \multicolumn{1}{c}{$\tilde{g}_{1}$} & \multicolumn{1}{c}{$\tilde{g}_{2}$} & \multicolumn{1}{c}{$\tilde{g}_{3}$} & \multicolumn{1}{c}{$\tilde{g}_{1}$} & \multicolumn{1}{c}{$\tilde{g}_{2}$} & \multicolumn{1}{c}{$\tilde{g}_{3}$} \\
				\cmidrule{2-10}\cmidrule{12-20}    
				MAE   & 0.0160  & 0.0189  & 0.0456  & 0.0159  & 0.0188  & 0.0445  & 0.0160  & 0.0189  & 0.0433  &       & 0.0162  & 0.0190  & 0.0483  & 0.0162  & 0.0189  & 0.0467  & 0.0167  & 0.0189  & 0.0460  \\
				Bias  & -0.0089  & -0.0154  & -0.0433  & -0.0080  & -0.0152  & -0.0424  & -0.0069  & -0.0150  & -0.0413  &       & -0.0094  & -0.0152  & -0.0437  & -0.0079  & -0.0150  & -0.0419  & -0.0080  & -0.0147  & -0.0417  \\
				RMSE  & 0.0115  & 0.0177  & 0.0490  & 0.0108  & 0.0176  & 0.0476  & 0.0106  & 0.0175  & 0.0462  &       & 0.0111  & 0.0171  & 0.0482  & 0.0098  & 0.0170  & 0.0463  & 0.0101  & 0.0169  & 0.0458  \\
				\cmidrule{2-10}\cmidrule{12-20}   
				& \multicolumn{9}{c}{$(n,T,\elln)=(200,10,2)$} &   & \multicolumn{9}{c}{$(n,T,\elln)=(200,10,[n^{1/5}]+2)$} \\
				\cmidrule{2-10}\cmidrule{12-20}    
				\multirow{2}[4]{*}{} & \multicolumn{3}{c}{2SLS} & \multicolumn{3}{c}{OGMM} & \multicolumn{3}{c}{BGMM} & \multirow{2}[4]{*}{} & \multicolumn{3}{c}{2SLS} & \multicolumn{3}{c}{OGMM} & \multicolumn{3}{c}{BGMM} \\        
				\cmidrule{2-10}\cmidrule{12-20}        & \multicolumn{1}{c}{$\tilde{g}_1$} & \multicolumn{1}{c}{$\tilde{g}_2$} & \multicolumn{1}{c}{$\tilde{g}_3$} & \multicolumn{1}{c}{$\tilde{g}_{1}$} & \multicolumn{1}{c}{$\tilde{g}_{2}$} & \multicolumn{1}{c}{$\tilde{g}_{3}$} & \multicolumn{1}{c}{$\tilde{g}_{1}$} & \multicolumn{1}{c}{$\tilde{g}_{2}$} & \multicolumn{1}{c}{$\tilde{g}_{3}$} &       & \multicolumn{1}{c}{$\tilde{g}_{1}$} & \multicolumn{1}{c}{$\tilde{g}_{2}$} & \multicolumn{1}{c}{$\tilde{g}_{3}$} & \multicolumn{1}{c}{$\tilde{g}_{1}$} & \multicolumn{1}{c}{$\tilde{g}_{2}$} & \multicolumn{1}{c}{$\tilde{g}_{3}$} & \multicolumn{1}{c}{$\tilde{g}_{1}$} & \multicolumn{1}{c}{$\tilde{g}_{2}$} & \multicolumn{1}{c}{$\tilde{g}_{3}$} \\
				\cmidrule{2-10}\cmidrule{12-20}    
				MAE   & 0.0333  & 0.0377  & 0.0417  & 0.0326  & 0.0376  & 0.0491  & 0.0321  & 0.0381  & 0.0466  &       & 0.0327  & 0.0368  & 0.0393  & 0.0319  & 0.0366  & 0.0359  & 0.0320  & 0.0363  & 0.0338  \\
				Bias  & -0.0233  & -0.0325  & -0.0397  & -0.0208  & -0.0324  & -0.0371  & -0.0180  & -0.0328  & -0.0342  &       & -0.0228  & -0.0312  & -0.0371  & -0.0198  & -0.0308  & -0.0305  & -0.0177  & -0.0298  & -0.0312  \\
				RMSE  & 0.0248  & 0.0337  & 0.0461  & 0.0225  & 0.0337  & 0.0432  & 0.0203  & 0.0343  & 0.0401  &       & 0.0247  & 0.0327  & 0.0442  & 0.0220  & 0.0324  & 0.0378  & 0.0206  & 0.0316  & 0.0376  \\
				\cmidrule{2-10}\cmidrule{12-20}   
				& \multicolumn{9}{c}{$(n,T,\elln)=(200,25,2)$} &   & \multicolumn{9}{c}{$(n,T,\elln)=(200,25,[n^{1/5}]+2)$} \\
				\cmidrule{2-10}\cmidrule{12-20}  
				\multirow{2}[4]{*}{} & \multicolumn{3}{c}{2SLS} & \multicolumn{3}{c}{OGMM} & \multicolumn{3}{c}{BGMM} & \multirow{2}[4]{*}{} & \multicolumn{3}{c}{2SLS} & \multicolumn{3}{c}{OGMM} & \multicolumn{3}{c}{BGMM} \\  
				\cmidrule{2-10}\cmidrule{12-20}          & \multicolumn{1}{c}{$\tilde{g}_{1}$} & \multicolumn{1}{c}{$\tilde{g}_{2}$} & \multicolumn{1}{c}{$\tilde{g}_{3}$} & \multicolumn{1}{c}{$\tilde{g}_{1}$} & \multicolumn{1}{c}{$\tilde{g}_{2}$} & \multicolumn{1}{c}{$\tilde{g}_{3}$} & \multicolumn{1}{c}{$\tilde{g}_{1}$} & \multicolumn{1}{c}{$\tilde{g}_{2}$} & \multicolumn{1}{c}{$\tilde{g}_{3}$} &       & \multicolumn{1}{c}{$\tilde{g}_{1}$} & \multicolumn{1}{c}{$\tilde{g}_{2}$} & \multicolumn{1}{c}{$\tilde{g}_{3}$} & \multicolumn{1}{c}{$\tilde{g}_{1}$} & \multicolumn{1}{c}{$\tilde{g}_{2}$} & \multicolumn{1}{c}{$\tilde{g}_{3}$} & \multicolumn{1}{c}{$\tilde{g}_{1}$} & \multicolumn{1}{c}{$\tilde{g}_{2}$} & \multicolumn{1}{c}{$\tilde{g}_{3}$} \\
				\cmidrule{2-10}\cmidrule{12-20}    
				MAE   & 0.0074  & 0.0089  & 0.0150  & 0.0074  & 0.0089  & 0.0148  & 0.0076  & 0.0089  & 0.0143  &       & 0.0073  & 0.0085  & 0.0149  & 0.0073  & 0.0085  & 0.0145  & 0.0072  & 0.0084  & 0.0139  \\
				Bias  & -0.0038  & -0.0065  & -0.0104  & -0.0038  & -0.0065  & -0.0104  & -0.0031  & -0.0061  & -0.0103  &       & -0.0037  & -0.0063  & -0.0102  & -0.0037  & -0.0063  & -0.0102  & -0.0031  & -0.0060  & -0.0101  \\
				RMSE  & 0.0056  & 0.0089  & 0.0134  & 0.0056  & 0.0089  & 0.0134  & 0.0052  & 0.0086  & 0.0133  &       & 0.0052  & 0.0083  & 0.0130  & 0.0052  & 0.0084  & 0.0130  & 0.0047  & 0.0080  & 0.0129  \\
				\cmidrule{2-10}\cmidrule{12-20}   
				& \multicolumn{9}{c}{$(n,T,\elln)=(400,10,2)$} &   & \multicolumn{9}{c}{$(n,T,\elln)=(400,10,[n^{1/5}]+2)$} \\
				\cmidrule{2-10}\cmidrule{12-20}    
				\multirow{2}[4]{*}{} & \multicolumn{3}{c}{2SLS} & \multicolumn{3}{c}{OGMM} & \multicolumn{3}{c}{BGMM} & \multirow{2}[4]{*}{} & \multicolumn{3}{c}{2SLS} & \multicolumn{3}{c}{OGMM} & \multicolumn{3}{c}{BGMM} \\  
				\cmidrule{2-10}\cmidrule{12-20}        & \multicolumn{1}{c}{$\tilde{g}_1$} & \multicolumn{1}{c}{$\tilde{g}_2$} & \multicolumn{1}{c}{$\tilde{g}_3$} & \multicolumn{1}{c}{$\tilde{g}_{1}$} & \multicolumn{1}{c}{$\tilde{g}_{2}$} & \multicolumn{1}{c}{$\tilde{g}_{3}$} & \multicolumn{1}{c}{$\tilde{g}_{1}$} & \multicolumn{1}{c}{$\tilde{g}_{2}$} & \multicolumn{1}{c}{$\tilde{g}_{3}$} &       & \multicolumn{1}{c}{$\tilde{g}_{1}$} & \multicolumn{1}{c}{$\tilde{g}_{2}$} & \multicolumn{1}{c}{$\tilde{g}_{3}$} & \multicolumn{1}{c}{$\tilde{g}_{1}$} & \multicolumn{1}{c}{$\tilde{g}_{2}$} & \multicolumn{1}{c}{$\tilde{g}_{3}$} & \multicolumn{1}{c}{$\tilde{g}_{1}$} & \multicolumn{1}{c}{$\tilde{g}_{2}$} & \multicolumn{1}{c}{$\tilde{g}_{3}$} \\
				\cmidrule{2-10}\cmidrule{12-20}    
				MAE   & 0.0086  & 0.0103  & 0.0174  & 0.0086  & 0.0103  & 0.0172  & 0.0088  & 0.0103  & 0.0166  &       & 0.0085  & 0.0099  & 0.0173  & 0.0085  & 0.0099  & 0.0168  & 0.0084  & 0.0097  & 0.0161  \\
				Bias  & -0.0044  & -0.0075  & -0.0121  & -0.0044  & -0.0075  & -0.0121  & -0.0036  & -0.0071  & -0.0119  &       & -0.0043  & -0.0073  & -0.0118  & -0.0043  & -0.0073  & -0.0118  & -0.0036  & -0.0070  & -0.0117  \\
				RMSE  & 0.0057  & 0.0081  & 0.0128  & 0.0057  & 0.0081  & 0.0128  & 0.0058  & 0.0080  & 0.0127  &       & 0.0052  & 0.0077  & 0.0125  & 0.0052  & 0.0078  & 0.0125  & 0.0048  & 0.0075  & 0.0123  \\
				\cmidrule{2-10}\cmidrule{12-20}   
				& \multicolumn{9}{c}{$(n,T,\elln)=(400,25,2)$} &   & \multicolumn{9}{c}{$(n,T,\elln)=(400,25,[n^{1/5}]+2)$} \\
				\cmidrule{2-10}\cmidrule{12-20} 
				\multirow{2}[4]{*}{} & \multicolumn{3}{c}{2SLS} & \multicolumn{3}{c}{OGMM} & \multicolumn{3}{c}{BGMM} & \multirow{2}[4]{*}{} & \multicolumn{3}{c}{2SLS} & \multicolumn{3}{c}{OGMM} & \multicolumn{3}{c}{BGMM} \\  
				\cmidrule{2-10}\cmidrule{12-20}          & \multicolumn{1}{c}{$\tilde{g}_{1}$} & \multicolumn{1}{c}{$\tilde{g}_{2}$} & \multicolumn{1}{c}{$\tilde{g}_{3}$} & \multicolumn{1}{c}{$\tilde{g}_{1}$} & \multicolumn{1}{c}{$\tilde{g}_{2}$} & \multicolumn{1}{c}{$\tilde{g}_{3}$} & \multicolumn{1}{c}{$\tilde{g}_{1}$} & \multicolumn{1}{c}{$\tilde{g}_{2}$} & \multicolumn{1}{c}{$\tilde{g}_{3}$} &       & \multicolumn{1}{c}{$\tilde{g}_{1}$} & \multicolumn{1}{c}{$\tilde{g}_{2}$} & \multicolumn{1}{c}{$\tilde{g}_{3}$} & \multicolumn{1}{c}{$\tilde{g}_{1}$} & \multicolumn{1}{c}{$\tilde{g}_{2}$} & \multicolumn{1}{c}{$\tilde{g}_{3}$} & \multicolumn{1}{c}{$\tilde{g}_{1}$} & \multicolumn{1}{c}{$\tilde{g}_{2}$} & \multicolumn{1}{c}{$\tilde{g}_{3}$} \\
				\cmidrule{2-10}\cmidrule{12-20}   
				MAE   & 0.0072  & 0.0085  & 0.0128  & 0.0072  & 0.0084  & 0.0126  & 0.0072  & 0.0085  & 0.0122  &    & 0.0070  & 0.0081  & 0.0124  & 0.0070  & 0.0081  & 0.0122  & 0.0069  & 0.0084  & 0.0118  \\
				Bias  & -0.0035  & -0.0060  & -0.0088  & -0.0035  & -0.0060  & -0.0088  & -0.0029  & -0.0056  & -0.0088  &   & -0.0034  & -0.0058  & -0.0086  & -0.0034  & -0.0058  & -0.0086  & -0.0029  & -0.0055  & -0.0085  \\
				RMSE  & 0.0041  & 0.0065  & 0.0103  & 0.0041  & 0.0065  & 0.0103  & 0.0038  & 0.0063  & 0.0103  &   & 0.0038  & 0.0062  & 0.0099  & 0.0038  & 0.0062  & 0.0099  & 0.0034  & 0.0060  & 0.0098  \\
				\bottomrule
			\end{tabular}%
			\hspace*{-1cm}
			\begin{tablenotes}
				\small
				\item \textbf{Note:} The $\tilde{g}_{k}$ extracts the column vectors composed of non-zero elements from the upper triangular submatrix of $G_{k}$, and $\bard_0=10\%$. The results are based on 1,000 Monte Carlo replications. MAE denotes the mean absolute error. Bias denotes the mean bias of the estimates, and RMSE denotes the root mean squared error. The 2SLS refers to the two-stage least square estimator, OGMM refers to the feasible optimal GMM estimator, and BGMM refers to the feasible best GMM estimator. 
			\end{tablenotes}
		\end{spacing}
		\label{sim:heteG}%
	\end{sidewaystable}
	
	\begin{table}[htbp]
		\centering \footnotesize \stla{0cm} \stl{0.8mm} 
		\caption{Finite sample performance of estimators of $\pi$ under misspecification, $\Var(\epsilon_{it})=\sigma_i^2$.}
		\begin{spacing}{1}%\hspace*{-5mm}
			\begin{tabular}{lcccccclcccccc}
				\toprule     
				& \multicolumn{13}{c}{DGP: SAR, Estimation: MESS} \\
				\cmidrule{2-7}\cmidrule{9-14} 
				\multirow{3}[6]{*}{} & \multicolumn{2}{c}{2SLS} & \multicolumn{2}{c}{OGMM} & \multicolumn{2}{c}{BGMM} &       & \multicolumn{2}{c}{2SLS} & \multicolumn{2}{c}{OGMM} & \multicolumn{2}{c}{BGMM} \\
				\cmidrule{2-7}\cmidrule{9-14}         & $\gamma$ & $\beta$  & $\gamma$ & $\beta$  & $\gamma$ & $\beta$  &       & $\gamma$ & $\beta$  & $\gamma$ & $\beta$  & $\gamma$ & $\beta$ \\
				\cmidrule{2-7}\cmidrule{9-14} 
				& \multicolumn{6}{c}{$(n,T,\elln)=(100,10,2)$}                         &       & \multicolumn{6}{c}{$(n,T,\elln)=(100,10,[n^{1/5}]+2)$} \\
				\cmidrule{2-7}\cmidrule{9-14} 
				Bias  & 0.0037  & -0.0024  & 0.0037  & 0.0012  & 0.0032  & 0.0012  
				&  & -0.0044  & 0.0002  & -0.0034  & 0.0002  & -0.0034  & -0.0005  \\
				RMSE  & 0.0442  & 0.0361  & 0.0400  & 0.0329  & 0.0399  & 0.0329  
				&  & 0.0430  & 0.0378  & 0.0400  & 0.0344  & 0.0400  & 0.0344  \\
				CP    & 0.8990  & 0.9000  & 0.8990  & 0.8970  & 0.7940  & 0.8610  
				&    & 0.9210  & 0.8850  & 0.9190  & 0.8820  & 0.7950  & 0.8410  \\
				\cmidrule{2-7}\cmidrule{9-14} 
				& \multicolumn{6}{c}{$(n,T,\elln)=(100,25,2)$}                         &       & \multicolumn{6}{c}{$(n,T,\elln)=(100,25,[n^{1/5}]+2)$} \\
				\cmidrule{2-7}\cmidrule{9-14} 
				Bias  & 0.0040  & 0.0024  & 0.0029  & 0.0024  & 0.0029  & -0.0020  
				&  & -0.0017  & 0.0010  & -0.0003  & 0.0010  & -0.0003  & -0.0001  \\
				RMSE  & 0.0340  & 0.0277  & 0.0280  & 0.0224  & 0.0280  & 0.0224  
				&  & 0.0306  & 0.0269  & 0.0265  & 0.0222  & 0.0265  & 0.0222  \\
				CP    & 0.9120  & 0.9100  & 0.9120  & 0.9070  & 0.7310  & 0.8230  
				&  & 0.9220  & 0.9220  & 0.9230  & 0.9220  & 0.7620  & 0.8350  \\
				\cmidrule{2-7}\cmidrule{9-14} 
				& \multicolumn{6}{c}{$(n,T,\elln)=(200,10,2)$}                         &       & \multicolumn{6}{c}{$(n,T,\elln)=(200,10,[n^{1/5}]+2)$} \\
				\cmidrule{2-7}\cmidrule{9-14} 
				Bias  & 0.0167  & 0.0033  & 0.0163  & 0.0033  & 0.0163  & -0.0008 
				&  & 0.0153  & 0.0034  & 0.0143  & 0.0034  & 0.0143  & 0.0020  \\
				RMSE  & 0.0332  & 0.0235  & 0.0282  & 0.0207  & 0.0282  & 0.0204 
				&   & 0.0324  & 0.0247  & 0.0275  & 0.0207  & 0.0275  & 0.0205  \\
				CP    & 0.7810  & 0.8880  & 0.7810  & 0.8860  & 0.5800  & 0.8340 
				&  & 0.7910  & 0.8880  & 0.7920  & 0.8880  & 0.6180  & 0.8100  \\
				\cmidrule{2-7}\cmidrule{9-14} 
				& \multicolumn{6}{c}{$(n,T,\elln)=(200,25,2)$}                         &       & \multicolumn{6}{c}{$(n,T,\elln)=(200,25,[n^{1/5}]+2)$} \\
				\cmidrule{2-7}\cmidrule{9-14} 
				Bias  & 0.0116  & 0.0032  & 0.0116  & 0.0032  & 0.0103  & 0.0022  
				&  & 0.0098  & 0.0026  & 0.0095  & 0.0019  & 0.0095  & 0.0019  \\
				RMSE  & 0.0235  & 0.0153  & 0.0196  & 0.0133  & 0.0189  & 0.0131 
				&   & 0.0227  & 0.0185  & 0.0187  & 0.0131  & 0.0187  & 0.0131  \\
				CP    & 0.7940  & 0.9220  & 0.7960  & 0.9220  & 0.5930  & 0.8840 
				&  & 0.7980  & 0.9120  & 0.8010  & 0.9120  & 0.6050  & 0.7750  \\
				\cmidrule{2-7}\cmidrule{9-14} 
				& \multicolumn{13}{c}{DGP: MESS; Estimation: SAR} \\
				\cmidrule{2-7}\cmidrule{9-14} 
				\multirow{3}[6]{*}{} & \multicolumn{2}{c}{2SLS} & \multicolumn{2}{c}{OGMM} & \multicolumn{2}{c}{BGMM} &       & \multicolumn{2}{c}{2SLS} & \multicolumn{2}{c}{OGMM} & \multicolumn{2}{c}{BGMM} \\
				\cmidrule{2-7}\cmidrule{9-14}         & $\gamma$ & $\beta$  & $\gamma$ & $\beta$  & $\gamma$ & $\beta$  &       & $\gamma$ & $\beta$  & $\gamma$ & $\beta$  & $\gamma$ & $\beta$ \\
				\cmidrule{2-7}\cmidrule{9-14}     
				& \multicolumn{6}{c}{$(n,T,\elln)=(100,10,2)$}                         &       & \multicolumn{6}{c}{$(n,T,\elln)=(100,10,[n^{1/5}]+2)$} \\
				\cmidrule{2-7}\cmidrule{9-14} 
				Bias  & -0.0250  & -0.0102  & -0.0164  & -0.0099  & -0.0164  & -0.0099  &   & -0.0247  & -0.0182  & -0.0132  & -0.0182  & -0.0132  & -0.0186  \\
				RMSE  & 0.0497  & 0.0310  & 0.0440  & 0.0308  & 0.0419  & 0.0308  &   & 0.0457  & 0.0384  & 0.0382  & 0.0382  & 0.0349  & 0.0384  \\
				CP    & 0.8880  & 0.9040  & 0.9060  & 0.9050  & 0.8070  & 0.8980  &      & 0.8590  & 0.9180  & 0.8940  & 0.9080  & 0.7540  & 0.8960  \\
				\cmidrule{2-7}\cmidrule{9-14} 
				& \multicolumn{6}{c}{$(n,T,\elln)=(100,25,2)$}                         &       & \multicolumn{6}{c}{$(n,T,\elln)=(100,25,[n^{1/5}]+2)$} \\
				\cmidrule{2-7}\cmidrule{9-14} 
				Bias  & -0.0101  & -0.0121  & -0.0049  & -0.0120  & -0.0049  & -0.0111  &  & -0.0086  & -0.0089  & -0.0023  & -0.0116  & -0.0023  & -0.0116  \\
				RMSE  & 0.0270  & 0.0243  & 0.0255  & 0.0242  & 0.0240  & 0.0238  &   & 0.0246  & 0.0213  & 0.0231  & 0.0226  & 0.0224  & 0.0223  \\
				CP    & 0.9100  & 0.8930  & 0.9310  & 0.8710  & 0.8320  & 0.8690  &      & 0.9120  & 0.8590  & 0.9270  & 0.8070  & 0.8070  & 0.7980  \\
				\cmidrule{2-7}\cmidrule{9-14} 
				& \multicolumn{6}{c}{$(n,T,\elln)=(200,10,2)$}                         &       & \multicolumn{6}{c}{$(n,T,\elln)=(200,10,[n^{1/5}]+2)$} \\
				\cmidrule{2-7}\cmidrule{9-14} 
				Bias  & -0.0220  & -0.0077  & -0.0185  & -0.0077  & -0.0185  & -0.0076  &   & -0.0172  & -0.0116  & -0.0131  & -0.0131  & -0.0131  & -0.0131  \\
				RMSE  & 0.0293  & 0.0202  & 0.0268  & 0.0191  & 0.0263  & 0.0191  &  & 0.0259  & 0.0204  & 0.0231  & 0.0213  & 0.0231  & 0.0211  \\
				CP    & 0.8220  & 0.9240  & 0.8600  & 0.9060  & 0.7070  & 0.9040  &       & 0.8410  & 0.9100  & 0.8660  & 0.8710  & 0.7490  & 0.8680  \\
				\cmidrule{2-7}\cmidrule{9-14} 
				& \multicolumn{6}{c}{$(n,T,\elln)=(200,25,2)$}                         &       & \multicolumn{6}{c}{$(n,T,\elln)=(200,25,[n^{1/5}]+2)$} \\
				\cmidrule{2-7}\cmidrule{9-14} 
				Bias  & -0.0080  & -0.0099  & -0.0052  & -0.0099  & -0.0052  & -0.0077  &   & -0.0044  & -0.0110  & -0.0015  & -0.0146  & -0.0015  & -0.0146  \\
				RMSE  & 0.0151  & 0.0163  & 0.0138  & 0.0163  & 0.0130  & 0.0149  &  & 0.0149  & 0.0164  & 0.0143  & 0.0190  & 0.0132  & 0.0186  \\
				CP    & 0.8040  & 0.8610  & 0.8110  & 0.8050  & 0.8290  & 0.8040  &       & 0.8290  & 0.8210  & 0.8310  & 0.7080  & 0.8390  & 0.7060  \\
				\bottomrule
			\end{tabular}%
			\hspace*{-1cm}
			\begin{tablenotes}
				\footnotesize
				\item \textbf{Note:} The results are based on 1,000 Monte Carlo replications. Bias denotes the mean bias of the estimates, RMSE denotes the root mean squared error, and CP denotes the 95\% coverage probability. The 2SLS refers to the two-stage least square estimator, OGMM refers to the feasible optimal GMM estimator, and BGMM refers to the feasible best GMM estimator. 
			\end{tablenotes}
		\end{spacing}
		\label{sim:misspecification}%
	\end{table}%

	\section{An empirical application}\label{sec:witch}
	\cite{miguel2005poverty} empirically finds that exogenous extreme rainfall causes lower income in rural Tanzania, and leads to a large increase in the murder of `witches', who are nearly all elderly women. \cite{miguel2005poverty} further separates this income shock hypothesis from non-economic factors, namely the scapegoat culture hypothesis that families under exogenous extreme rainfall or disease epidemic shocks treat witches as scapegoats and increase their vilification and killing. The results show that income shocks are the driver of the increase in witch killings, instead of this scapegoat culture. Based on \cite{miguel2005poverty}, we apply our proposed model to investigate the mechanism of spatial spillover effects regarding witch killings in rural Tanzania.
	This analysis involves a dataset from 67 villages over the period from 1992 to 2002. We plot the specific location of these villages in Figure \ref{emp:fig_ploty} and it illustrates clear clustering within these communities. 
	\par 
	The clustering pattern of the data motivates our study, in which we conjecture that geographic closeness leads to parallel behavioral patterns in witch killings. For example, when a village experiences severe weather or unexpected illnesses, it could lead to a rise in witch killings, an action that nearby communities might mimic. While it is natural to construct analyses using matrices such as the $k$-nearest neighbor matrix or inverse distance geographical matrix, these are limited by their specific forms. Therefore, we consider unknown weights $G_{1}, G_{2}, G_{3}$, which are functions of the geographical distance $d_{ij}$. We then approximate using a method similar to that in Section \ref{sec:mc}, where $\ell_{k}=[n^{1/5}]+2$, $\bar{d}_{0}$ is set to the 10\%, 25\%, and 50\% quantiles of $d_{ij}$, respectively. 
	Specifically, our empirical models are
	\begin{flalign}
		\tag{EP1}\label{emp:EP1} B_{1}Y_{t}&=(\gamma I_{n}+B_2)Y_{t-1}+X_{1}\beta_{1}+FEs+U_{t},\\
		\tag{EP2}\label{emp:EP2} B_{1}Y_{t}&=(\gamma I_{n}+B_2)Y_{t-1}+X_{1t}\beta_{1}+X_{2t}\beta_{2}+X_{3t}\beta_{3}+FEs+U_{t},\\
		\tag{EP3}\label{emp:EP3} B_{1}Y_{t}&=(\gamma I_{n}+B_2)Y_{t-1}+X_{1t}\beta_{1}+X_{2t}\beta_{2}+X_{3t}\beta_{3}+X_{4t}\beta_{4}+FEs+U_{t},
	\end{flalign}
	with $B_{3}U_{t}=E_{t}$, where the description of variables is shown in Table \ref{tab:des}.
	\par 
	In Table \ref{emp:hetei}, we report the baseline regression results using the number of witch murders as the primary dependent variable under the case of \ref{var:hetei}. To ensure the reliability of our estimates, we consider both the MESS and SAR specifications. The robustness checks for the MESS are presented in Table \ref{emp:robust}, and the corresponding estimation results for the SAR are presented in Table \ref{emp:hetei-sar}. These results are qualitatively similar to our baseline findings, confirming the robustness of our framework.\footnote{The Monte Carlo results indicate that the estimator performs well in the \ref{var:hetei} case with $T=10$. This validates our use of the \ref{var:hetei} specification in the empirical application, where the sample size is comparable ($T=11$).}
	\par 
	Our baseline results reveal that as the geographic distance between communities decreases, the absolute value of estimated spatial weights increases, as shown in Figure \ref{emp:fig_Ghetei}. This implies that neighboring communities exert a stronger spillover effect on one another's witch killing rates when they are located closer together, indicating stronger community spillovers in crimes. We further examine the mechanism of witch killing spillovers between communities from the economic and cultural perspectives, respectively.%
	\footnote{Because the household survey in \cite{miguel2005poverty} was conducted as a single cross-section, consumption data are only available for the year 2001. We construct our economic distance matrix using these 2001 values. While using end-of-period data introduces potential endogeneity, we follow the methodology of \cite{miguel2005poverty}, who utilizes this same 2001 cross-section to proxy for village characteristics across the entire panel. This approach relies on the reasonable assumption that the relative wealth distances between rural villages are structurally persistent over a single decade, even in the presence of transient weather shocks.}
	\par 
	We evaluate economic distance by measuring the similarity in annual per capita consumption expenditures, defined as $ce_{ij}$, using the 2001 survey data. Specifically, we define the distance as $d_{ij}=d_{ij}^{*}ce_{ij}$, where $ce_{ij}=1+|ce_{i,2001}-ce_{j,2001}|,$
	and $d_{ij}^{*}$ is the geographical distance mentioned earlier, denoted with a superscript asterisk.
	We find that as the geographic distance between communities decreases, and if the similarity in annual per capita consumption expenditures increases, which indicates a smaller economic difference, the absolute value of estimated spatial weights rises, as shown in Figure \ref{emp:fig_Ghatecon}. The mean coefficient of $\hat{G}_{1}$ in this specification ranges from $-0.2$ to $-0.25$ and is significantly larger in magnitude than the baseline result of around $-0.08$ to $-0.1$. This finding highlights that economic similarity is a key driver of the copycat pattern in witch killings.
	\par
	We then evaluate cultural distance using ethnic shares.\footnote{Although the demographic data were collected during the 2001 survey described in \cite{miguel2005poverty}, village-level ethnic composition is highly time-invariant. Therefore, we use these ethnic shares to construct the cultural distance matrix for the entire sample period, thereby avoiding the endogeneity concerns associated with time-varying economic variables.} 
	Specifically, we define the cultural distance by the similarity in Sukuma ethnic shares, defined as $st_{ij}$, over all $t$, where $st_{ij}=1+|st_{i}-st_{j}|$. We explicitly focus on the Sukuma population because established literature demonstrates their unique relevance to this phenomenon. Specifically, \cite{miguel2005poverty} notes that, according to \cite{mesaki1994}, the ethnically Sukuma regions of western Tanzania historically account for roughly two-thirds of all reported witch killings in the country, even though the Sukuma ethnic group comprises only 12\% of the national population.
	We find that as the geographic distance decreases, if the similarity in Sukuma ethnic shares increases, the absolute value of estimated spatial weights rises. However, the magnitude of the mean coefficient in this specification (around $-0.06$ to $-0.07$) is not statistically larger than the magnitude of the baseline result. In other words, cultural proximity does not significantly amplify spatial dependence in witch killings. These results corroborate \cite{miguel2005poverty}'s conclusions on how economic and cultural factors influence crime rates, extending the framework by modeling the spatial spillovers in witch killings.
	\par
	Specifically, Figure \ref{emp:fig_Ghatecon} illustrates that the estimated spatial weighting functions gradually vanish as the geographical economic distance increases. We note that the estimated elements in $\hat{G}_{1}$ are predominantly negative. Since $e^{-\hat{G}_1}=I_n+\sum_{k=1}^{\infty}{(-1)^k\hat{G}_{1}^k}/{k!}$, negative entries in $\hat{G}_1$ result in positive off-diagonal elements in the inverse spatial filter, indicating a positive reinforcement mechanism between neighbors. We observe that if the entries of $\hat{G}_1$ are negative, the linear term $-\hat{G}_1$ becomes positive. This implies that the spatial multiplier $B_1^{-1}$ contains positive off-diagonal elements, indicating a positive reinforcement mechanism between neighbors. Thus, an increase in witch killings in one village spills over to increase killings in neighboring villages, consistent with the `copycat' hypothesis, even though the underlying weight function $g_1(d)$ takes negative values. This contrasts with standard specifications that typically impose non-negative weights a priori, a flexibility emphasized by \cite{sun2016functional}.
	\par 
	Moreover, we confirm that the stability condition $\rho(\hat{A})<1$ holds, as shown in Table \ref{emp:rhoA}, and proceed to examine the marginal effects, adopting the methodology of \cite{lesage2009introduction} and \cite{elhorst2014spatial}. The short-term and long-term marginal effects with respect to the $k$-th independent variable are evaluated as:
	\begin{flalign*}
		f_{Si}(\beta)=\te_{n,i}\Diag(S^{-1})\te_{n,i}'\beta \quad \text{and} \quad f_{Li}(\beta)=\te_{n,i}\Diag((S-B)^{-1})\te_{n,i}'\beta,
	\end{flalign*}
	respectively, for $i=1,\ldots,n$, where $\te_{n,i}=(0,\ldots,\underbrace{1}_{\text{$i$th}},\ldots,0)'$.
	\par 
	Figure \ref{emp:fig_margin} presents the results for the case where $\bar{d}_{0}=25\%$. The figure reveals that focusing solely on short-term marginal effects underestimates the true magnitude of the coefficients $\beta_{1}$ and $\beta_{4}$. The short-term estimates capture only 50\% to 95\% of the total long-term impact, with the average long-term effect being significantly larger. This discrepancy suggests that spatial spillovers driven by geographic-economic proximity amplify the initial shocks over time, resulting in persistent and long-lasting dynamics.
	
	\begin{figure}[htbp]
		\begin{minipage}[b]{0.58\textwidth}
			\centering\small\stla{0cm} \stl{0.8mm} 
			\captionof{table}{Data descriptions}
			\begin{spacing}{0.8}
				\begin{tabular}{ll >{\raggedright\arraybackslash}p{4.5cm} rr}
					\toprule
					\multicolumn{2}{c}{Variable}    & Definitions & Mean & ESD\\
					\midrule
					\multirow{4}[2]{*}{$Y_{t}$} & Baseline   & Witch murders       & 0.09  & 0.32  \\
					& $\mathrm{Robust}_1$   & Witch murders per 1000 households        & 0.27  & 1.06  \\
					& $\mathrm{Robust}_2$   & Witch murders and attacks per 1000 households       & 0.56  & 1.86  \\
					& $\mathrm{Robust}_3$   & Total murders per 1000 households       & 0.50  & 1.56  \\
					\midrule
					\multirow{4}[2]{*}{$X_{t}$} & $X_{1t}$ & Extreme rainfall (drought or flood)       & 0.18  & 0.38  \\
					& $X_{2t}$ & Extreme rainfall, previous year        & 0.17  & 0.37  \\
					& $X_{3t}$ & Extreme rainfall, current year and previous year        & 0.08  & 0.27  \\
					& $X_{4t}$ & Human disease epidemic       & 0.15  & 0.36  \\
					\bottomrule
				\end{tabular}%
			\end{spacing}
			\label{tab:des}%
		\end{minipage}%
		\hfill
		\begin{minipage}[b]{0.45\textwidth}
			\centering
			\includegraphics[width=5.5cm, height=4.5cm]{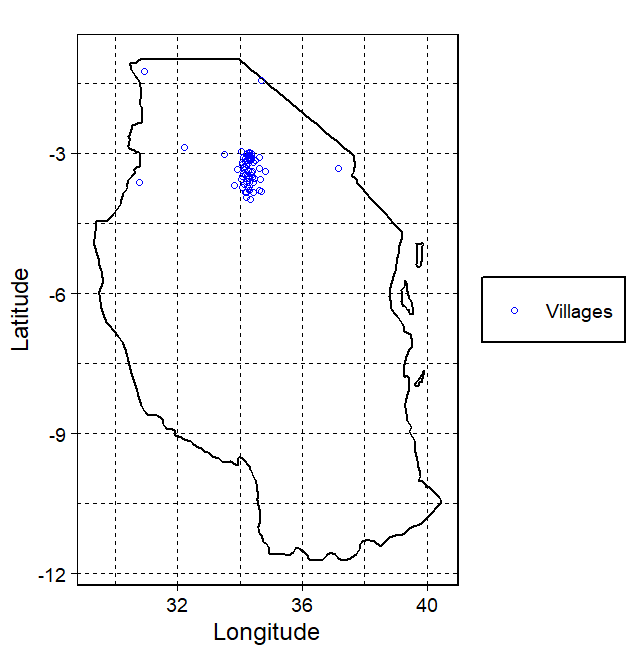}
			\caption{Scatter plot of sample villages in Tanzania}
			\label{emp:fig_ploty}
		\end{minipage}
	\end{figure}

	\begin{table}[htbp]
		\centering\small\stla{0cm} \stl{0.8mm} 
		\caption{Estimation results for the MESS, $\Var(\epsilon_{it})=\sigma_i^2$.}
		\begin{spacing}{0.7}
			\begin{tabular}{lcccccccccccc}
				\toprule
				\multicolumn{2}{c}{\multirow{2}[3]{*}{}} & 2SLS & OGMM & BGMM  &       & 2SLS & OGMM & BGMM  &       & 2SLS & OGMM & BGMM \\
				\cmidrule{3-5}\cmidrule{7-9}\cmidrule{11-13}    \multicolumn{2}{c}{} & \multicolumn{3}{c}{$\bar{d}_{0}=10\%$} &       & \multicolumn{3}{c}{$\bar{d}_{0}=25\%$} &       & \multicolumn{3}{c}{$\bar{d}_{0}=50\%$} \\
				\cmidrule{3-5}\cmidrule{7-9}\cmidrule{11-13} 
				\multicolumn{1}{l}{\multirow{2}[1]{*}{\eqref{emp:EP1}}} & \multicolumn{1}{l}{$\hat{\beta}_{1}$} & \multicolumn{1}{l}{0.0640$^{**}$} & \multicolumn{1}{l}{0.0540$^{*}$} & \multicolumn{1}{l}{0.0740$^{**}$} &       & \multicolumn{1}{l}{0.0540 } & \multicolumn{1}{l}{0.0240 } & \multicolumn{1}{l}{0.0340 } &       & \multicolumn{1}{l}{0.0360 } & \multicolumn{1}{l}{0.0660$^{*}$} & \multicolumn{1}{l}{0.0860$^{**}$} \\
				&       & \multicolumn{1}{l}{(0.0314)} & \multicolumn{1}{l}{(0.0292)} & \multicolumn{1}{l}{(0.0323)} &       & \multicolumn{1}{l}{(0.0334)} & \multicolumn{1}{l}{(0.0329)} & \multicolumn{1}{l}{(0.0463)} &       & \multicolumn{1}{l}{(0.0361)} & \multicolumn{1}{l}{(0.0335)} & \multicolumn{1}{l}{(0.0388)} \\
				\cmidrule{3-5}\cmidrule{7-9}\cmidrule{11-13} 
				\multicolumn{1}{l}{\multirow{6}[2]{*}{\eqref{emp:EP2}}} & \multicolumn{1}{l}{$\hat{\beta}_{1}$} & \multicolumn{1}{l}{0.0650$^{*}$} & \multicolumn{1}{l}{0.0510$^{*}$} & \multicolumn{1}{l}{0.0710$^{**}$} &       & \multicolumn{1}{l}{0.0590 } & \multicolumn{1}{l}{0.0890$^{**}$} & \multicolumn{1}{l}{0.0790$^{*}$} &       & \multicolumn{1}{l}{0.0400 } & \multicolumn{1}{l}{0.0700$^{*}$} & \multicolumn{1}{l}{0.0900$^{*}$} \\
				&       & \multicolumn{1}{l}{(0.0382)} & \multicolumn{1}{l}{(0.0310)} & \multicolumn{1}{l}{(0.0357)} &       & \multicolumn{1}{l}{(0.0409)} & \multicolumn{1}{l}{(0.0425)} & \multicolumn{1}{l}{(0.0419)} &       & \multicolumn{1}{l}{(0.0424)} & \multicolumn{1}{l}{(0.0405)} & \multicolumn{1}{l}{(0.0463)} \\
				& \multicolumn{1}{l}{$\hat{\beta}_{2}$} & \multicolumn{1}{l}{0.0490 } & \multicolumn{1}{l}{0.0520 } & \multicolumn{1}{l}{0.0720$^{*}$} &       & \multicolumn{1}{l}{0.0410 } & \multicolumn{1}{l}{0.0110 } & \multicolumn{1}{l}{0.0310 } &       & \multicolumn{1}{l}{0.0280 } & \multicolumn{1}{l}{0.0220 } & \multicolumn{1}{l}{0.0220 } \\
				&       & \multicolumn{1}{l}{(0.0369)} & \multicolumn{1}{l}{(0.0363)} & \multicolumn{1}{l}{(0.0374)} &       & \multicolumn{1}{l}{(0.0371)} & \multicolumn{1}{l}{(0.0365)} & \multicolumn{1}{l}{(0.0373)} &       & \multicolumn{1}{l}{(0.0385)} & \multicolumn{1}{l}{(0.0374)} & \multicolumn{1}{l}{(0.0381)} \\
				& \multicolumn{1}{l}{$\hat{\beta}_{3}$} & \multicolumn{1}{l}{-0.0120 } & \multicolumn{1}{l}{-0.0120 } & \multicolumn{1}{l}{-0.0080 } &       & \multicolumn{1}{l}{-0.0280 } & \multicolumn{1}{l}{-0.0580 } & \multicolumn{1}{l}{-0.0780 } &       & \multicolumn{1}{l}{-0.0070 } & \multicolumn{1}{l}{-0.0230 } & \multicolumn{1}{l}{-0.0430 } \\
				&       & \multicolumn{1}{l}{(0.0709)} & \multicolumn{1}{l}{(0.0650)} & \multicolumn{1}{l}{(0.0652)} &       & \multicolumn{1}{l}{(0.0658)} & \multicolumn{1}{l}{(0.0665)} & \multicolumn{1}{l}{(0.0984)} &       & \multicolumn{1}{l}{(0.0672)} & \multicolumn{1}{l}{(0.0658)} & \multicolumn{1}{l}{(0.0700)} \\
				\cmidrule{3-5}\cmidrule{7-9}\cmidrule{11-13} 
				\multicolumn{1}{l}{\multirow{8}[2]{*}{\eqref{emp:EP3}}} & \multicolumn{1}{l}{$\hat{\beta}_{1}$} & \multicolumn{1}{l}{0.0600$^{*}$} & \multicolumn{1}{l}{0.0500 } & \multicolumn{1}{l}{0.0540 } &       & \multicolumn{1}{l}{0.0560$^{*}$} & \multicolumn{1}{l}{0.0860$^{**}$} & \multicolumn{1}{l}{0.0660$^{*}$} &       & \multicolumn{1}{l}{0.0370 } & \multicolumn{1}{l}{0.0670$^{*}$} & \multicolumn{1}{l}{0.0870$^{**}$} \\
				&       & \multicolumn{1}{l}{(0.0360)} & \multicolumn{1}{l}{(0.0378)} & \multicolumn{1}{l}{(0.0392)} &       & \multicolumn{1}{l}{(0.0330)} & \multicolumn{1}{l}{(0.0429)} & \multicolumn{1}{l}{(0.0360)} &       & \multicolumn{1}{l}{(0.0415)} & \multicolumn{1}{l}{(0.0400)} & \multicolumn{1}{l}{(0.0418)} \\
				& \multicolumn{1}{l}{$\hat{\beta}_{2}$} & \multicolumn{1}{l}{0.0460 } & \multicolumn{1}{l}{0.0500 } & \multicolumn{1}{l}{0.0570 } &       & \multicolumn{1}{l}{0.0390 } & \multicolumn{1}{l}{0.0090 } & \multicolumn{1}{l}{0.0290 } &       & \multicolumn{1}{l}{0.0270 } & \multicolumn{1}{l}{0.0260 } & \multicolumn{1}{l}{0.0170 } \\
				&       & \multicolumn{1}{l}{(0.0362)} & \multicolumn{1}{l}{(0.0358)} & \multicolumn{1}{l}{(0.0372)} &       & \multicolumn{1}{l}{(0.0369)} & \multicolumn{1}{l}{(0.0364)} & \multicolumn{1}{l}{(0.0369)} &       & \multicolumn{1}{l}{(0.0382)} & \multicolumn{1}{l}{(0.0373)} & \multicolumn{1}{l}{(0.0382)} \\
				& \multicolumn{1}{l}{$\hat{\beta}_{3}$} & \multicolumn{1}{l}{-0.0020 } & \multicolumn{1}{l}{-0.0010 } & \multicolumn{1}{l}{-0.0020 } &       & \multicolumn{1}{l}{-0.0150 } & \multicolumn{1}{l}{-0.0450 } & \multicolumn{1}{l}{-0.0280 } &       & \multicolumn{1}{l}{-0.0060 } & \multicolumn{1}{l}{-0.0360 } & \multicolumn{1}{l}{-0.0560 } \\
				&       & \multicolumn{1}{l}{(0.0649)} & \multicolumn{1}{l}{(0.0640)} & \multicolumn{1}{l}{(0.0065)} &       & \multicolumn{1}{l}{(0.0661)} & \multicolumn{1}{l}{(0.0664)} & \multicolumn{1}{l}{(0.0889)} &       & \multicolumn{1}{l}{(0.0663)} & \multicolumn{1}{l}{(0.0652)} & \multicolumn{1}{l}{(0.0650)} \\
				& \multicolumn{1}{l}{$\hat{\beta}_{4}$} & \multicolumn{1}{l}{0.0870$^{**}$} & \multicolumn{1}{l}{0.0760$^{*}$} & \multicolumn{1}{l}{0.0820$^{**}$} &       & \multicolumn{1}{l}{0.0790$^{*}$} & \multicolumn{1}{l}{0.1090$^{**}$} & \multicolumn{1}{l}{0.0890$^{*}$} &       & \multicolumn{1}{l}{0.1030$^{**}$} & \multicolumn{1}{l}{0.1330$^{***}$} & \multicolumn{1}{l}{0.1530$^{***}$} \\
				&       & \multicolumn{1}{l}{(0.0423)} & \multicolumn{1}{l}{(0.0402)} & \multicolumn{1}{l}{(0.0412)} &       & \multicolumn{1}{l}{(0.0436)} & \multicolumn{1}{l}{(0.0431)} & \multicolumn{1}{l}{(0.0477)} &       & \multicolumn{1}{l}{(0.0444)} & \multicolumn{1}{l}{(0.0434)} & \multicolumn{1}{l}{(0.0465)} \\
				\midrule
				\multicolumn{1}{l}{Village} &\multicolumn{1}{l}{FEs}   & \multicolumn{11}{c}{YES} \\
				\multicolumn{1}{l}{Time}  & \multicolumn{1}{l}{FEs}   & \multicolumn{11}{c}{YES} \\
				\bottomrule
			\end{tabular}%
			\hspace*{-0.2cm}
			\begin{tablenotes}\footnotesize
				\item 
				\textbf{Notes}: The parameter $\hat{\beta}_1$ denotes the coefficient of extreme rainfall, $\hat{\beta}_2$ denotes the coefficient of the previous year of extreme rainfall, $\hat{\beta}_3$ denotes the coefficient of the current and previous year of extreme rainfall, and $\hat{\beta}_4$ denotes the coefficient of the human disease epidemic. 
				The term 2SLS refers to the 2SLS estimator, OGMM refers to the feasible optimal GMM estimator, and BGMM refers to the feasible best GMM estimator. 
				The $^{*}$ denotes $p<0.1$, $^{**}$ denotes $p<0.05$ and $^{***}$ denotes $p<0.01$. Theoretical standard deviations are reported in parentheses. 
			\end{tablenotes}
		\end{spacing}
		\label{emp:hetei}%
	\end{table}%
	
	\begin{table}[htbp]
		\centering\small\stla{0cm} \stl{0.8mm} 
		\caption{Estimation results for the MESS robustness check.}
		\begin{spacing}{0.8}
			\begin{tabular}{llccccccccccc}
				\toprule
				\multicolumn{2}{c}{\multirow{2}[4]{*}{}} & 2SLS & OGMM & BGMM  &       & 2SLS & OGMM & BGMM  &       & 2SLS & OGMM & BGMM \\
				\cmidrule{3-5}\cmidrule{7-9}\cmidrule{11-13}    \multicolumn{2}{c}{} & \multicolumn{3}{c}{$\bar{d}_0=10\%$} &       & \multicolumn{3}{c}{$\bar{d}_0=25\%$} &       & \multicolumn{3}{c}{$\bar{d}_0=50\%$} \\
				\cmidrule{3-5}\cmidrule{7-9}\cmidrule{11-13} 
				\multirow{4}[2]{*}{$\mathrm{Robust}_1$} & $\hat{\beta}_{1}$ & \multicolumn{1}{l}{0.1070$^{*}$} & \multicolumn{1}{l}{0.0900$^{*}$} & \multicolumn{1}{l}{0.1100$^{*}$} &       & \multicolumn{1}{l}{0.0880 } & \multicolumn{1}{l}{0.1180$^{*}$} & \multicolumn{1}{l}{0.0980$^{*}$} &       & \multicolumn{1}{l}{0.0420 } & \multicolumn{1}{l}{0.0720$^{*}$} & \multicolumn{1}{l}{0.0920$^{*}$} \\
				&       & \multicolumn{1}{l}{(0.0617)} & \multicolumn{1}{l}{(0.0499)} & \multicolumn{1}{l}{(0.0610)} &       & \multicolumn{1}{l}{(0.1198)} & \multicolumn{1}{l}{(0.0709)} & \multicolumn{1}{l}{(0.0533)} &       & \multicolumn{1}{l}{(0.0634)} & \multicolumn{1}{l}{(0.0435)} & \multicolumn{1}{l}{(0.0558)} \\
				& $\hat{\beta}_{4}$  & \multicolumn{1}{l}{0.1220 } & \multicolumn{1}{l}{0.1020 } & \multicolumn{1}{l}{0.1220 } &       & \multicolumn{1}{l}{0.1200 } & \multicolumn{1}{l}{0.1500$^{**}$} & \multicolumn{1}{l}{0.1300$^{*}$} &       & \multicolumn{1}{l}{0.1410$^{*}$} & \multicolumn{1}{l}{0.1710$^{**}$} & \multicolumn{1}{l}{0.1910$^{**}$} \\
				&       & \multicolumn{1}{l}{(0.0756)} & \multicolumn{1}{l}{(0.0754)} & \multicolumn{1}{l}{(0.0747)} &       & \multicolumn{1}{l}{(0.0820)} & \multicolumn{1}{l}{(0.0759)} & \multicolumn{1}{l}{(0.0774)} &       & \multicolumn{1}{l}{(0.0815)} & \multicolumn{1}{l}{(0.0847)} & \multicolumn{1}{l}{(0.0831)} \\
				\cmidrule{3-5}\cmidrule{7-9}\cmidrule{11-13}    \multirow{4}[2]{*}{$\mathrm{Robust}_2$} & $\hat{\beta}_{1}$  & \multicolumn{1}{l}{0.0470 } & \multicolumn{1}{l}{0.0500$^{*}$} & \multicolumn{1}{l}{0.0700$^{*}$} &       & \multicolumn{1}{l}{0.0280 } & \multicolumn{1}{l}{0.0200 } & \multicolumn{1}{l}{0.0290 } &       & \multicolumn{1}{l}{0.0320$^{*}$} & \multicolumn{1}{l}{0.0620$^{*}$} & \multicolumn{1}{l}{0.0420$^{*}$} \\
				&       & \multicolumn{1}{l}{(0.0828)} & \multicolumn{1}{l}{(0.0297)} & \multicolumn{1}{l}{(0.0378)} &       & \multicolumn{1}{l}{(0.0893)} & \multicolumn{1}{l}{(0.8679)} & \multicolumn{1}{l}{(0.0232)} &       & \multicolumn{1}{l}{(0.0183)} & \multicolumn{1}{l}{(0.0365)} & \multicolumn{1}{l}{(0.0245)} \\
				& $\hat{\beta}_{4}$ & \multicolumn{1}{l}{0.1100 } & \multicolumn{1}{l}{0.0940$^{*}$} & \multicolumn{1}{l}{0.0740$^{*}$} &       & \multicolumn{1}{l}{0.0900 } & \multicolumn{1}{l}{0.0600 } & \multicolumn{1}{l}{0.0800 } &       & \multicolumn{1}{l}{0.1390 } & \multicolumn{1}{l}{0.1090 } & \multicolumn{1}{l}{0.0890 } \\
				&       & \multicolumn{1}{l}{(0.1066)} & \multicolumn{1}{l}{(0.0563)} & \multicolumn{1}{l}{(0.0431)} &       & \multicolumn{1}{l}{(0.1118)} & \multicolumn{1}{l}{(0.1118)} & \multicolumn{1}{l}{(0.1098)} &       & \multicolumn{1}{l}{(0.1212)} & \multicolumn{1}{l}{(0.1109)} & \multicolumn{1}{l}{(0.1099)} \\
				\cmidrule{3-5}\cmidrule{7-9}\cmidrule{11-13}   \multirow{4}[2]{*}{$\mathrm{Robust}_3$} & $\hat{\beta}_{1}$  & \multicolumn{1}{l}{0.0680 } & \multicolumn{1}{l}{0.0660 } & \multicolumn{1}{l}{0.0860 } &       & \multicolumn{1}{l}{0.0110 } & \multicolumn{1}{l}{-0.0190 } & \multicolumn{1}{l}{-0.0390 } &       & \multicolumn{1}{l}{0.0200 } & \multicolumn{1}{l}{0.0500 } & \multicolumn{1}{l}{0.0700 } \\
				&       & \multicolumn{1}{l}{(0.0904)} & \multicolumn{1}{l}{(0.0865)} & \multicolumn{1}{l}{(0.0819)} &       & \multicolumn{1}{l}{(0.0977)} & \multicolumn{1}{l}{(0.1159)} & \multicolumn{1}{l}{(0.1441)} &       & \multicolumn{1}{l}{(0.0850)} & \multicolumn{1}{l}{(0.0845)} & \multicolumn{1}{l}{(0.0811)} \\
				& $\hat{\beta}_{4}$  & \multicolumn{1}{l}{0.1380} & \multicolumn{1}{l}{0.1480 } & \multicolumn{1}{l}{0.1680 } &       & \multicolumn{1}{l}{0.1940 } & \multicolumn{1}{l}{0.1640 } & \multicolumn{1}{l}{0.1840 } &       & \multicolumn{1}{l}{0.1190} & \multicolumn{1}{l}{0.0830} & \multicolumn{1}{l}{0.0960} \\
				&       & \multicolumn{1}{l}{(0.1017)} & \multicolumn{1}{l}{(0.1024)} & \multicolumn{1}{l}{(0.1030)} &       & \multicolumn{1}{l}{(0.1534)} & \multicolumn{1}{l}{(0.1163)} & \multicolumn{1}{l}{(0.1693)} &       & \multicolumn{1}{l}{(0.1180)} & \multicolumn{1}{l}{(0.1038)} & \multicolumn{1}{l}{(0.0986)} \\
				\midrule
				\multicolumn{1}{l}{Village} &\multicolumn{1}{l}{FEs}   & \multicolumn{11}{c}{YES} \\
				\multicolumn{1}{l}{Time}  & \multicolumn{1}{l}{FEs}   & \multicolumn{11}{c}{YES} \\
				\bottomrule
			\end{tabular}%
			\hspace*{-0.2cm}
			\begin{tablenotes}\footnotesize
				\item 
				\textbf{Notes}: In $\mathrm{Robust}_1$, $Y_{t}$ denotes witch murders per 1000 households. In $\mathrm{Robust}_2$, $Y_{t}$ denotes witch murders and attacks per 1000 households, and in $\mathrm{Robust}_3$, $Y_{t}$ denotes total murders per 1000 households. The parameter $\hat{\beta}_1$ denotes the coefficient of extreme rainfall, and $\hat{\beta}_4$ denotes the coefficient of the human disease epidemic. 
				The term 2SLS refers to the two-stage least squares estimator, OGMM refers to the feasible optimal GMM estimator, and BGMM refers to the feasible best GMM estimator. 
				The $^{*}$ denotes $p<0.1$, $^{**}$ denotes $p<0.05$ and $^{***}$ denotes $p<0.01$. Theoretical standard deviations are reported in parentheses. 
				We observe that the regressors have no significant impact on total murders under the $\mathrm{Robust}_3$ specification, and this is consistent with \cite{miguel2005poverty}.
			\end{tablenotes}
		\end{spacing}
		\label{emp:robust}%
	\end{table}%

	\begin{table}[htbp]
		\centering\small\stla{0cm} \stl{0.8mm} 
		\caption{Estimation results for the SAR, $\Var(\epsilon_{it})=\sigma_i^2$.}
		\begin{spacing}{0.7}
			\begin{tabular}{llccccccccccc}
				\toprule
				\multicolumn{2}{c}{\multirow{2}[3]{*}{}} & 2SLS & OGMM & BGMM  &       & 2SLS & OGMM & BGMM  &       & 2SLS & OGMM & BGMM \\
				\cmidrule{3-5}\cmidrule{7-9}\cmidrule{11-13}    \multicolumn{2}{c}{} & \multicolumn{3}{c}{$\bar{d}_{0}=10\%$} &       & \multicolumn{3}{c}{$\bar{d}_{0}=25\%$} &       & \multicolumn{3}{c}{$\bar{d}_{0}=50\%$} \\
				\midrule
				\multirow{2}[2]{*}{\eqref{emp:EP1}} & \multicolumn{1}{l}{$\hat\beta_1$} & 0.0810$^{*}$ & 0.0880$^{*}$ & 0.0890$^{*}$ &       & 0.0720$^{*}$ & 0.0730$^{*}$ & 0.0760$^{*}$ &       & 0.0930$^{*}$ & 0.0830$^{*}$ & 0.0820$^{*}$ \\
				&       & (0.0427) & (0.0460) & (0.0503) &       & (0.0409) & (0.0418) & (0.0403) &       & (0.0516) & (0.0478) & (0.0469) \\
				\cmidrule{3-5}\cmidrule{7-9}\cmidrule{11-13} 
				\multirow{6}[2]{*}{\eqref{emp:EP2}} & \multicolumn{1}{l}{$\hat\beta_1$} & 0.0770$^{*}$ & 0.0800$^{*}$ & 0.0860$^{*}$ &       & 0.0660$^{*}$ & 0.0690$^{*}$ & 0.0670$^{*}$ &       & 0.0920$^{*}$ & 0.0830$^{*}$ & 0.0840$^{*}$ \\
				&       & (0.0395) & (0.0439) & (0.0474) &       & (0.0363) & (0.0406) & (0.0344) &       & (0.0477) & (0.0426) & (0.0436) \\
				& \multicolumn{1}{l}{$\hat\beta_2$} & 0.1050  & 0.0200  & 0.0030  &       & 0.1510$^{**}$ & 0.0950  & 0.0620  &       & 0.1490$^{*}$ & 0.0840  & 0.0420  \\
				&       & (0.0755) & (0.0566) & (0.0635) &       & (0.0741) & (0.0634) & (0.0666) &       & (0.0771) & (0.0702) & (0.0874) \\
				& \multicolumn{1}{l}{$\hat\beta_3$} & -0.0190  & 0.0410  & 0.0270  &       & 0.0230  & -0.0140  & -0.0350  &       & -0.0130  & -0.0240  & -0.0410  \\
				&       & (0.1327) & (0.1026) & (0.1054) &       & (0.1339) & (0.1206) & (0.1157) &       & (0.1343) & (0.1149) & (0.1373) \\
				\cmidrule{3-5}\cmidrule{7-9}\cmidrule{11-13} 
				\multirow{8}[2]{*}{\eqref{emp:EP3}} & \multicolumn{1}{l}{$\hat\beta_1$} & 0.0900$^{*}$ & 0.0890$^{*}$ & 0.0860$^{*}$ &       & 0.0620$^{*}$ & 0.0620$^{*}$ & 0.0550$^{*}$ &       & 0.0960$^{*}$ & 0.0870$^{**}$ & 0.0870$^{*}$ \\
				&       & (0.0521) & (0.0510) & (0.0488) &       & (0.0350) & (0.0319) & (0.0333) &       & (0.0576) & (0.0445) & (0.0490) \\
				& \multicolumn{1}{l}{$\hat\beta_2$} & 0.1160  & 0.1160$^{*}$ & 0.0710  &       & 0.1380$^{*}$ & 0.1260$^{*}$ & 0.0780  &       & 0.1350$^{*}$ & 0.0730  & 0.0720  \\
				&       & (0.0738) & (0.0671) & (0.0634) &       & (0.0740) & (0.0673) & (0.0806) &       & (0.0774) & (0.0677) & (0.0888) \\
				& \multicolumn{1}{l}{$\hat\beta_3$} & -0.0230  & -0.0230  & 0.0000  &       & 0.0620  & -0.0230  & -0.0430  &       & 0.0390  & 0.0050  & 0.0050  \\
				&       & (0.1342) & (0.1255) & (0.1156) &       & (0.1342) & (0.1123) & (0.1418) &       & (0.1353) & (0.1149) & (0.2274) \\
				& \multicolumn{1}{l}{$\hat\beta_4$} & 0.1230$^{*}$ & 0.1230$^{*}$ & 1.1930$^{*}$ &       & 0.1040$^{*}$ & 0.1330$^{*}$ & 0.1710$^{*}$ &       & 0.1170$^{*}$ & 0.0850$^{*}$ & 0.0850$^{*}$ \\
				&       & (0.0739) & (0.0732) & (0.7025) &       & (0.0540) & (0.0746) & (0.0998) &       & (0.0608) & (0.0493) & (0.0485) \\
				\midrule
				\multicolumn{1}{l}{Village} &\multicolumn{1}{l}{FEs}   & \multicolumn{11}{c}{YES} \\
				\multicolumn{1}{l}{Time}  & \multicolumn{1}{l}{FEs}   & \multicolumn{11}{c}{YES} \\
				%    $nT$ & & \multicolumn{11}{c}{737} \\
				\bottomrule
			\end{tabular}%
			\hspace*{-0.2cm}
			\begin{tablenotes}\footnotesize
				\item 
				\textbf{Notes}: The parameter $\hat{\beta}_1$ denotes the coefficient of extreme rainfall, $\hat{\beta}_2$ denotes the coefficient of the previous year of extreme rainfall, $\hat{\beta}_3$ denotes the coefficient of the current and previous year of extreme rainfall, and $\hat{\beta}_4$ denotes the coefficient of the human disease epidemic. 
				The term 2SLS refers to the two-stage least squares estimator, OGMM refers to the feasible optimal GMM estimator, and BGMM refers to the feasible best GMM estimator. 
				The $^{*}$ denotes $p<0.1$, $^{**}$ denotes $p<0.05$ and $^{***}$ denotes $p<0.01$. Theoretical standard deviations are reported in parentheses. 
			\end{tablenotes}
		\end{spacing}
		\label{emp:hetei-sar}%
	\end{table}%
	
	\begin{table}[htbp]
		\centering\small\stla{0cm} \stl{0.8mm}
		\caption{Estimation results for the MESS, $\rho(\hat{A})$, in the cases of $\mathrm{Var}(\epsilon_{it})=\sigma_i^2$, $\sigma_t^2$ or $\sigma^2$.}
		\begin{spacing}{0.8}
			\hspace*{-8mm}%
			\begin{minipage}{1.1\textwidth} 
				\begin{tabular}{llrrrllrrrllrrr}
					\toprule
					\multicolumn{2}{l}{\ref{var:hetei}, $\bar{d}_{0}$} & 2SLS & OGMM & BGMM  & \multicolumn{2}{l}{\ref{var:hetet}, $\bar{d}_{0}$} & 2SLS & OGMM & BGMM  & \multicolumn{2}{l}{\ref{var:homo}, $\bar{d}_{0}$} & 2SLS & OGMM & BGMM \\
					\cmidrule{3-5}\cmidrule{8-10}\cmidrule{13-15} 
					\multirow{3}[2]{*}{$10\%$} & \eqref{emp:EP1}    & 0.88487  & 0.88470  & 0.88270  & \multirow{3}[2]{*}{$10\%$} & \eqref{emp:EP1}    & 0.88485  & 0.88486  & 0.88678  & \multirow{3}[2]{*}{$10\%$} & \eqref{emp:EP1}    & 0.88486  & 0.88479  & 0.88260  \\
					& \eqref{emp:EP2}    & 0.88434  & 0.88259  & 0.88051  &       & \eqref{emp:EP2}    & 0.88433  & 0.88336  & 0.88546  &       & \eqref{emp:EP2}    & 0.88427  & 0.88327  & 0.88533  \\
					& \eqref{emp:EP3}    & 0.88330  & 0.88171  & 0.87959  &       & \eqref{emp:EP3}    & 0.88326  & 0.88252  & 0.88465  &       & \eqref{emp:EP3}    & 0.88330  & 0.88251  & 0.88458  \\
					\cmidrule{3-5}\cmidrule{8-10}\cmidrule{13-15} 
					\multirow{3}[2]{*}{$25\%$} & \eqref{emp:EP1}    & 0.93905  & 0.93790  & 0.93709  & \multirow{3}[2]{*}{$25\%$} & \eqref{emp:EP1}    & 0.93892  & 0.93763  & 0.93683  & \multirow{3}[2]{*}{$25\%$} & \eqref{emp:EP1}    & 0.93893  & 0.93767  & 0.93687  \\
					& \eqref{emp:EP2}    & 0.91590  & 0.91436  & 0.91329  &       & \eqref{emp:EP2}    & 0.91595  & 0.91445  & 0.91339  &       & \eqref{emp:EP2}    & 0.91602  & 0.91454  & 0.91348  \\
					& \eqref{emp:EP3}    & 0.91324  & 0.91165  & 0.91056  &       & \eqref{emp:EP3}    & 0.91327  & 0.91184  & 0.91075  &       & \eqref{emp:EP3}    & 0.91326  & 0.91182  & 0.91094  \\
					\cmidrule{3-5}\cmidrule{8-10}\cmidrule{13-15} 
					\multirow{3}[2]{*}{$50\%$} & \eqref{emp:EP1}    & 0.95968  & 0.95938  & 0.95917  & \multirow{3}[2]{*}{$50\%$} & \eqref{emp:EP1}    & 0.95965  & 0.95935  & 0.95915  & \multirow{3}[2]{*}{$50\%$} & \eqref{emp:EP1}    & 0.95969  & 0.95939  & 0.95918  \\
					& \eqref{emp:EP2}    & 0.92426  & 0.92373  & 0.92337  &       & \eqref{emp:EP2}    & 0.92431  & 0.92378  & 0.92342  &       & \eqref{emp:EP2}    & 0.92431  & 0.92378  & 0.92342  \\
					& \eqref{emp:EP3}    & 0.92809  & 0.92758  & 0.92723  &       & \eqref{emp:EP3}    & 0.92809  & 0.92758  & 0.92723  &       & \eqref{emp:EP3}    & 0.92809  & 0.92758  & 0.92723  \\
					\bottomrule
				\end{tabular}%            
				\begin{tablenotes}\footnotesize
					\item \textbf{Notes}: The 2SLS refers to the 2SLS estimator, OGMM refers to the feasible optimal GMM estimator, and BGMM refers to the feasible best GMM estimator. The parameter $\hat{\beta}_1$ denotes the coefficient of extreme rainfall, $\hat{\beta}_2$ denotes the coefficient of the previous year of extreme rainfall, $\hat{\beta}_3$ denotes the coefficient of the current and previous year of extreme rainfall, and $\hat{\beta}_4$ denotes the coefficient of the human disease epidemic.
				\end{tablenotes}
			\end{minipage}
		\end{spacing}
		\label{emp:rhoA}%
	\end{table}%
	%%%%%%%%%%%%%%  =========================  %%%%%%%%%%%%%% 
	%%%%%%%%%%%%           Figures    %%%%%%%%%%%% 
	%%%%%%%%%%%%%%  =========================  %%%%%%%%%%%%%% 
	\begin{figure}[htbp]
		\centering
		\begin{subfigure}[t]{0.6\textwidth}\hspace*{-1.5cm}
			\includegraphics[width=\textwidth]{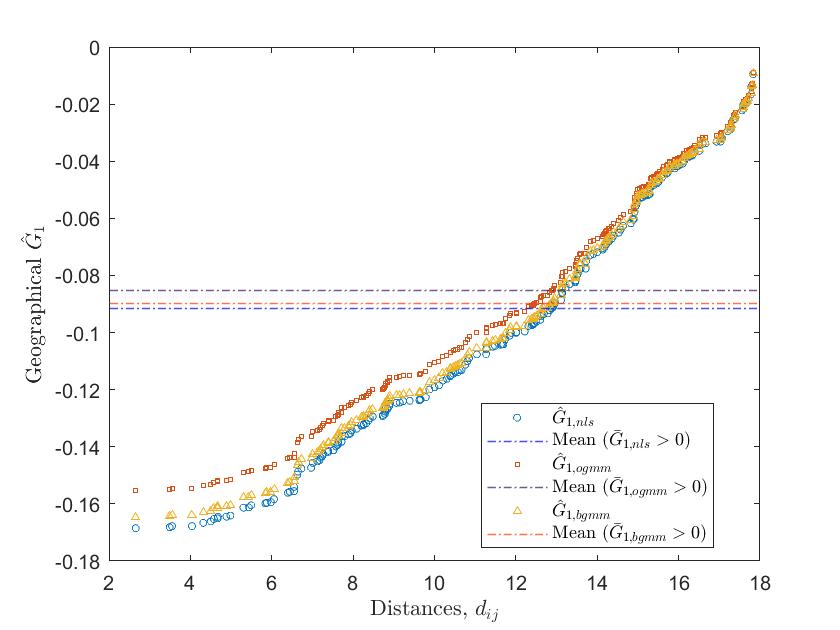}
		\end{subfigure}%\hfill
		\begin{subfigure}[t]{0.6\textwidth}\hspace*{-2cm}
			\includegraphics[width=\textwidth]{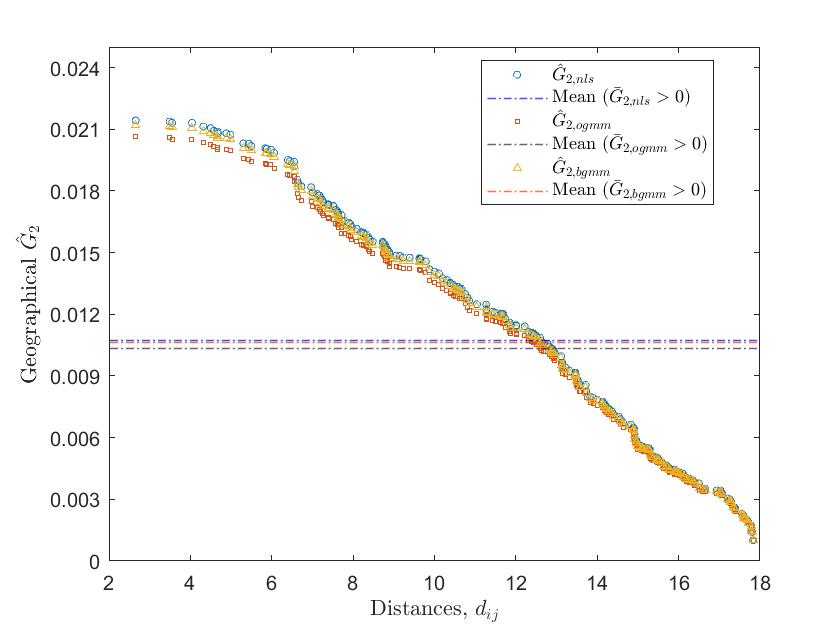}
		\end{subfigure}%\hfill
		\\
		\begin{subfigure}[t]{0.6\textwidth}%\hspace*{-2cm}
			\includegraphics[width=\textwidth]{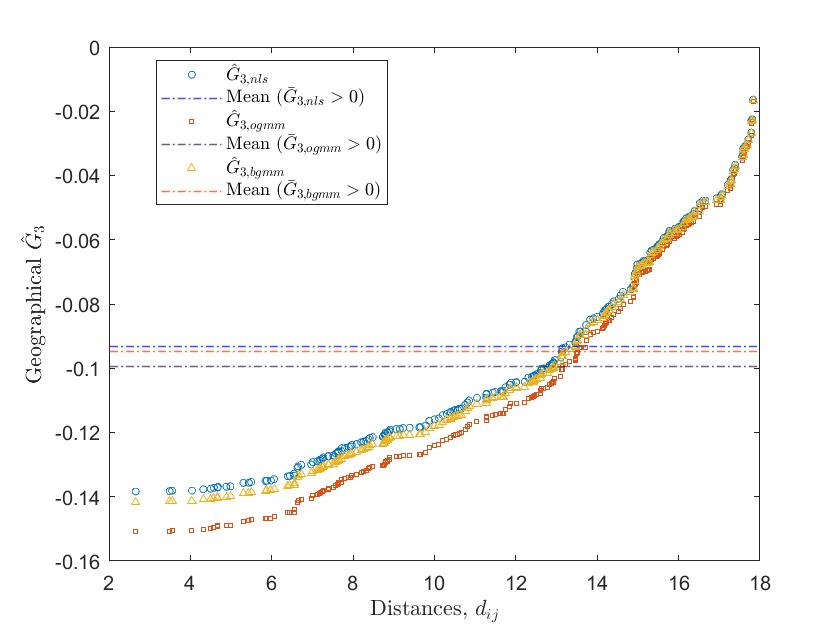}
		\end{subfigure}
		\caption{Estimated pure-geographical $G_{1}$, $G_{2}$ and $G_{3}$ based on three estimators, the case of $\Var(\epsilon_{it})=\sigma_i^2$.}\label{emp:fig_Ghetei}
		\begin{tablenotes}\footnotesize
			\item \textbf{Notes}: For the MESS specification, the negative entries in $\hat{G}_1$ indicate a positive mechanism between neighbors.
		\end{tablenotes}
	\end{figure}
	
	\begin{figure}[htbp]
		\centering
		\begin{subfigure}[t]{0.5\textwidth}%\hspace*{-1cm}
			\includegraphics[width=\textwidth]{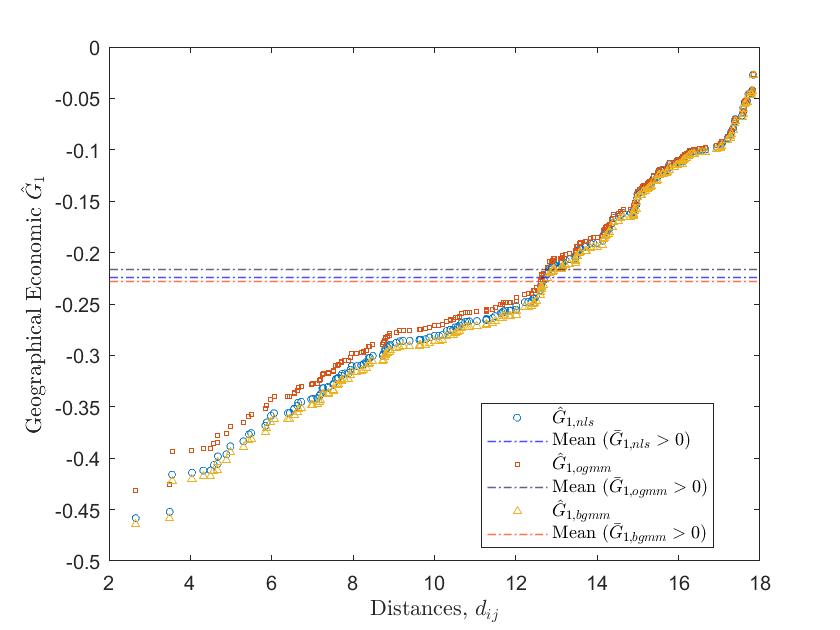}
		\end{subfigure}%\hfill
		\begin{subfigure}[t]{0.5\textwidth}%\hspace*{-1.5cm}
			\includegraphics[width=\textwidth]{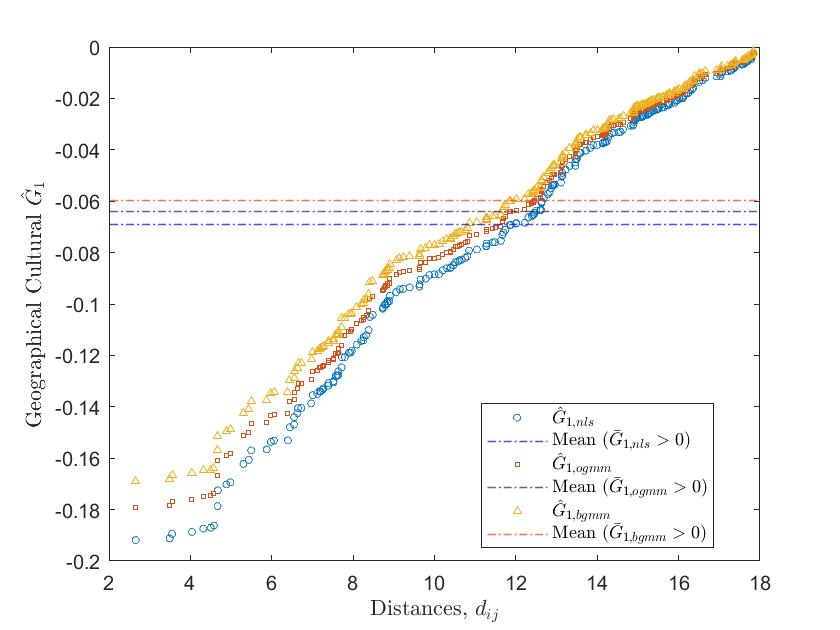}
		\end{subfigure}%\hfill
		\caption{Estimated $G_{1}$ based on economic and cultural geographic distances for the MESS, in the case of $\Var(\epsilon_{it})=\sigma_i^2$.}\label{emp:fig_Ghatecon}
		\begin{tablenotes}\footnotesize
			\item \textbf{Notes}: For the MESS specification, the negative entries in $\hat{G}_1$ indicate a positive mechanism between neighbors. 
		\end{tablenotes}
	\end{figure}
	
	\begin{figure}[htbp]
		\centering
		\begin{subfigure}[t]{0.5\textwidth}%\hspace*{-1cm}
			\includegraphics[width=\textwidth]{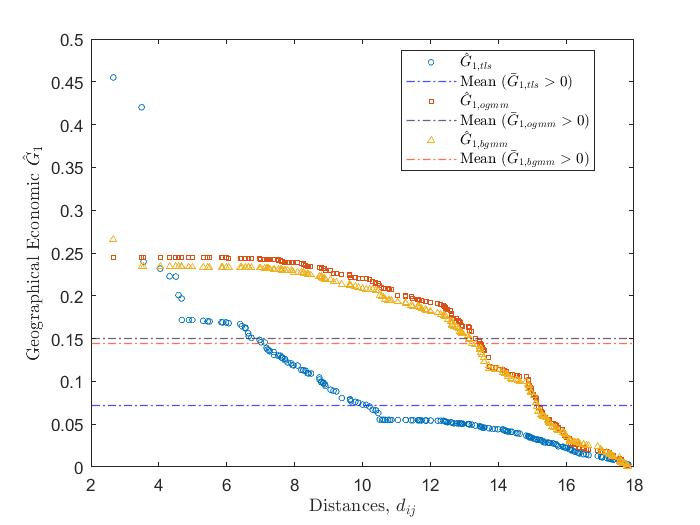}
		\end{subfigure}%\hfill
		\begin{subfigure}[t]{0.5\textwidth}%\hspace*{-1.5cm}
			\includegraphics[width=\textwidth]{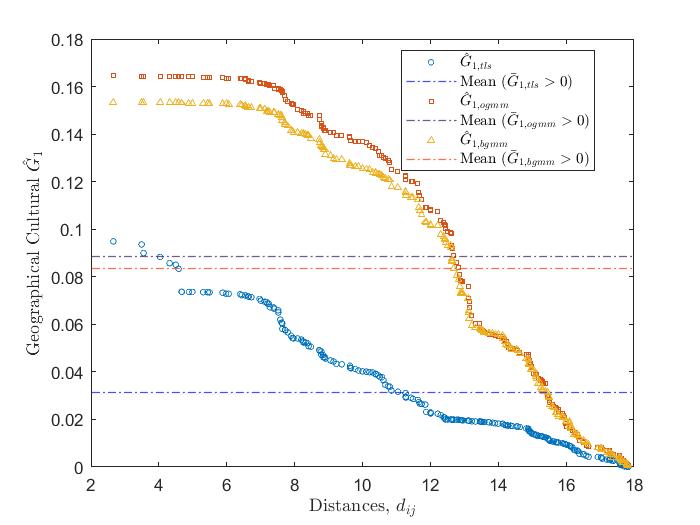}
		\end{subfigure}%\hfill
		\caption{Estimated $G_{1}$ based on economic and cultural geographic distances for the SAR, $\Var(\epsilon_{it})=\sigma_i^2$. }\label{emp:fig_Ghatecon_sar}
	\end{figure}

	\begin{figure}[htbp]
		\centering
		\begin{subfigure}[t]{0.4\textwidth}
			\includegraphics[width=\textwidth]{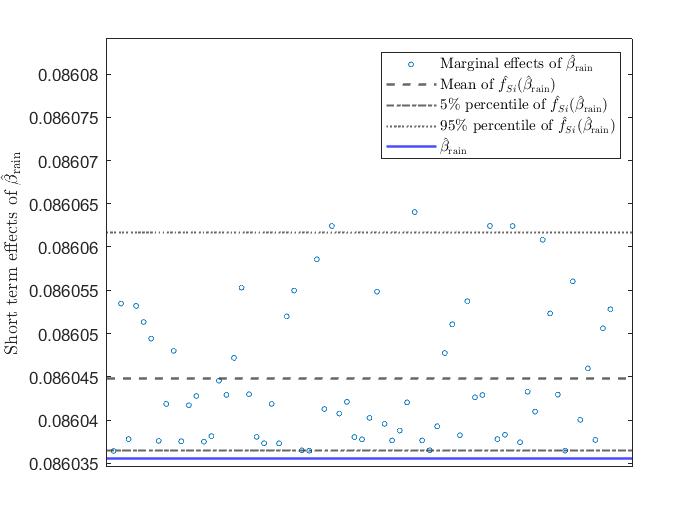}
			\caption{Short-term marginal effects, rainfall}
		\end{subfigure}%
		\begin{subfigure}[t]{0.4\textwidth}
			\includegraphics[width=\textwidth]{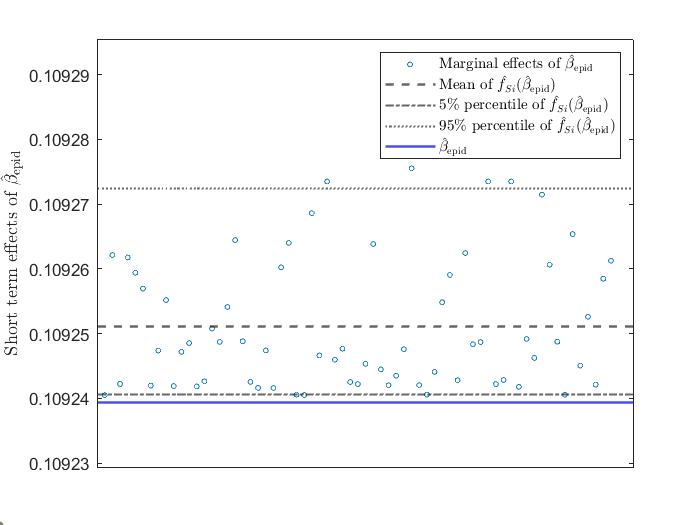}
			\caption{Short-term marginal effects, epidemic}
		\end{subfigure}
		
		\begin{subfigure}[t]{0.4\textwidth}
			\includegraphics[width=\textwidth]{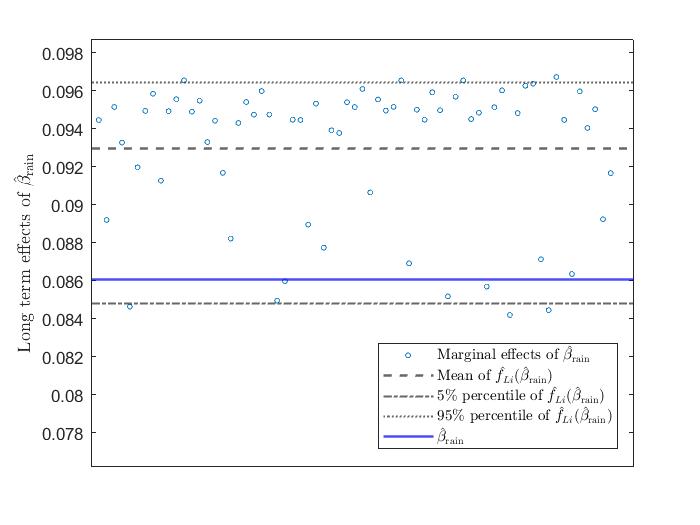}
			\caption{Long-term marginal effects, rainfall}
		\end{subfigure}%
		\begin{subfigure}[t]{0.4\textwidth}
			\includegraphics[width=\textwidth]{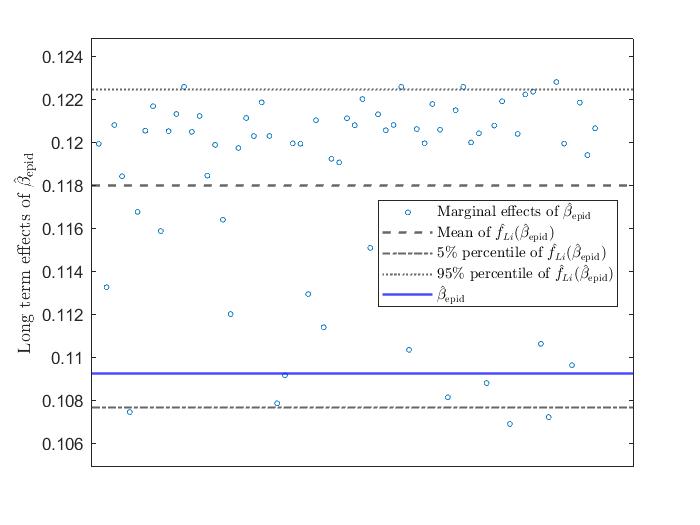}
			\caption{Long-term marginal effects, epidemic}
		\end{subfigure}
		\stla{-0.5mm}
		\caption{Marginal effects with $\bar{d}_{0}=25\%$. Top row: Short-term effects. Bottom row: Long-term effects.}
		\label{emp:fig_margin}
	\end{figure}
	%%%%%%%%%%%%%%%%%%%%%%%%%%%%%%%%%%%%%%
	\section{Conclusion}\label{sec:con}
	We provide a method for estimating SDPD models with unknown spatial weights using series approximations, while maintaining ease of implementation. 
	We allow the spatial weights in the outcome, lagged-outcome, and disturbance channels to be unknown functions of exogenous distance measures and cover both SAR and MESS specifications. Using sieve-based GMMEs, we construct an OGMME that is robust to unknown heteroskedasticity and a more efficient feasible BGMME via both linear and quadratic moments. This estimation approach can be directly extended to time-varying exogenous distances without any technical difficulties. Testing and selecting models between SAR and MESS matrix remain an area for future research.
	%%%%%%%%%%%%%% references %%%%%%%%%%%%%%%%%
	%\normalem
	\setlength{\bibsep}{0.6ex}
	\begin{spacing}{0.8}
		\normalem
		\bibliographystyle{chicago}
		\bibliography{semiW}
	\end{spacing}
	% \endgroup
	%%%%%%%%%%%%%%%%%%%%%%%%%%%%%%%%
	\appendix
	\setcounter{assumption}{0}
	\setcounter{lemma}{0}
	\setcounter{footnote}{0}
	\setcounter{table}{0}
	\setcounter{figure}{0}
	\setcounter{equation}{0}
	\setcounter{section}{0}
	\setcounter{definition}{0}
	\setcounter{notation}{0}
	%%%%%%%%%%%%%%%%%%%%%%%%%%%%%%%%%%%%%%%%%%%%%%%%%
	\renewcommand{\theequation}{\thesection\arabic{equation}}%
	\renewcommand{\thelemma}{\thesection\arabic{lemma}}
	\renewcommand{\theassumption}{\thesection\arabic{assumption}}
	\renewcommand{\thefootnote}{A\arabic{footnote}}
	\renewcommand{\theremark}{\thesection\arabic{remark}}
	\renewcommand{\thedefinition}{\thesection\arabic{definition}}
	\renewcommand{\thenotation}{\thesection\arabic{notation}}
	%\renewcommand{\theHnotation}{S\arabic{notation}}
	%%%%%%%%%%%%%%%%%%%%%%%%%%%%%%%%%%%%%%%%%%%%%%%%%
	\newgeometry{top=2cm, bottom=2cm, left=1.7cm, right=1.7cm}
	\section{Notations and Definitions}\label{sec:notation}
	We first collect all notations in our paper.
	Although some of them are already defined in the main text we believe it will be useful to state all notations at this point.
	\begin{notation}\label{notation:sup}
		(\rn1) 
		For the sieve approximation, define $\elln=\max_{1\leq k\leq 3}\ell_{k}$ and $\vsinf=\min_{1\leq k\leq 3}\varsigma_{k}$.
		\par
		(\rn2) Define $H=\mat{h_{ij}}$ as the matrix whose $(i,j)$-th entry is $h_{ij}$.
		Then, for $k=1,2,3$ and $j=1,...,\ell_{k}$, we assume $g_{k}(d)=\sum_{\pk=1}^{\ell_{k}}\lambda_{\pk}\phi_{k\pk}(d)+\delta_{k}(d)$, where $\phi_{k\pk}(\cdot)$ are basis functions and $\delta_{k}(d)$ are suitably decaying approximation errors. 
		We then denote the sieve approximation and related matrices as $\xi_{k}(d)=\sum_{\pk=1}^{\ell_{k}}\lambda_{\pk}\phi_{k\pk}(d)$, $\Xi_{k}=\mat{\xi_{k}(d_{ij})}$, and $\varPhi_{k\pk}=\mat{\phi_{k\pk}(d_{ij})}$. Defining $\Delta_{k}=\mat{\delta_{k}(d_{ij})}$, we obtain the decomposition $B_{k}=S_{k}+R_{k}$, where $S_{k}$ depends on the sieve matrix $\Xi_{k}$ and $R_{k}$ collects the approximation error matrix. Specifically, for the SAR, $S_{1}=I_{n}-\Xi_{1}, S_{2}=\Xi_{2}, S_{3}=I_{n}-\Xi_{3}$ with $R_{1}=-\Delta_{1}, R_{2}=\Delta_{2}, R_{3}=-\Delta_{3}$; for the MESS, $S_{k}=e^{\Xi_{k}}$ and $R_{k}=S_{k}(e^{\Delta_{k}}-I_{n})$. 
		\par
		(\rn3) For a set of matrices $M_{j}$, $j=1,...,p$, the following equation represents a specific matrix operation: 
		\[
		M_{1}[M_{2},M_{3},...,M_{p-1}]M_{p}=[M_{1}M_{2}M_{p},M_{1}M_{3}M_{p},...,M_{1}M_{p-1}M_{p}]
		\]
		where the column-dimension in $M_{m}$ is equal to the row-dimension in $M_{1}$, and the row-dimension in $M_{m}$ is equal to the column-dimension in $M_{p}$, for $m=2,...,p-1$. 
		\par 
		(\rn4) For the FOD transformation,
		Let $[F_{T,T-1},\frac{l_T}{\sqrt{T}}]$ be the orthogonal matrix of $J_T=I_{T}-\frac{1}{T}l_{T}l_{T}'$, where $l_{T}$ is the $T\times 1$ vector of ones, and then, the $n \times T$ matrix $\left[H_{1}, H_{2}, \cdots, H_{T}\right]$ can be transformed into the $n(T-1)$ matrix $\left[H_{1}^{*}, H_{2}^{*}, \cdots, H_{T-1}^{*}\right]=\left[H_{1}, H_{2}, \cdots, H_{T}\right] F_{T,T-1}$. Also, the $n(T-1)$ matrix $[H_{0}^{(*,-1)}, H_{1}^{(*,-1)}, \cdots, H_{T-2}^{(*,-1)}]=[H_{0}, H_{1}, \cdots, H_{T-1}]F_{T,T-1}$.  For example, $Y_{t}^{*}=\hTt(Y_{t}-\suma{h=t+1}{T}Y_{h})$ and $Y_{t-1}^{(*,-1)}=\hTt(Y_{t-1}-\suma{h=t}{T-1}Y_{h})$ depend on current and future variables, but not on the past ones, where $h_{Tt}=\sqrt{\frac{T-t}{T-t+1}}$.
		\par 
		(\rn5) For the stacked matrix form, $\bH_{N}^{*}=(H_{1}^{*\prime},...,H_{T-1}^{*\prime})'$ and $\bH_{N}^{(*,-1)}=(H_{0}^{(*,-1)\prime},...,H_{T-2}^{(*,-1)\prime})'$ for any $n\times \ell_{h}$ matrices $H_{t}$ and $H_{t-1}$ with full columns. 
		For example, $\bY_{N}^{*}=(Y_{1}^{*\prime},...,Y_{T-1}^{*\prime})$, $\bY_{N}^{(*,-1)}=(Y_{0}^{(*,-1)\prime},...,Y_{T-2}^{(*,-1)\prime})$ and $\bE_{N}^{*}=(E_{1}^{*\prime},...,E_{T-1}^{*\prime})$.
		For $k=1,2,3$ and $j=1,..,\ell_{k}$, $\bS_{kN}=I_{T-1}\otimes S_{k}$, $\bvarPhi_{kj,N}=I_{T-1}\otimes \varPhi_{kj}$, where $\varPhi_{kj}$ and $S_{k}$ are defined as Notation \ref{notation:sup}(\rn2). Thus, 
		\sloppy
		$\bZ_{N}^{*}=[\bY_{N}^{(*,-1)}, \bX_{N}^{*}]$, and
		$\pmb{\alpha}_{T-1}^{*}=(\alpha_{0}^{*},...,\alpha_{T-1}^{*})'$. 
		Moreover, $\varPhi_{k}=[\varPhi_{k1},\ldots,\varPhi_{kj}]$, $\bvarPhi_{Nk}=[\bvarPhi_{k1,N},\ldots,\bvarPhi_{kj,N}]$, and $\bL_{N}^{*}=\bS_{3N}\left[ \bZ_{N}^{*},\tfrac{\partial\bS_{1N}}{\lambda_{1}}\bY_{N}^{*}, \tfrac{\partial\bS_{2N}}{\lambda_{2}}\bY_{N}^{(*,-1)},\bzero_{N\times \ell_3}\right]$ with
		$\tfrac{\partial\bS_{kN}}{\partial\lambda_{k}}=\left[\tfrac{\partial\bS_{kN}}{\partial\lambda_{k1}},\cdots,\tfrac{\partial\bS_{kN}}{\partial\lambda_{kj}}\right]$. 
		\par 
		(\rn6) 
		Denote $N=n(T-1)$ and $\theta=(\pi',\lambda')'$, where $\pi=(\gamma,\beta')'$ and $\lambda=(\lambda_{1}^{\prime},\lambda_{2}',\lambda_{3}')'$ with $\lambda_{k}=(\lambda_{k1},...,\lambda_{k\pk})'$ for $k=1,2,3$. Denote  $\hat{g}_{k}(d)=\bphi_{k}'(d)\hat\lambda_{k}$, where $\bphi_{k}(d)=[\phi_{k1}(d),...,\phi_{k\ell_{k}}(d)]$ and $\hat\lambda_{k}$ is an estimator of $\lambda_{k}$. For the dimension,
		$\dim(\theta)=\ell_{\theta}$, $\dim(\pi)=\ell_{\pi}$, $\dim(\lambda)=\ell_{\lambda}$, $\dim(m_{N}^{\qua}(\theta))=\ell_{p}$, $\dim(m_{N}^{\lin}(\theta))=\ell_{q}$, $\dim(m_{N}(\theta))=\ell_{m}=\ell_{p}+\ell_{q}$.
		. 
	\end{notation}
	
	We then introduce the following definitions.
	\begin{definition}\label{def:UB}
		(\rn1) For any real time-varying $n\times n$ random matrix $H_t$, we say that $H_t$ has \labeltext{\textit{Property UB}}\label{propUB} if
		\begin{align*}
			\sup_{t}\rcnorm{H_t}&=O_p(1), \
			\sup_{t}\spnorm{H_t}=O_p(1),
			\ \text{and} \
			\sup_{t}\Fnorm{H_t}=O_p(\sqrt{n}),\\
			\intertext{
				\text{or the non-stochastic version}
			} 
			\sup_{t}\rcnorm{H_t}&=O(1), \
			\sup_{t}\spnorm{H_t}=O(1),
			\ \text{and} \
			\sup_{t}\Fnorm{H_t}=O(\sqrt{n})
		\end{align*}
		(\rn2) For any real time-varying $n\times n$ random matrix $H_t$, we say that $H_t$ has \labeltext{\textit{Property}}\label{propzero} $O_p(h^j)$ if
		\begin{align*}
			\sup\limits_{t}\rcnorm{H_t}&=O_p(h^j), \
			\sup\limits_{t}\spnorm{H_t}=O_p(h^j),
			\ \text{and} \
			\sup\limits_{t}\Fnorm{H_t}=O_p((\sqrt{n}h)^{j}),\\
			\intertext{
				\text{or the non-stochastic version $O(h^j)$}
			} 
			\sup\limits_{t}\rcnorm{H_t}&=O(h^j), \
			\sup\limits_{t}\spnorm{H_t}=O(h^j),
			\ \text{and} \
			\sup\limits_{t}\Fnorm{H_t}=O((\sqrt{n}h)^{j})
		\end{align*}
		with $h \to 0$ and $\lim_{n \to\infty}\sqrt{n}h \to 0$  for some positive integer $j$.
	\end{definition}
	
	\begin{definition}\label{def:eigs}
		For any real square matrix (or random matrix) $H_{N}$, where $N=n(T-1)$, we say that $H_{N}$ has \labeltext{\textit{Property SP}}\label{propSP} if
		\begin{align*}
			\sup_{N}\spnorm{H_N}&=\lim_{N\to\infty}\mu_{\max}(H_N'H_N)=O_p(1), 
			\ \text{and} \  
			\bigl(\sup_{N}\mu_{\min}(H_N'H_N)\bigr)^{-1}=O_p(1),\\ 
			\intertext{
				%\qquad\qquad\qquad\quad
				\text{or the non-stochastic version}
			} 
			\lim_{N\to\infty}\spnorm{H_N}&=\lim_{N\to\infty}\mu_{\max}(H_N'H_N)<\infty, 
			\ \text{and} \ 
			\lim_{N\to\infty}\mu_{\min}(H_N'H_N)>0.
		\end{align*}
	\end{definition}
	
	Definition \ref{def:UB} introduces the UB and convergence to zero for certain  $n\times n$ random and nonrandom matrices in each $t$. 
	For example, according to Lemma \ref{lem:boundB}(\rn2), $S_{k}$'s satisfy \ref{propUB} and $R_{k}$'s satisfy \ref{propzero}~$\Op(\elln)$.
	Definition \ref{def:eigs} captures boundedness and non-multicollinearity properties of random or fixed matrices of growing dimension. 
	\section{Regularity Assumptions}\label{sec:assumption}
	\begin{assumption}\label{ass:disterbance}
		The $\epsilon_{it}$ are independent with zero mean and finite variances that are bounded away from zero (denoted by $\sigma^2$, $\sigma_{i}^2$, or $\sigma_{t}^2$), as discussed in \ref{var:homo}, \ref{var:hetei} and \ref{var:hetet}, respectively.
		Furthermore, $\E|\epsilon_{it}|^{4+\delta_{\epsilon}}\leq C_{\epsilon}$ for some constant $C_{\epsilon}$ and $\delta_{\epsilon}>0$.%
		\footnote{If $0<\delta_{\epsilon}<1$, where $\delta_{\epsilon}$ is in Assumption \ref{ass:disterbance}, then we need $\delta>\frac{p-2}{p+4}$ to guarantee the CLT. Discussions are in footnote \footref{foot:clt}. Moreover, Assumption \ref{ass:nT} is required for the joint limit theory where $(n,T) \to \infty$. For the fixed $T$ case in Remark \ref{remark:finite T}, this rate condition is not required.}
	\end{assumption}
	\begin{assumption}\label{ass:appro_G}
		For $k=1,2,3$ and $l=1,2,\cdots,\ell_{k}$:\\  
		(\rn1) $\rcnorm{G_{k}}$ and $\rcnorm{\Upsilon}$ are uniformly bounded in $n$, and $\rho(A)<1$, where $A=B_{1}^{-1}(\gamma I_{n}+B_{2})$, $\Upsilon=\sum_{h=0}^{\infty}\abs(A^{h})$, $[\abs(A)]_{ij}=|A_{ij}|$ and $A_{ij}$ is the $(i,j)$-th element of $A$. Moreover, the initial condition $Y_{0}$ is exogenously given.
		\\
		(\rn2) The basis functions $\phi_{kl}(\cdot)$ satisfy
		$\sup_{i,k,l}\sum_{j=1}^{n}|\phi_{kl}(d_{ij})|<\infty$,
		$\sup_{j,k,l}\sum_{i=1}^{n}|\phi_{kl}(d_{ij})|<\infty$. 
		\\
		(\rn3) For some fixed $\tau>1$, $\sup_{l}|\lambda_{l}l^{\tau}|<\infty$.\\
		%$|\lambda_{l}|=O(l^{-\varsigma_{k}-1})$. 
		(\rn3') The $(i,j)$th element of the approximation error matrix $\Delta_{kl}$, defined as
		$\delta_{k}(d_{ij})=\sum_{l=\ell_{k}+1}^{\infty}\lambda_{l}\phi_{kl}(d_{ij})$, satisfying $|\delta_{k}(d_{ij})|=O(l^{-\varsigma_{k}})$ for some $\varsigma_{k}>2$.
	\end{assumption}
	Assumption \ref{ass:disterbance} provides regularity assumptions for the idiosyncratic term, and $\spnorm{\Sigma_{t}}=O(1)$ and $\spnorm{\Sigma_{t}^{-1}}=O(1)$ since $\tfrac{1}{\mumin(\Sigma_{t})}=O(1)$. In particular, we require the existence of fourth moments, thereby relaxing the condition in \citet{yang2025estimation} and \citet{gupta2025wald}, which assumes the existence of eighth moments.
	Assumption \ref{ass:appro_G} imposes the structures of unknown spatial weights matrices $G_{k}$'s and approximation tail matrices $\Delta_{k}$'s. 
	The condition of (\rn1) is parallel to assumptions on spatial weights in parametric SDPD models, see, e.g. \cite{yu2008quasi-maximum,lee2014efficient}.\footnote{One may also consider settings in which $G_k$ contains ``stars'', so that the column sums grow with $n$, see, e.g. \citet{pesaran2021estimation,lee2022qml,zhang2025nonlinear}. We leave this extension for future research.}
	$\rcnorm{G_{k}}$ requires that $g_{k}(d_{ij})$ either has a compact support or approaches zero fast enough for large $d_{ij}$ values, and that a growing sample increases the domain, not the density of $d_{ij}$. Such assumptions are also discussed in \citet{sun2016functional} and \citet{chen2025npspatial}. Alternatively, \citet{pinkse2002spatial} assume that 
	$\frac{1}{n}\sum_{i=1}^{n}\sum_{j=1}^{n}\cI{d_{ij}\in D}<\infty$ for any fixed domain set $D$, where $\cI{\cdot}$ is the indicator function. This condition is sufficient to ensure $\rcnorm{G_{k}}<\infty$ the when $g_{k}(\cdot)$ has bounded support $D$. $\rho(A)<1$ indicates that the dynamic system \eqref{eq:sdpd_np} is stable, ruling out unit root or related issues.
	$\rcnorm{\Upsilon}<\infty$ combines the absolute summability condition and the UB condition of the power series of $A$. We consider the initial condition $Y_{0}$ is exogenously given, which does not need to be specified in the estimation or asymptotic analysis.
	\footnote{For the finite $T$ case, one may specify the initial condition that is endogenous. Such a specification is beyond the scope of this paper and deserves further study.}
	(\rn2) implies that each approximation matrix $\varPhi_{kl}=\mat{\phi_{kl}(d_{ij})}$ is UB in both row and column sums. Also, it requires that there are only a finite number of neighbors for each $i$ such that  $\phi_{kl}(d_{ij})\neq 0$ in the limit. This means $\phi_{kl}(d_{ij})$ declines fast enough with growing $d_{ij}$, and that observations spread out in an increasing domain.
	(\rn3) is a `smoothness condition', and the polynomial expansion can be used by this condition, as discussed in \citet[Theorem 1 (vii)]{pinkse2002spatial}.
	By (\rn3), we can deduce that $|\delta_{k}(d_{ij})|=\Op(l^{-\varsigma_{k}})=\op(1)$. (\rn3') corresponds to a more general setting. In particular, when $g_k(d_{ij})$ is supported on an unbounded domain, such as $[0, \infty)$ or $(-\infty, \infty)$, the basis function $\mat{\phi_{kl}(\cdot)}$ can be chosen as normalized Laguerre (or Hermite) orthonormal polynomials, following \citet[Assumption A.2]{sun2016functional}.
	
	\begin{assumption}\label{ass:UB of reg}
		$\mat{x_{it}}$, $\mat{c_{i}}$ and $\mat{\alpha_{t}}$ are uniformly bounded constants.
	\end{assumption}
	Assumption \ref{ass:UB of reg} introduces the UB condition of $X_{t}$ and FE, which can also be found in \cite{lee2014efficient} and \cite{su2023identifying}. We can also assume $\sup_{1\leq l\leq \ell_{x}}\E|x_{lit}|^2\leq C_{x}$ for some constant $C_{x}$. By Assumptions \ref{ass:disterbance}, \ref{ass:appro_G}(i) and \ref{ass:UB of reg}, we can deduce that $\Fnorm{Y_{t}}=\Op(\sqrt{n})$ and $\Fnorm{Y_{t-1}}=\Op(\sqrt{n})$ by the backward substitution.
	
	\begin{assumption}[IV matrix]\label{ass:IV}
		(\rn1) The typical elements of $Q_{t}$, denoted by $q_{ijt}$ for $i=1,\cdots,n$ and $j=1,\cdots,\lqn$, satisfy $\E|q_{ijt}|^2\leq C_{q}$ with some constant $C_q$. Furthermore, $\E\bigl(\suma{t=1}{T-1}Q_{t}'E_{t}^{*}\big|\sigmaI{t-1}\bigr)=0$, where $\sigmaI{t-1}$ denotes the $\sigma$-field spanned
		by $(Y_{0},\cdots,Y_{t-1})$, conditional on $(X_{1},\cdots,X_{T},\mathbf{c}_{n},\alpha_{1},\cdots,\alpha_{T}$). 
		\\
		(\rn2) For each $j=1,...,\lpn$, $P_{jt}$ satisfies \ref{propUB}.
		\\
		(\rn3) $\lim\limits_{n,T \to\infty}\invN\bQ_{N}'\bJ_{N}\bQ_{N}$, $\lim\limits_{n,T \to\infty}\invN\bQ_{N}'J_{N}\Sigep_{N}J_{N}\bQ_{N}$, and $\lim\limits_{n,T \to\infty}\invN\bQ_{N}'\JSL{N}\bQ_{N}$ satisfy \ref{propSP}.
		\\
		(\rn4) $\lqn\sim\elln$ and $\lpn\sim\elln$.
	\end{assumption}
	
	Assumption \ref{ass:IV} collects the regularity conditions on the instrumental variables. 
	Condition (\rn1) imposes the standard IV requirements used in SDPD models (see, e.g. \citealp{lee2014efficient}).
	Condition (\rn2) requires a UB property for the IV matrices entering the quadratic moment conditions, which in turn guarantees that $\varPhi_{kl}$ is well-defined. Importantly, we do not impose any additional high-level restrictions on the instruments, such as $\tr(J_n P_{jt})=0$ (as in \citealp{lee2014efficient}) or $\Diag(P_{jt})=0$ (as in \citealp{lin2010gmm}).
	Condition (\rn3) ensures the existence of the asymptotic covariance matrix $\Omega_{N}$, which involves the inverse of $\lim\limits_{n,T\to\infty}\invN\bQ_{N}'J_{N}\bQ_{N}$ in the 2SLS approch,  the inverse of $\lim\limits_{n,T\to\infty}\invN\bQ_{N}'J_{N}\Sigep_{N}J_{N}\bQ_{N}$ in the OGMM approach and the inverse of $\lim\limits_{n,T\to\infty}\invN\bQ_{N}'\JSL{N}\Sigep_{N}\JSL{N}\bQ_{N}=\lim\limits_{n,T\to\infty}\invN\bQ_{N}'\JSL{N}\bQ_{N}$ in the BGMM approach.
	Finally, condition (\rn4) restricts the dimension of the moment conditions, $\lmn=\lpn+\lqn\sim \elln$, so that the number of moments grows at the same rate as the cross-sectional dimension $n$, in line with \citet{gupta2025wald}. Unlike the setting in \citet{lee2014efficient}, where the many-moments problem arises from both spatial and time expansions, in our framework, it stems solely
	from the number of instruments, $\lmn$, which increases with the order of the spatial power series used to approximate the unknown spatial weights. To control this growth, we impose explicit rate conditions that ensure the expansion order (and hence $\lmn$) is asymptotically dominated by both the cross-sectional
	dimension $n$ and the time dimension $T$.
	\begin{assumption}[Identification]\label{ass:id}
		(\rn1)
		$\invN\bQ_{N}'\bJ_{N}\bL_{N}^{*}$ and
		$\invN\bL_{N}^{*\prime}\JSL{N}\bL_{N}^{*}$ satisfy \ref{propSP}, where $\bL_{N}^{*}$ is defined in Notation \ref{notation:sup}.\\
		(\rn2)
		$\Omega_{N}$ satisfies \ref{propSP}, where $\Omega_{N}$ is defined in equation \eqref{eq:Var_moment}.\\
		(\rn3) $D_{N}'\Omega_{N}^{-1}D_{N}$ satisfies \ref{propSP}, where $D_{N}$ is defined in equation \eqref{eq:DN}.
	\end{assumption}
	\sloppy
	Assumption \ref{ass:id} is a regularity condition to ensure the identifiability of $\theta_{0}$. From (\rn3), we have $\lim_{n,T\to\infty}\spnorm{D_{N}'\Omega_{N}^{-1}D_{N}}=O(1)$.
	%%%%%%%%%%%%%%%%%%%%%%%%%%%%%%%%%%
	
	\section{Lemmas and Theorems}
	\subsection{Key lemmas}
	In this section, we introduce the key lemmas, while the remaining lemmas are collected in the supplementary file.
	In the following part, we denote the true estimand as $\theta_{0}=(\gamma_0',\lambda_{0}')'$.
	\begin{comment}
		\begin{assumption}\label{ass:lem_ass}
			Let $\mathbf{a}\sim(0,\mathbf{\Sigma}_a)$, where $\mathbf{\Sigma}_a = \Diag(\sigma_{a1}^2, \dots, \sigma_{an}^2)$, be an $n \times 1$ random vector with independent entries $a_i$. For each entry $a_i$:
			\\
			(\rn1) $\E[a_i] = 0$.
			\\
			(\rn2) $0<\E[|a_i|^k] \leq C_{k}$ for $k=1,2,3,4$.
			\\
			(\rn3)
			Let $\mathbf{M}$ be an $n \times n$ projection matrix that satisfies $\rank{\mathbf{M}}\leq C_\mathbf{M}\cdot m$ w.p.1 for large $n$ and some constant $C_\mathbf{M}$.
		\end{assumption}
		
		\begin{lemma}\label{lem:EnormE4}
			\ag{This is a trivial result...do we really need to state it as a lemma?}Under Assumption \ref{ass:lem_ass}, 
			$\E\Fnorm{\mathbf{a}'\mathbf{M}\mathbf{a}}= O(m).$
		\end{lemma}
		\begin{proof}
			Note that $\E\Fnorm{\mathbf{a}'\mathbf{M}\mathbf{a}}=\E|\mathbf{a}'\mathbf{M}\mathbf{a}|$ since $\mathbf{a}'\mathbf{M}\mathbf{a}$ is a scalar.
			Since $M$ is a projection matrix, $|\mathbf{a}'\mathbf{M}\mathbf{a}|=\mathbf{a}'\mathbf{M}\mathbf{a}$. Conditioning on $\mathbf{M}$,
			we have
			$\E[\mathbf{a}'\mathbf{M}\mathbf{a}\mid \mathbf{M}]=\tr(\mathbf{M}\bm\Sigma_a)\leq C_2\tr(\mathbf{M})\leq C_2 C_\mathbf{M} m=\Op(m)$ by Assumption \ref{ass:lem_ass}.
			Taking expectations and using the tower property, the proof is finished.
		\end{proof}
	\end{comment}
	
	\begin{lemma}\label{lem:transJ}
		(\rn1) For any time-varying $n \times n$ matrix $H_{t}$, denote $H^{\diamond}_{t}$ as a transformed version of $H_{t}$, such that $\left[H^{\diamond}_{t}\right]_{i j}=[H_{t}]_{i j}$ if $i \neq j$, and 
		\begin{flalign}\label{eq:transK}
			\left[H^{\diamond}_{t}\right]_{ii}=\frac{1}{n-2}\sum_{j\neq i}^n\big([H_{t}]_{ij}+[H_{t}]_{ji}\big)-\frac{1}{(n-1)(n-2)}\sum_{j=1}^{n}\sum_{k\neq j}^n[H_{t}]_{jk}. 
		\end{flalign}
		Then, $\diagM\left(J_n H^{\diamond}_{t} J_n\right)=\bzero$.
		\\
		(\rn2)
		In the case of \ref{var:hetei}, the vector of diagonal elements $\mathbf{h}_{diag}^\diamond = ([H_t]_{11}^\diamond, \dots, [H_t]_{nn}^\diamond)'$ is the solution to $C \mathbf{h}_{diag}^\diamond = d$, given by $C^{+} d$, where:
		\begin{equation*}
			C = \begin{pmatrix}
				1 - \frac{2\sigma_1^{-2}}{\varsigma} + \frac{\sigma_1^{-4}}{\varsigma^2} & \frac{\sigma_2^{-4}}{\varsigma^2} & \cdots & \frac{\sigma_n^{-4}}{\varsigma^2} \\
				\frac{\sigma_1^{-4}}{\varsigma^2} & 1 - \frac{2\sigma_2^{-2}}{\varsigma} + \frac{\sigma_2^{-4}}{\varsigma^2} & \cdots & \frac{\sigma_n^{-4}}{\varsigma^2} \\
				\vdots & \vdots & \ddots & \vdots \\
				\frac{\sigma_1^{-4}}{\varsigma^2} & \frac{\sigma_2^{-4}}{\varsigma^2} & \cdots & 1 - \frac{2\sigma_n^{-2}}{\varsigma} + \frac{\sigma_n^{-4}}{\varsigma^2}
			\end{pmatrix},
		\end{equation*}
		with $\varsigma = \sum_{k=1}^n \sigma_k^{-2}$. The $i$-th element of the vector $d$ is:
		\begin{equation*}
			d_i = \frac{1}{\varsigma \sigma_i^2} \left( \sum_{j \neq i} \sigma_j^{-2} [H_{t}]_{ji} + \sum_{j \neq i} \sigma_j^{-2} [H_{t}]_{ij} \right) - \frac{1}{\varsigma^2 \sigma_i^2} \sum_{j=1}^n \sum_{k \neq j} \sigma_j^{-2} \sigma_k^{-2} [H_{t}]_{jk}.
		\end{equation*}
		In the cases of \ref{var:hetet} and \ref{var:homo}, $\left[H^{\diamond}_{t}\right]_{ii}$ is the same as equation \eqref{eq:transK}.
		Then, the diagonal elements of $\JSL{t} H_t^\diamond \JSL{t}$ are equal to zero. 
	\end{lemma}    
	\begin{proof}
		For (\rn1), 
		the proof can be found in \citet[Lemma S3.1]{chen2025spatial}. 
		For (\rn2), the proof can be found in \citet[Lemma S3.5]{chen2025spatial} in the case of \ref{var:hetei}. In the cases of \ref{var:hetet} and \ref{var:homo}, substituting $J(\Sigma_t)=\frac{1}{\sigma_t^2} J_n$ or $\frac{1}{\sigma^2} J_n$, we have $\Diag(J_n H_t^\diamond J_n)=0$.
	\end{proof}
	%%%%%%%%%%%%%%%%%%%%%%%%%%%%%%%%%%%%%%%
	\bigskip
	Denote $\sG_{N}(\theta)=\sGN{\theta}$ and $\E\sG_{N}(\theta)=\EsGN{\theta}$. Then we consider the following lemma to establish the consistency, i.e., $\Fnorm{\hat\theta_{gmm}-\theta_{0}}=\op(1)$. The idea is borrowed from \citet[Proof of Theorem 3.3]{su2016sieve}.
	\begin{lemma}\label{lem:consistency}
		Let $\Theta_{N}=\Gamma\times\Lambda_{N}$, where $\Gamma\subset\mathbb{R}^{\ell_{\pi}}$ is a totally UB set with fixed dimension and $\Lambda_{N}=\mat{\lambda\subset\mathbb{R}^{\elln}:\rownorm{\lambda}\leq C_{\lambda}}$ for each $l=1,\cdots,\elln$.
		Under Assumptions \ref{ass:disterbance}-\ref{ass:id}:\\
		%, and $\elln^{3/2-\vsinf}+\frac{\elln^{3/2}}{n}\to 0$ as $(n,T)\to\infty$: 
		(\rn1)  %
		For any $\theta\in\Theta_{N}$ and $\eta>0$, if $\Fnorm{\theta-\theta_{0}}\geq\eta$, there exists a constant $C_{\eta}$ such that 
		$\bigl|\E\sG_{N}(\theta)-\E\sG_{N}(\theta_{0})\bigr|\geq C_{\eta}$ when $\elln^{1/2-\vsinf}\to 0$ as $(n,T)\to\infty$.
		\\
		(\rn2) $\sup_{\theta\in\Theta}\bigl|\sG_{N}(\theta)-\E\bigl(\sG_{N}(\theta)\bigr)\bigr|=\op(1)$ when $\elln^{2-\vsinf}+\frac{\elln^{2}}{n}+\frac{\elln}{\sqrt{n(T-1)}}\to 0$ as $(n,T)\to\infty$.
		% (\rn2) $\cQ_{N}(\pi_0,g(\cdot))-\cQ_{N}(\theta_{0})=\op(...)$\\
		% (\rn2) [Uniform Convergence] $\theta,\tilde{\theta}\in\Theta, \Fnorm{\sG_{N}(\theta)-\sG_{N}(\tilde{\theta})}\leq\Delta_{N}\Fnorm{\theta-\tilde{\theta}}$ where $\spnorm{\Delta_{N}}=\Op(1)$.
	\end{lemma}
	
	\begin{proof}
		(\rn1) ensures the identification. We have
		\[
		\E\sG_{N}(\theta)-\E\sG_{N}(\theta_{0})=
		\varrho_{1N}(\theta)+2\varrho_{2N}(\theta),
		\]
		where 
		\[
		\varrho_{1N}(\theta)=\left[\E(m_{N}(\theta))-\E(m_{N}(\theta_{0}))\right]'\Omega_{N}^{-1}\left[\E(m_{N}(\theta))-\E(m_{N}(\theta_{0}))\right],
		\]
		and
		\[
		\varrho_{2N}(\theta)=\left[\E(m_{N}(\theta))-\E(m_{N}(\theta_{0}))\right]'\Omega_{N}^{-1}\E(m_{N}(\theta_{0})).
		\]
		By the mean-value theorem, $\E(m_{N}(\theta))-\E(m_{N}(\theta_{0}))=D_{N}(\tilde{\theta})(\theta-\theta_{0})$, where $D_{N}$ is defined as equation \eqref{eq:DN} and $\tilde{\theta}$, where $\tilde{\theta}=\theta_{0}+\kappa(\theta-\theta_{0})$ with $|\kappa|<1$. Then, $\varrho_{1N}(\theta)=(\theta-\theta_{0})'D_{N}(\tilde{\theta})'\Omega_{N}^{-1}D_{N}(\tilde{\theta})(\theta-\theta_{0})$.
		From Assumption \ref{ass:id}(\rn2), we have $\spnorm{\Omega_{N}^{-1}}\leq\tfrac{1}{\mumin(\Omega_{N})}=O(1)$.
		Combing Assumption \ref{ass:id}(\rn3), we can find that there exists a constant $C_{\varrho}>0$, such that for all large $(n,T)$, $\varrho_{1N}(\theta)\geq C_{\varrho}\Fnorm{\theta-\theta_{0}}^2$. Hence, if $\Fnorm{\theta-\theta_{0}}\geq\eta$, then $\varrho_{1N}(\theta)\geq C_{\varrho}\eta^2$. 
		Note that 
		\[
		\begin{split}
			|\varrho_{2N}|&=\left|\left[\E(m_{N}(\theta))-\E(m_{N}(\theta_{0}))\right]'\Omega_{N}^{-1/2}\Omega_{N}^{-1/2}\E(m_{N}(\theta_{0}))\right|\leq |\varrho_{1N}|^{1/2}|\E'(m_{N}(\theta_{0}))\Omega_{N}^{-1}\E(m_{N}(\theta_{0}))|^{1/2}\\
			&\leq |\varrho_{1N}|^{1/2}\cdot\Fnorm{\Omega_{N}^{-1/2}\E(m_{N}(\theta_{0}))}
		\end{split}
		\]
		by Cauchy-Schwarz inequality. From Lemma \ref{lem:norm of the moments} in the supplement and Assumption \ref{ass:id}(\rn2), we have $\Fnorm{\Omega_{N}^{-1/2}\E(m_{N}(\theta_{0}))}\leq\spnorm{\Omega_{N}^{-1/2}}\Fnorm{\E(m_{N}(\theta_{0}))}\leq C\elln^{1/2-\vsinf}$, and thus 
		$|\varrho_{2N}(\theta)|\leq C\elln^{1/2-\vsinf}|\varrho_{1N}(\theta)|^{1/2}$. For $\Fnorm{\theta-\theta_{0}}\geq\eta$, we have 
		\[
		|\E\sG_{N}(\theta)-\E\sG_{N}(\theta_{0})|\geq\varrho_{1N}(\theta)-|\varrho_{2N}(\theta)|\geq\varrho_{1N}(\theta)^{1/2}(\varrho_{1N}(\theta)^{1/2}-2C\elln^{1/2-\vsinf}).
		\]
		Since $\varrho_{1N}(\theta)^{1/2}\geq C_{\varrho}^{1/2}\eta^{1/2}$, for all large $(n,T)$, we can have $2C\elln^{1/2-\vsinf}\leq  C_{\varrho}^{1/2}\eta^{1/2}$ as $\elln^{1/2-\vsinf}$ being sufficiently small, and thus
		\[
		\bigl|\E\sG_{N}(\theta)-\E\sG_{N}(\theta_{0})\bigr|\geq C_{\varrho}\eta=:C_{\eta}>0.
		\]
		(\rn2) ensures the uniform convergence. We use the method proposed by \citet[Theorem 1]{andrews1992generic}. 
		First, note that $\sG_{N}(\theta)-\E\sG_{N}(\theta)=H_{1N}(\theta)+2H_{2N}(\theta)$, where
		\[
		H_{1N}(\theta)=\left[m_{N}(\theta)-\E(m_{N}(\theta))\right]'\Omega_{N}^{-1}\left[m_{N}(\theta)-\E(m_{N}(\theta))\right],
		\]
		and
		\[
		H_{2N}(\theta)=\left[m_{N}(\theta)-\E(m_{N}(\theta))\right]'\Omega_{N}^{-1}\E(m_{N}(\theta)).
		\]
		From the Proof of Theorem \ref{thm:gm_pi}, we have
		$\Fnorm{m_{N}(\theta)-\E(m_{N}(\theta))}=\Op\Bigl(\elln^{3/2-\vsinf}+\frac{\elln^{3/2}}{n}+\sqrt{\frac{\elln}{n(T-1)}}\Bigr)$ and $\E\Fnorm{m_{N}(\theta)}=O(\sqrt{\elln})$. Thus, using $\spnorm{\Omega_{N}^{-1}}=O(1)$, we have
		\[
		\begin{split}
			|\sG_{N}(\theta)-\E\sG_{N}(\theta)|&
			\leq|H_{1N}(\theta)|+|H_{2N}(\theta)|
			\leq\spnorm{\Omega_{N}^{-1}}|\left(\Fnorm{m_{N}(\theta)-\E(m_{N}(\theta))}^2+2\Fnorm{m_{N}(\theta)-\E(m_{N}(\theta))}\Fnorm{m_{N}(\theta)}\right)
			\\&
			=\Op\Big(\elln^{2-\vsinf}+\frac{\elln^{2}}{n}+\frac{\elln}{\sqrt{n(T-1)}}\Big)=\op(1).
		\end{split}
		\]
		Second, for any $\theta, \tilde{\theta}\in \Theta_{0}$, we have
		$\sG_{N}(\theta)-\sG_{N}(\tilde\theta)=I_{1N}(\theta,\tilde\theta)+2I_{2N}(\theta,\tilde\theta)$, where
		\[
		I_{1N}(\theta,\tilde\theta)=\left[m_{N}(\theta)-m_{N}(\tilde\theta)\right]'\Omega_{N}^{-1}\left[m_{N}(\theta)-m_{N}(\tilde\theta)\right],
		\]
		and
		\[
		I_{2N}(\theta,\tilde\theta)=\left[m_{N}(\theta)-m_{N}(\tilde\theta)\right]'\Omega_{N}^{-1}m_{N}(\tilde\theta).
		\]
		Then we have $m_{N}(\theta)-m_{N}(\tilde\theta)=D_{N}(\bar\theta)(\theta-\tilde\theta)+\op(1)$, where $\bar\theta=\tilde\theta+\kappa(\theta-\tilde\theta)$ with $|\kappa|<1$ and $\spnorm{m_{N}(\tilde\theta)}=\Op(\sqrt{\elln})$. 
		Define $B_{e}=\mat{\theta,\tilde\theta\in\Theta_{0}:\Fnorm{\theta-\tilde{\theta}}\leq e}$ and $h_{e}=B_{e}\sup_{u\in\Theta_{0}}\Fnorm{m_{N}(u
			)}=B_{e}\cdot\Op(\sqrt{\ell_{n}})$.
		Hence,
		\[
		\begin{split}
			\sup_{\theta,\tilde\theta\in B_{e}}|I_{1N}(\theta,\tilde\theta)|&
			\leq
			\sup_{u\in\Theta_{0}}\spnorm{D_{N}(u)}^2\spnorm{\Omega_{N}^{-1}}B_{e}^2=\Op(B_{e}^2)
			\\
			\sup_{\theta,\tilde\theta\in B_{e}}|I_{2N}(\theta,\tilde\theta)|&
			\leq 
			\sup_{u\in\Theta_{0}}\spnorm{D_{N}(u)}\spnorm{\Omega_{N}^{-1}}h_{e}=O(h_{e})=\Op(\sqrt{\elln}B_{e}).
		\end{split}
		\]
		Then, $\sup_{\theta,\tilde\theta\in B_{e}}|\sG_{N}(\theta)-\sG_{N}(\tilde\theta)|=\Op(B_e^2+h_e)$, and we can choose $B_{e}=C\elln^{-1}$ as $(n,T)\to \infty$ such that $B_{e}+\sqrt{\elln}B_{e}\to 0$. Combining
		$|\sG_{N}(\theta)-\E\sG_{N}(\theta)|=\op(1)$, the consistency follows.
	\end{proof}

	\color{black}

	%%%%%%%%%%%%%%%%%%%%%%%%%%%%%%%%%%%%%%%%%%%%%%%
	To derive the asymptotic normality, we consider the following lemma.
	\begin{lemma}\label{lem:CLT}
		Let $\sU_{N}=\frac{1}{\sqrt{n(T-1)}}\pmb\tau'\Sigep_{\pi_0,gmm}^{-1/2}K_{\pi}m_{\Delta}$, where the $\ell_{\pi}\times\lmn$ matrix
		$K_{\pi}$ is defined in equation \eqref{eq:H_pi},
		\begin{flalign}\label{eq:CLT_rv}
			m_{\Delta}\equiv m_{\Delta}(\theta_{0})=(\bEL_{N}^{*\prime}\bJ_{N}\bP_{N1}^{\prime}\bJ_{N}\bE_{N}^{*},\cdots,\bEL_{N}^{*\prime}\bJ_{N}\bP_{N\lpn}^{\prime}\bJ_{N}\bE_{N}^{*},\bEL_{N}^{*\prime}\bJ_{N}\bQ_{N})',
		\end{flalign}
		and
		$\pmb\tau$ is the $\ell_{\pi}\times 1$ vector such that $\spnorm{\pmb\tau}=1$. Under Assumptions \ref{ass:disterbance}-\ref{ass:IV}, and $\sqrt{\frac{(T-1)\elln^{3}}{n}}+\sqrt{n(T-1)}\elln^{3/2-\vsinf}\to 0$ as $(n,T)\to\infty$, $\sU_{N}\stackrel{d}{\to}N(0,1)$.
	\end{lemma}
	
	\begin{proof}
		\sloppy
		By construction, we have $\E(\sU_{N})=0$ and $\Var(\sU_{N})=1$. 
		Denote 
		$
		\Tau\equiv\pmb\tau'\Sigep_{\pi_0,gmm}^{-1/2}K_{\pi}=(\Tau_{1},...,\Tau_{\lpn},\underbrace{\Tau_{\lpn+1},...,\Tau_{\lmn}}_{\Tau_{\lqn}'})
		$
		be then $1\times\lmn$ vector, then we have
		\par 
		$
		\sU_{N}=\invNsq\left(\bEL_{N}^{*\prime}(\sum_{l=1}^{\lpn}\Tau_{l}\bJ_{N}\bP_{Nl}\bJ_{N})\bE_{N}^{*}+\Tau_{\lqn}'\bQ_{N}^{\prime}\bJ_{N}\bE_{N}^{*}\right)=
		\invNsq\left(\bEL_{N}^{*\prime}\cA_{1N}\bE_{N}^{*}+\cA_{2N}\bE_{N}^{*}\right)
		$,
		where $\cA_{1N}=\sum_{l=1}^{\lpn}\Tau_{l}\bJ_{N}\bP_{Nl}\bJ_{N}$ and $\cA_{2N}=\Tau_{\lqn}'\bQ_{N}'\bJ_{N}$. 
		Thus, we have $\colnorm{\cA_{1N}}=O(\sqrt{\elln})$, $\sum_{l=1}^{N}|\cA_{2N,l}|=\Op(\sqrt{\elln})$.
		This is because $\colnorm{\cA_{1N}}\leq\sum_{l=1}^{\lpn}\colnorm{\Tau_{l}\bJ_{N}\bP_{Nl}\bJ_{N}}\leq\sum_{l=1}^{\lpn}|\Tau_{l}|\colnorm{\bJ_{N}}^2\colnorm{P_{Nl}}=O(\sqrt{\elln})$ by Assumption \ref{ass:IV}, and 
		$\sum_{l=1}^{N}|\cA_{2N,l}|=\Op(\sqrt{\elln})$ by an analogous argument.
		Note that $\E(\bEL_{N}^{*\prime}\cA_{1N}\bE_{N}^{*})=\sum_{l=1}^{\lpn}\Tau_{l}\E\big(\bEL_{N}^{*\prime}(\bJ_{N}\bP_{Nl}\bJ_{N})\bE_{N}^{*}\big)=0$, from Lemma \ref{lem:hTt} in the supplement, Assumption \ref{ass:nT}, \ref{ass:sieve-basis}, and $\sqrt{\frac{(T-1)\elln^{3}}{n}}+\sqrt{n(T-1)}\elln^{3/2-\vsinf}\to 0$ as $(n,T)\to\infty$
		we have
		$\sU_{N}=\invNsq\sum_{l=1}^{N}\calX_{l}+\op(1)$, where \begin{flalign}\label{eq:calX}
			\calX_{l}=\cA_{2N,l}\epsilon_{l}+2\epsilon_{l}\sum_{j=1}^{l-1}\cA_{1N,lj}\epsilon_{j}=\epsilon_{l}(\cA_{2N,l}+2\sum_{j=1}^{l-1}\cA_{1N,lj}\epsilon_{j}).
		\end{flalign} 
		Consider the $\sigma$-field $\mathcal{F}_{l}=\sigma(\epsilon_{1},...,\epsilon_{l})$, where $\epsilon_{l}=\epsilon_{it}$, with $l=i+n(t-1)$ and $1\leq i\leq n$, $1\leq t\leq T$. Then $\cF_{l-1}\subseteq\cF_{l}$ and $\calX_{l}$ is $\cF_{l}$-measurable\footnote{We use the trivial $\sigma$-field to define $\cF_{0}$.}. Hence, $\E(\calX_{l}|\cF_{l-1})=0 \ \text{a.s.}$, and $\mat{(\calX_{l},\cF_{l})|1\leq l\leq N}$ forms a martingale difference array. Then, we define 
		\begin{flalign}\label{eq:calX_normal}
			\calX_{l}^{\circ}=\frac{1}{\sqrt{N}}\calX_{l}.
		\end{flalign}
		From \citet[Chapter 3.2]{hall2014martingale}, sufficient conditions%
		\footnote{We need to verify $\sum_{l=1}^{N}\E\left(\calX_{l}^{\circ2}\cI{|\calX_{l}^{\circ}|>e}\right)\stackrel{p}{\rightarrow} 0$ for some constant $e$. By the elementary truncation bound $\calX_{l}^{\circ2}\cI{|\calX_{l}^{\circ}|>e}\leq e^{-\delta}|\calX_{l}^{\circ}||^{2+\delta}$, and use the Markov inequality, we can verify the sufficient condition \eqref{eq:CLT1}.} % 
		to verify the martingale difference array CLT are
		\begin{flalign}\label{eq:CLT1}
			\sum_{l=1}^{N}\E|\calX_{l}^{\circ}|^{2+\delta}\stackrel{p}{\rightarrow} 0
		\end{flalign}
		for some $\delta>0$,
		and 
		\begin{flalign}\label{eq:CLT2}
			\sum_{l=1}^{N}\left[\E(\calX_{l}^{\circ2}|\cF_{l-1})-\E(\calX_{l}^{\circ2})\right]\stackrel{p}{\rightarrow} 0.
		\end{flalign}
		% To verify \eqref{eq:CLT1}, we use $\delta=2$.
		For \eqref{eq:CLT1}, denote $u=2+\delta$, we have $\calX_{l}^{\circ}=N^{-u/2}\colnorm{\calX_{l}}^{u}$.
		and $|\calX_{l}|^{u}\leq C_u|\epsilon_{l}|^{u}|\cA_{2N,l}|^{u}+2\colnorm{\cA_{1N}}^{u}\bigl(\max_{j<l}|\epsilon_{j}|^{u}\bigr)$ by Minkowski inequality and the convex inequality that $(a+b)^u\leq C_u(a^u+b^u)$, where $u>1$ and $C_u=2^{u-1}$.
		Thus, $\E(\calX_{l}|\cF_{l-1})\leq C_{\epsilon}^2C_uC\elln^{u/2}=O(\elln^{u/2})$ and $\sum_{l=1}^{N}\E|\calX_{l}^{\circ}|^{u}=O(N)\cdot O(N^{-u/2})\cdot O(\elln^{u/2})=O\Bigl(\frac{\elln^{2+\delta}}{(n(T-1))^{\delta}}\Bigr)=o(1)$ according to Assumptions \ref{ass:nT}, \ref{ass:disterbance}, and the rate conditions in Theorem \ref{thm:gm_pi}(ii).% 
		\footnote{\label{foot:clt}
			In Theorem \ref{thm:gm_pi}(ii), we need $\sqrt{\frac{(T-1)\elln^{3}}{n}}\to 0$ which implies $\elln=o\big((\frac{n}{T})^{1/3}\big)$. Then, $O\bigl(\frac{\elln^{2+\delta}}{(n(T-1))^{\delta}}\bigr)=o\bigl(n^{(1-\delta)/3}T^{-(1+2\delta)/3}\bigr)$. Using $n=o(T^{p/2})$, $n^{(1-\delta)/3}T^{-(1+2\delta)/3}$ goes to zero provided $(p/2)(1-\delta)<1+2\delta$. Therefore, if $0<\delta<1$, we need $\delta>\frac{p-2}{p+4}$.}
		\par 
		For \eqref{eq:CLT2}, from \eqref{eq:calX} and \eqref{eq:calX_normal}, we have
		$\E(\calX_{l}^{\circ 2}|\cF_{l-1})=N^{-1}\sigma_{l}^2(\E\cA_{2N,l}+2\sum_{j<l}\cA_{1N,lj}\epsilon_{j})^2$, and 
		$\E(\calX_{l}^{\circ 2})=N^{-1}\sigma_{l}^2(\E\cA_{2N,l}^2+4\sum_{j<l}\cA_{1N,lj}^2\sigma_{j}^2)$. Thus,
		\[
		\begin{split}
			\E(\calX_{l}^{\circ2}|\cF_{l-1})-\E(\calX_{l}^{\circ2})&
			=4N^{-1}\sigma_{l}^2\Big\{\cA_{2N,l}\sum_{j<l}\cA_{1N,lj}\epsilon_{j}+4\big[\big(\sum_{j<l}\cA_{1N,lj}\epsilon_{j}\big)^2-\sum_{j<l}\cA_{1N,lj}\sigma_{j}^2\big]\Big\}\\&
			=4N^{-1}\sigma_{l}^2\left[\cA_{2N,l}\sum_{j<l}\cA_{1N,lj}\epsilon_{j}+\sum_{j<l}\cA_{1N,lj}(\epsilon_{j}^2-\sigma_{j}^2)+2\sum_{j<l}\sum_{k<l}\cA_{1N,lj}\cA_{1N,lk}\epsilon_{j}\epsilon_{k}\right]
		\end{split}
		\]
		and 
		\[
		\sum_{l=1}^{N}\left[\E(\calX_{l}^{\circ2}|\cF_{l-1})-\E(\calX_{l}^{\circ2})\right]
		=4(\Delta_{1}+\Delta_{2}+2\Delta_{3}),
		\]
		where %
		\[
		\begin{split}
			&\Delta_{1}=\frac{1}{N}\sum_{l=1}^{N}\sigma_{l}^2\cA_{2N,l}\sum_{j<l}\cA_{1N,lj}\epsilon_{j}, 
			\Delta_{2}=\frac{1}{N}\sum_{l=1}^{N}\sigma_{l}^2\sum_{j<l}\cA_{1N,lj}^2(\epsilon_{j}^2-\sigma_{j}^2) 
			\ \text{and}, \ \\
			&\Delta_{3}=\frac{1}{N}\sum_{l=1}^{N}\sigma_{l}^2\sum_{j<k<l}\cA_{1N,lj}\cA_{1N,lk}\epsilon_{j}\epsilon_{k},
		\end{split}
		\]
		and thus $\E(\Delta_{k})=0$ for $k=1,2,3$.
		For $\Delta_{1}$, rewrite by swapping sums:
		$\Delta_{1}=\frac{1}{N}\sum_{l=1}^{N}\sigma_{l}^2\cA_{2N,l}\sum_{j<l}\cA_{1N,lj}\epsilon_{j}=\frac{1}{N}\sum_{j=1}^{N-1}a_{j}\epsilon_{j}$. where $a_{j}=\sum_{l>j}\sigma_{l}^2\cA_{2N,l}\cA_{1N,lj}$.
		Thus, and $\E(\Delta_{1})=\Var(\Delta_{1})\leq\frac{C_{\epsilon}^{2/(4+\delta_{\epsilon})}}{N^2}\sum_{j=1}^{N-1}a_{j}^2$.
		This is because $\Var(\epsilon_{j})\leq\E(\epsilon_{j}^2)\leq(\E|\epsilon|^{4+\delta_{\epsilon}})^{2/(4+\delta_{\epsilon})}$ by the norm inequality.
		Using 
		\[
		|a_{j}|\leq C_{\epsilon}^{2/(4+\delta_{\epsilon})}\sum_{l>j}|\cA_{2N,l}|\sum_{l>j}|\cA_{1N,lj}|=\Op(\elln),
		\]
		and $\sum_{j=1}^{N-1}c_{j}^2=\Op(N\elln^{2})$, we have $\Var(\Delta_{1})=O\left(\frac{\elln^{2}}{N}\right)$ and $\Fnorm{\Delta_{1}}=\Op\left(\sqrt{\frac{\elln^{2}}{n(T-1)}}\right)$.
		For $\Delta_{2}$, rewrite by swapping sums:
		$\Delta_{2}=\frac{1}{N}\sum_{j=1}^{N-1}a_{2j}(\epsilon_{j}^2-\sigma_{j}^2)$, where $a_{2j}=\sum_{l>j}\sigma_{l}^2\cA_{1N,lj}^2$. Note that 
		\[
		\Var(\epsilon_{j}^2-\sigma_{j}^2)=\Var(\epsilon_{j}^2)\leq\E(\epsilon_{j}^4)\leq C_{\epsilon}^{4/(4+\delta_{\epsilon})},
		\]
		we have $\Fnorm{\Delta_{2}}=\Op\left(\sqrt{\frac{\elln^{2}}{n(T-1)}}\right)$.
		For $\Delta_{3}$, rewrite by swapping sums:
		$\Delta_{3}=\frac{1}{N}\sum_{1\leq j<k\leq N-1}\epsilon_{j}\epsilon_{k}a_{3jk}$, where $a_{3jk}=\sum_{l>k}\sigma_{l}^2\cA_{1N,lj}\cA_{1N,lk}$. Hence, 
		$\Delta_{3}^2=\frac{1}{N^2}\sum_{j<k}\sum_{r<s}a_{3jk}a_{3rs}\epsilon_{j}\epsilon_{k}\epsilon_{r}\epsilon_{s}
		=\frac{1}{N}\sum_{j<k}a_{3jk}^2\epsilon_{j}^2\epsilon_{k}^2+2\sum_{(j<k)<(r<s)}\epsilon_{j}\epsilon_{k}\epsilon_{r}\epsilon_{s}$
		where $(j<k)<(r<s)$ means sum over distinct ordered pairs, so each cross term is counted once. If $\mat{j,k}\cap\mat{r,s}=\varnothing$ or $j=r$ but $k\neq s$, then $\E\epsilon_{j}\epsilon_{k}\epsilon_{r}\epsilon_{s}=0$. The only nonzero part is the paired case $(j,k)=(r,s)$ with $j<k$ and $r<s$. Hence,
		$\E(\Delta_{3}^2)=\frac{1}{N^2}\sum_{j<k}a_{3ij}^2\E\epsilon_{i}^2\E\epsilon_{j}^2\leq\frac{C_{\epsilon}^{4/(4+\delta_\epsilon)}}{N^2}\sum_{j<k}a_{3ij}^2$. 
		By the Cauchy-Schwarz inequality, we have 
		\[
		a_{3ij}^2=(\sum_{l>k}\sigma_{l}^2\cA_{1N,lj}\cA_{1N,lk})^2\leq C_{\epsilon}^{4/(4+\delta_\epsilon)}(\sum_{l}b_{lj}^2)(\sum_{l}b_{lk}^2)\leq C_{\epsilon}^{4/(4+\delta_\epsilon)}(\sup_{l}\sum_{l}|b_{lj}|)^2(\sup_{l}\sum_{l}|b_{lk}|)^2=O(\elln),
		\] 
		$\E(\Delta_{3}^2)=O(\frac{\elln^{2}}{N})$ and thus $\Fnorm{\Delta_{3}}=\Op\left(\sqrt{\frac{\elln^{2}}{n(T-1)}}\right)$. Combining Assumptions \ref{ass:nT}, \ref{ass:disterbance} and $\sqrt{\frac{(T-1)\elln^{3}}{n}}\to 0$ as $(n,T)\to\infty$, we have $\Delta_{k}\stackrel{p}{\to} 0$ for $k=1,2,3$. The asymptotic normality of $\sU_{N}$ follows.
	\end{proof}
	
	%%%%%%%%%%%%%%%%%%%%%%%%%%%%%%%%%%%%%%%%%%%%%%%%%%%%%%%
	We then establish the consistent estimator for $\Sigma_{t}$.
	\begin{lemma}\label{lem:est_Sigma}
		Denote $\mathscr{W}_{t}(\theta)=J_{n}(V_{t}(\theta)-\frac{1}{T}\sum_{h=1}^{T}V_{h}(\theta))$ for any $\theta\in\Theta$. We can estimate $\sigma_{i}^2$ by $\hat\sigma_{i}^2=\frac{1}{T}\sum_{t=1}^{T}\hat\omega_{it}^2$ in the case of \ref{var:hetei}, and estimate $\sigma_{t}^2$ by $\hat{\sigma}_{t}^2=\frac{1}{n}\sum_{i=1}^{n}\hat\omega_{it}^2$ in the case of \ref{var:hetet}, where $\hat\omega_{it}=\omega_{it}(\hat{\theta})$,  and $\omega_{it}=\omega_{it}(\theta_0)$ is the $i$-th entry of $\mathscr{W}_{t}$ and $\hat{\theta}$ is a consistent estimator of the true estimand $\theta_0$ with $\Fnorm{\hat\theta-\theta_0}=\op(1)$. 
		Under Assumptions \ref{ass:nT}, \ref{ass:sieve-basis}, and \ref{ass:disterbance}-\ref{ass:id}, and as $(n,T) \to \infty$, the following hold:\\
		(\rn1) In the case of \ref{var:hetei}, $|\hat{\sigma}_i^2-\sigma_i^2|=\op(1)$.
		\\
		(\rn2) In the case of \ref{var:hetet}, $|\hat{\sigma}_t^2-\sigma_t^2|=\op(1)$.
	\end{lemma}
	
	\begin{proof}
		We define the centered components:
		$\Delta\epsilon_{it}=\epsilon_{it}- \bar{\epsilon}_{i\cdot}-\bar{\epsilon}_{\cdot t}+\bar{\epsilon}_{\cdot\cdot}$, where $\bar{\epsilon}_{i\cdot}=
		\frac{1}{T}\sum_{t=1}^{T}\epsilon_{it}$, $\bar{\epsilon}_{\cdot t}=
		\frac{1}{n}\sum_{i=1}^{n}\epsilon_{it}$ and  
		$\bar{\epsilon}_{\cdot\cdot}=\frac{1}{nT}\sum_{i,t}\epsilon_{it}$. And $\Delta r_{it}=r_{it}-\bar{r}_{i\cdot}-\bar{r}_{\cdot t}+\bar{r}_{\cdot\cdot}$ is defined by similar.
		The proof of (\rn2) follows similarly to that of (\rn1), so we will only provide the proof for (\rn1) here.
		Thus, $\omega_{it}=\Delta\epsilon_{it}+\Delta r_{it}$ and $\hat\omega_{it}=\omega_{it}+(\hat\omega_{it}-\omega_{it})$. We have $\hat\omega_{it}^2=\omega_{it}^2+(\hat\omega_{it}-\omega_{it})^2+2\omega_{it}(\hat\omega_{it}-\omega_{it})$. 
		Thus, for (\rn1),
		\[
		\hat\sigma_{i}^2=\dfrac{1}{T}\sum_{t=1}^{T}\omega_{it}^2+
		\dfrac{1}{T}\sum_{t=1}^{T}(\hat\omega_{it}-\omega_{it})^2+
		\dfrac{2}{T}\sum_{t=1}^{T}\omega_{it}(\hat\omega_{it}-\omega_{it})=a_1+a_2+a_3.
		\]
		For $a_1$, we have 
		\[
		a_1=\frac{1}{T}\sum_{t=1}^{T}(\Delta\epsilon_{it})^2+\frac{1}{T}\sum_{t=1}^{T}(\Delta r_{it})^2+\frac{2}{T}\Delta\epsilon_{it}\Delta r_{it}=a_{11}+a_{12}+a_{13}.
		\]
		Note that $a_{11}
		=\frac{1}{T}\sum_{t=1}^T\epsilon_{it}^2
		+(\bar{\epsilon}_{i\cdot})^2 
		+\frac{1}{T}\sum_{t=1}^T\bar{\epsilon}_{\cdot t}^2
		+\bar{\epsilon}_{\cdot\cdot}^2
		+\text{Cross Terms}$.
		For the single square terms, $\E\left[\frac{1}{T}\sum_{t=1}^T\epsilon_{it}^2\right]=\sigma_{i}^2$, $\E[\bar{\epsilon}_{i\cdot}^2]=\E\left[ \left( \frac{1}{T} \sum_{t=1}^T \epsilon_{it} \right)^2\right]=\frac{1}{T^2}\sum_{t=1}^T \E[\epsilon_{it}^2]+\frac{2}{T^2}\sum_{t<s}\E[\epsilon_{it}\epsilon_{is}]=\frac{\sigma_i^2}{T}=O(\frac{1}{T})$. 
		Similarly, $\E\left[\frac{1}{T}\sum_{t=1}^T\bar{\epsilon}_{\cdot t}^2\right]=O(\frac{1}{n})$ and $\E\left[\bar{\epsilon}_{\cdot\cdot}^2\right]=O(\frac{1}{nT})$. By the Cauchy-Schwarz inequality, the cross terms are bounded by the product of the magnitudes. For example, the cross term with the overall mean is: 
		\[
		\frac{1}{T}\sum_{t=1}^T\epsilon_{it} \bar{\epsilon}_{\cdot\cdot}= \bar{\epsilon}_{\cdot\cdot} (\frac{1}{T} \sum \epsilon_{it})=\bar{\epsilon}_{\cdot \cdot} \cdot \bar{\epsilon}_{i\cdot}
		=\Op(\frac{1}{\sqrt{nT}}) \cdot \Op(\frac{1}{\sqrt{T}})=\Op(\frac{1}{\sqrt{n}T}).
		\]
		Hence, $a_{1}=\sigma_{i}^2+\op(1)$.
		For $a_{12}$ and $a_{13}$, we can show that $a_{12}=\Op(\elln^{-2\vsinf})$ and $|a_{13}|\leq 2\sqrt{a_{11}}\sqrt{a_{12}}=\op(1)$.
		\par 
		Using the Mean Value Theorem, we have $\hat\omega_{it}-\omega_{it}=\frac{\partial\omega_{it}}{\partial\theta'}\big|_{\theta=\bar\theta}(\hat\theta-\theta_0)$. Thus,
		\[
		a_{2}\leq \Fnorm{\hat\theta-\theta_0}^2\cdot \frac{1}{T}\sum_{t=1}^{T}\bFnorm{\frac{\partial\omega_{it}}{\partial\theta'}\big|_{\theta=\bar\theta}}^2=\op(1)\cdot O(1)=\op(1)
		\]
		by Assumption \ref{ass:id} and the consistency of $\hat\theta$. We can show that $|a_{3}|=\op(1)$ By the Cauchy-Schwarz inequality. Combining all above, we have $|\hat{\sigma}_i^2-\sigma_i^2|=\op(1)$.
	\end{proof}
	\begin{lemma}\label{lem:Var_moments}
		Under Assumptions \ref{ass:nT}, \ref{ass:sieve-basis}, and \ref{ass:disterbance}-\ref{ass:id} as $(n,T)\to\infty$, $\Fnorm{\hat{\Omega}_{t}-\Omega_{t}}=\op(1)$ where $\tilde{\Omega}_{t}$ is evaluated by the consistent estimator $\hat{\theta}$ with $\Fnorm{\hat\theta-\theta_0}=\op(1)$.
	\end{lemma}
	\begin{proof}
		This result follows from the magnitudes of the order of $\cJ_{i}$'s with $i=1,2$ in \eqref{eq:remainder term}. 
		We can show that $\Fnorm{\tilde{\Omega}_{t}-\Omega_{t}}=\Op(\Fnorm{\hat\theta-\theta_0})+\Op(\frac{1}{n})+\Op(\frac{1}{T})$.
		The proof can be analogous to the proof of Lemma \ref{lem:est_Sigma}, and is available upon request.
	\end{proof}
	
	\begin{lemma}\label{lem:best IV}
		Under Assumptions \ref{ass:disterbance}-\ref{ass:UB of reg}, the best $Q_{t}$ is given by equation \eqref{eq:best Qt}.
	\end{lemma}
	
	\begin{proof}
		The intuition of the construction of $Q_{t}$ is motivated by the first-order conditions of the GMM objective function.
		From \eqref{eq:sdpd_np}, we dentoe 
		\begin{flalign}\label{eq:cA}
			A=B_{1}^{-1}(\gamma I_{n}+B_{2})=A_{y}+A_{r}
		\end{flalign}
		where $A_{y}=S^{-1}(\gamma I_{n}+S_{2})$ and $A_r=(B_{1}^{-1}-S_{1}^{-1})(\gamma I_{n}+S_{2})=-S_{1}^{-1}R_{1}B_{1}^{-1}(\gamma I_{n}+S_{2})$. It is straightforward to verify that $A_x$ satisfies \ref{propUB} and $A_r$ satisfies \propref{propzero}{$\Op(\errorD{-}{1})$}.
		% Denoting $\chiL{t}=X_{nt}\beta+\alpha_{t}l_n$, 
		For any $h>0$, we have
		\begin{flalign}\label{eq:Yt+h}
			Y_{t+h}=A^{h+1}Y_{t-1}+\suma{j=0}{h}A^{j}\invB{1}(X_{t+h-j}\beta+\bc_{n}+\alpha_{t+h-j}l_n)+\suma{j=0}{h}A^{j}B_{1}^{-1}\YEL{t+h-j}
		\end{flalign}
		
		By denoting
		\begin{flalign*}
			\begin{split}
				\chiL{t+h-j}&=\suma{j=0}{h}A^{j}\invB{1}(X_{t+h-j}\beta+\frac{1}{t-1}\suma{h=1}{t-1}\mC_{h}+\alpha_{t+h-j}l_n)\\
				E_{t+h-j}&=\suma{j=0}{h}A^{j}\invB{1}\YEL{t+h-j}+\frac{1}{t-1}\suma{h=0}{t-1}\YEL{h}
			\end{split}
		\end{flalign*}
		where $\mC_{t}=\mC_{y,t}+\mC_{r,t}=B_{1}Y_{t}-(\gamma I_{n}+B_{2})Y_{t-1}-X_{nt}\beta-\alpha_{t}l_n$ with $\mC_{y,t}=S_{1}Y_{t}-(\gamma I_{n}+S_{2})Y_{t-1}-X_{nt}\beta-\alpha_{t}l_n$ and $\mC_{r,t}=R_{1}Y_{t}-R_{2}Y_{t-1}$. 
		We then have
		\begin{flalign*}
			\begin{split}
				Y_{t-1}^{(*,-1)}&=\hTt(Y_{t-1}-\frac{1}{T-t}\suma{s=t}{T-1}Y_s)=\hTt(Y_{t-1}-\frac{1}{T-t}\suma{h=0}{T-1-t}Y_h)\\
				&=\hTt\left(I_n-\frac{1}{T-t}\suma{h=0}{T-1-t}A^{h+1}\right)Y_{t-1}-\hTt\frac{1}{T-t}\suma{h=0}{T-1-t}\chiL{t+h-j}-\hTt\frac{1}{T-t}\suma{h=0}{T-1-t}E_{t+h-j}
			\end{split}
		\end{flalign*}
		Since we use $A_{y}$ to approximate $A$ according to \eqref{eq:cA}, we then have
		\begin{flalign}\label{eq:ylagappro}
			Y_{t-1}^{(*,-1)}=\hTt\left(I_n-\frac{1}{T-t}\suma{h=0}{T-1-t}A_{y}^{h+1}\right)Y_{t-1}-\hTt\frac{1}{T-t}\suma{h=0}{T-1-t}\chiL{y,t+h-j}-\hTt\frac{1}{T-t}\suma{h=0}{T-1-t}\cQ_{t+h-j}
		\end{flalign}
		where 
		\allowdisplaybreaks
		\begin{eqnarray}
			\nonumber
			\cQ_{t+h-j}&=&A_{r}^{h+1}Y_{t-1}+\chiL{r,t+h-j}+E_{t+h-j}
			\\ \nonumber
			\chiL{y,t+h-j}&=&\suma{j=0}{h}A_{y}^{j}\invS{1}(X_{t+h-j}\beta+\frac{1}{t-1}\suma{h=1}{t-1}\mC_{y,h}+\alpha_{t+h-j}l_n)
			\\ \nonumber
			\chiL{r,t+h-j}&=&\suma{j=0}{h}A_{r}^{j}\invB{1}(X_{t+h-j}\beta+\frac{1}{t-1}\suma{h=1}{t-1}\mC_{h}+\alpha_{t+h-j}l_n)+\suma{j=0}{h}A_{y}^{j}\invR{1}(X_{t+h-j}\beta+\frac{1}{t-1}\suma{h=1}{t-1}\mC_{h}+\alpha_{t+h-j}l_n)
			\\ \label{eq:Qt+h}
			&+&\suma{j=0}{h}A_{y}^{j}\invS{1}\frac{1}{t-1}\suma{h=1}{t-1}\mC_{r,h}    
		\end{eqnarray}
		Denote $\sigmaI{t-1}$ be the sigma-algebra spanned
		by $(Y_{0},...,Y_{t-1})$ conditional on $(X_{1},...X_{T},\bc_{0},\alpha_{1},...,\alpha_{T0})$ and $\cQtilde_{tT}=\frac{1}{T-t}\suma{h=0}{T-1-t}\cQ_{t+h-j}$, from \eqref{eq:ylagappro}, we have
		\begin{flalign*}
			Y_{t-1}^{(*,-1)}=\bestY_{t}-\hTt\cQtilde_{tT}
		\end{flalign*}
		where 
		\begin{flalign}\label{eq:bestylag}
			\bestY_{t}\equiv\E[Y_{t-1}^{(*,-1)}|\sigmaI{t-1}]=\hTt\left(I_n-\frac{1}{T-t}\suma{h=0}{T-1-t}A_{y}^{h+1}\right)Y_{t-1}-\hTt\frac{1}{T-t}\suma{h=0}{T-1-t}\chiL{y,t+h-j}.    
		\end{flalign}
		We can show that $\Fnorm{\cQtilde_{tT}}=\Op(\sqrt{n}\elln^{-\vsinf})$, which has the same asymptotic behavior as $\Fnorm{r_{t}}$, and thus is asymptotically negligible with respect to the estimators.
		This is because $\Fnorm{A_{r}^{h+1}Y_{t-1}}\leq\spnorm{A_{r}}^{h+1}\Fnorm{Y_{t-1}}=\Op(\elln^{-(h+1)\vsinf})\Op(\sqrt{n})=\op(\sqrt{n}\elln^{-\vsinf})$, $\spnorm{\sum_{j=0}^{h}A_{r}^{j}}=\Op(\elln^{-\vsinf})$, $\spnorm{\mC_{r,t}}=\Op(\elln^{-\vsinf})$ and $\Fnorm{X_{t}\beta+\fe}=\Op(\sqrt{n})$ by Assumptions \ref{ass:appro_G} and \ref{ass:UB of reg}.
		\par 
		Similarly, for some time-invariant $n\times n$ matrix $\cB_{k}$ satisfying \ref{propUB} we have
		\begin{flalign}\label{eq:bestBylag}
			\cB_{k}\bestY_{t}\equiv\E[\cB_{k}Y_{t-1}^{(*,-1)}|\sigmaI{t-1}]=\hTt\cB_{k}\bigl(I_n-\frac{1}{T-t}\suma{h=0}{T-1-t}A_{y}^{h+1}\bigr)Y_{t-1}-\hTt\cB_{k}\frac{1}{T-t}\suma{h=0}{T-1-t}\chiL{y,t+h-j}  
		\end{flalign}
		Therefore, according to equation \eqref{eq:FOC of theta}, $Q_{t}$ now is in \eqref{eq:best Qt}.
	\end{proof}
	%%%%%%    Proof of the theorems  %%%%%%%%%%%
	\subsection{Proof of Theorems}
	\textit{Proof of Theorem \ref{thm:gm_pi}}.
	We establish the consistency and asymptotic normality of the OGMME in this section. The corresponding results for the BGMME follow analogously, by replacing $\bJ_{N}$ with $\JSL{N}$.
	\\
	(\rn1) From Lemma \ref{lem:foc of mess} in the supplement, the first-order condition at $\theta_{0}$ times $n(T-1)$ is approximated by\footnote{Note that $r_{t}^{*}$ also contains $\lambda$, however, it is easy to verify that the order of $\spnorm{\frac{\partial r_{t}^{*}}{\partial\lambda_{j}'}}=\Op(\elln^{-\vsinf})$ by Assumption \ref{ass:appro_G}, which is negligible. Also, we need to divide $n(T-1)$ to investigate the consistency and divide $\sqrt{n(T-1)}$ to investigate the asymptotical normality. As shown in Lemma \ref{lem:foc of mess}, the approximation terms will be negligible.}
	%\begin{footnotesize}
	\begin{align}\label{eq:FOC of theta}
		\begin{split}
			\frac{\partial m_{N}(\theta_{0})}{\partial\lambda_{1}} &\approx
			-\begin{pmatrix} 				\left(\bS_{3N}\tfrac{\partial\bS_{1N}}{\partial\lambda_{1}}\bY_{N}^{*}\right)'\bJ_{N}\bP_{N1}^{s}\bJ_{N}\bV_{N}^{*} 
				\\
				\vdots 
				\\			\left(\bS_{3N}\tfrac{\partial\bS_{1N}}{\partial\lambda_{1}}\bY_{N}^{*}\right)'\bJ_{N}\bP_{N\lpn}^{s}\bJ_{N}\bV_{N}^{*} 
				\\			\bQ_{N}'\bJ_{N}\bS_{3N}\tfrac{\partial\bS_{1N}}{\partial\lambda_{1}}\bY_{N}^{*}
			\end{pmatrix},
			\frac{\partial m_{N}(\theta_{0})}{\partial\lambda_{3}} 
			\approx
			\begin{pmatrix}			\left(\tfrac{\partial\bS_{3N}}{\partial\lambda_{3}}\bS_{3N}^{-1}\bV_{N}^{*}\right)'\bJ_{N}\bP_{N1}^{s}\bJ_{N}\bV_{N}^{*} \\
				\vdots \\
				\left(\tfrac{\partial\bS_{3N}}{\partial\lambda_{3}}\bS_{3N}^{-1}\bV_{N}^{*}\right)'\bJ_{N}\bP_{N\lpn}^{s}\bJ_{N}\bV_{N}^{*} 
				\\
				\bQ_{N}'\bJ_{N}\tfrac{\partial\bS_{3N}}{\partial\lambda_{3}}\bS_{3N}^{-1}\bV_{N}^{*}
			\end{pmatrix},
			\\
			%\begin{aligned}[t]
			\frac{\partial m_{N}(\theta_{0})}{\partial\lambda_{2}} &\approx 
			-\begin{pmatrix}				\left(\bS_{3N}\tfrac{\partial\bS_{2N}}{\partial\lambda_{2}}\bY_{N}^{(*,-1)}\right)'\bJ_{N}\bP_{N1}^{s}\bJ_{N}\bV_{N}^{*}\\
				\vdots \\
				\left(\bS_{3N}\tfrac{\partial\bS_{2N}}{\partial\lambda_{2}}\bY_{N}^{(*,-1)}\right)'\bJ_{N}\bP_{N\lpn}^{s}\bJ_{N}\bV_{N}^{*}\\
				\bQ_{N}'\bJ_{N}\left(\bS_{3N}\tfrac{\partial\bS_{2N}}{\lambda_{2\lwyl}}\bY_{N}^{(*,-1)}\right)
			\end{pmatrix},     
			\frac{\partial m_{N}(\theta_{0})}{\partial\pi'}\approx
			-\begin{pmatrix}				(\bS_{3N}\bZ_{N}^{*})'\bJ_{N}\bP_{N1}^{s}\bJ_{N}\bV_{N}^{*} \\
				\vdots \\
				(\bS_{3N}\bZ_{N}^{*})'\bJ_{N}\bP_{N\lpn}^{s}\bJ_{N}\bV_{N}^{*} \\
				\bQ_{N}'\bJ_{N}\bS_{3N}\bZ_{N}^{*}
			\end{pmatrix}
		\end{split}
	\end{align}
	%\end{footnotesize}
	where $\bZ_{N}^{*}=[\bY_{N}^{(*,-1)},\bX_{N}^{*}]$.
	\par 
	From the Taylor expansion, we have
	\begin{flalign}\label{eq:expansion of theta}
		\hat\theta_{gmm}-\theta_{0} = {}& -\Big[
		\frac{\partial m_{N}^{\prime}(\hat\theta_{gmm})}{\partial \theta}\Omega_{N}^{-1}
		\frac{\partial m_{N}(\bar\theta^{\prime})}{\partial \theta}\Big]^{-1}
		\frac{\partial m_{N}^{\prime}(\hat\theta_{gmm})}{\partial \theta}\Omega_{N}^{-1}
		\invN m_{N}(\theta_{0})
	\end{flalign}
	and by Lemma \ref{lem:VBV} in the supplement, we have
	$\frac{\partial m_{N}^{\prime}(\hat\theta_{gmm})}{\partial \theta}=D_{N}+O(\elln^{-\vsinf})+o_p(1)$, where
	\begin{flalign}\label{eq:DN}
		D_{N}=-\invN\E\begin{pmatrix}
			% Column 1: pi
			\bzero_{\lpn\times\ell_{\pi}} &
			% Column 2: lambda_1
			\bC_{\lpn\times\dwy} &
			% Column 3: lambda_2
			\bzero_{\lpn\times\ell_{2}} &
			% Column 4: lambda_3
			\bC_{\lpn\times\dwu} \\
			% Column 1: pi (Linear part)
			\bQ_{N}^{\prime}\bJ_{N}\bS_{3N}\bZ_{N}^{*} &
			% Column 2: lambda_1 (Linear part)
			\bQ_{N}^{\prime}\bJ_{N}\bS_{3N}\tfrac{\partial\bS_{1N}}{\lambda_{1\lwy}}\bY_{N}^{*} &
			% Column 3: lambda_2 (Linear part)
			\bQ_{N}^{\prime}\bJ_{N}\bS_{3N}\tfrac{\partial\bS_{2N}}{\lambda_{2\lwyl}}\bY_{N}^{(*,-1)} &
			% Column 4: lambda_3 (Linear part)
			\bzero_{\lqn\times\dwu}
		\end{pmatrix},
	\end{flalign}
	and we denote 
	\[\hspace*{-1cm}
	\begin{split}
		% & 
		D_{\pi}=-\invN\E\begin{pmatrix}
			\bzero_{\lpn\times\ell_{\pi}}\\
			\bQ_{N}^{\prime}\bJ_{N}\bZ_{N}^{*}
		\end{pmatrix} 
		\ \text{and} \
		% \\& 
		D_{\lambda}=-\invN\E\begin{pmatrix}
			\bC_{\lpn\times\dwy} & \bzero_{\lpn\times\ell_{2}} & \bC_{\lpn\times\dwu} \\
			\bQ_{N}^{\prime}\bJ_{N}\bS_{3N}\tfrac{\partial\bS_{1N}}{\lambda_{1\lwy}}\bY_{N}^{*} & 
			\bQ_{N}^{\prime}\bJ_{N}\bS_{3N}\tfrac{\partial\bS_{2N}}{\lambda_{2\lwyl}}\bY_{N}^{(*,-1)} & \bzero_{\lqn\times\dwu} 
		\end{pmatrix}.
	\end{split}
	\]
	Furthermore,
	$[\bC_{\lpn\times\dwy}]_{jk}=
	\tr\bigl(\bJ_{N}\bP_{Nj}^{s}\bJ_{N}\bS_{3N}\tfrac{\bS_{1N}}{\partial\lambda_{1k}}\bS_{1N}^{-1}\bS_{3N}^{-1}\Sigep_{N}\bigr)$ 
	and 
	$[\bC_{\lpn\times\dwu}]_{jk}=
	\tr\bigl(\bJ_{N}\bP_{Nj}^{s}\bJ_{N}\tfrac{\partial\bS_{3N}}{\partial\lambda_{3k}}\bS_{3N}^{-1}\Sigep_{N}\bigr)$.
	Then we have
	\begin{align*}
		\frac{\partial m_{N}^{\prime}(\hat\theta_{gmm})}{\partial \theta}\Omega_{N}^{-1}\invN m_{N}(\theta_{0})=
		\invN\bC_{\lpn\times\dwy}'\Varm_{1N}^{-1}\cJ_{1}
		+\invN\bC_{\lpn\times\dwu}'\Varm_{1N}^{-1}\cJ_{1}
		+\invN\cJ_{2}+\op(1)
	\end{align*}
	where 
	% $\bM_{N}=\invN\bJ_{N}\bQ_{N}(\invN\bQ_{N}'J_{N}\Sigep_{N}J_{N}\bQ_{N})^{-1}\bQ_{N}'\bJ_{N}$.
	\begin{flalign}%\hspace*{-1cm}
		\label{eq:remainder term}%\small
		\cJ_{1}=-\begin{pmatrix}
			\bV_{N}^{*\prime}\bJ_{N}\bP_{N1}^{s}\bJ_{N}\bV_{N}^{*}
			\\
			\vdots
			\\
			\bV_{N}^{*\prime}\bJ_{N}\bP_{N\lpn}^{s}\bJ_{N}\bV_{N}^{*}
			\\
			\bzero_{\lqn\times 1}
		\end{pmatrix}, 
		\ \text{and} \
		\cJ_{2}=-\bL_{N}^{*\prime}\bM_{N}\bV_{N}^{*},
	\end{flalign}
	with
	$\bL_{N}^{*}=\big[
	\bS_{3N}\bZ_{N}^{*},
	\bS_{3N}\tfrac{\partial\bS_{1N}}{\partial\lambda_{1}}\bY_{N}^{*}, 
	\bS_{3N}\tfrac{\partial\bS_{2N}}{\partial\lambda_{2}}\bY_{N}^{(*,-1)},
	\bzero_{N\times \ell_{3}}
	\big]$ as defined in Notation \ref{notation:sup}(\rn5),
	$\bM_{N}=\Diag_{t=1}^{T-1}(M_{t})$ and $M_{t}=J_{n}Q_{t}(Q_{t}'J_{n}\Sigma_{t}J_{n}Q_{t})^{-1}Q_{t}'J_{n}$.
	
	Next, due to $\bY_{N}^{(*,-1)'}\bM_{N}\bE_{N}^{*}=\sum_{t=1}^{T-1}\frac{T-t}{T-t+1}(Y_{t}-\frac{1}{T-t}\sum_{h=t}^{T-1}Y_{h})'M_{t}(E_{t}-\frac{1}{T-t}\sum_{h=t+1}^{T-1}E_{t})$,
	and we can deduce that the non-zero expectation part of is $\sum_{t=1}^{T-1}\frac{1}{T-t+1}(\kappa_{1\gamma}+\kappa_{2\gamma})$, where 
	\[
	\kappa_{1\gamma,t}=-E_{t}'M_{t}B_{3}(\sum_{h=1}^{T-t}A^{h-1})B_{1}^{-1}B_{3}^{-1}E_{t} 
	\ \text{and} \ 
	\kappa_{2\gamma,t}=\frac{1}{(T-t)}\big[\sum_{h=t+1}^{T-1}E_{h}'M_{t}(\sum_{s=1}^{T-h}A^{s-1})B_{1}^{-1}B_{3}^{-1}E_{t}\big].
	\]
	By the triangle inequality, we have
	\sloppy
	$
	\bnorm\bY_{N}^{(*,-1)'}\bM_{N}\bE_{N}^{*}\bnorm\leq\sum_{t=1}^{T-1}\frac{1}{T-t+1}\big(\|\kappa_{1\gamma}\|+\|\kappa_{2\gamma}\|\big).
	%=O\left(\ln{T}\sqrt{\frac{\lqn}{T}}\right)
	$
	Note that
	$
	\E\Fnorm{\kappa_{1\gamma,t}}=O\left(\mn\right) \ \text{and} \
	\E\Fnorm{\kappa_{2\gamma,t}}=O\left(\mn\right)
	$
	since $\Fnorm{M_{t}(\sum_{h=0}^{\infty}A^{h})B_{1}^{-1}B_{3}^{-1}}_{\rc}=\Op(\mn)$ from Lemma \ref{lem:norm of M} in the supplement. 
	By Markov's inequality and the fact that the magnitude order of $\sum_{h=1}^{T-1}\frac{T-t}{T-t+1}$ is $O(T-1)$, we have 
	%\begin{flalign}%\label{eq:order of YlagE}
	$\bFnorm{\invN\bY_{N}^{(*,-1)'}\bM_{N}\bE_{N}^{*}}=\Op\left(\mn\cdot\frac{\lnT}{n(T-1)}\right).$
	%\end{flalign}
	Hence, we can deduce that
	\begin{flalign}\label{eq:order of Ylagr}
		\bFnorm{\invN\bY_{N}^{(*,-1)\prime}\bM_{N}\brL_{N}^{*}}=
		O_p\left(\frac{n(T-1)\elln^{-\vsinf}}{n(T-1)}\cdot\mn\right)=
		\Op\left(\elln^{1-\vsinf}\right)
	\end{flalign}
	since the relationship between $\brL_{N}$ and $\bE_{N}$.
	Then, we have 
	\begin{flalign*}%\label{eq:order of ZE}
		\begin{split}
			&\bFnorm{\invN\left(\bS_{3N}\tfrac{\partial\bS_{2N}}{\partial\lambda_{2}}\bY_{N}^{(*,-1)}\right)^{\prime}\bM_{N}\bE_{N}^{*}}=\Op\left(\sqrt{\elln}\cdot\frac{\mn\lnT}{n(T-1)}\right)=\Op\left(\frac{\elln^{3/2}\lnT}{n(T-1)}\right)
			\ \text{and} \
			\\
			&\bFnorm{\invN\left(\bS_{3N}\tfrac{\partial\bS_{2N}}{\partial\lambda_{2}}\bY_{N}^{(*,-1)}\right)^{\prime}\bM_{N}\brL_{N}^{*}}=\Op(\elln)\cdot\Op(\sqrt{\elln})\cdot\Op(\elln^{-\vsinf})=\Op\left(\elln^{3/2-\vsinf}\right)
		\end{split}
	\end{flalign*}
	since $\spnorm{\tfrac{\partial\bS_{2N}}{\partial\lambda_{2}}}=O(\sqrt{\elln})$.
	For the last two terms in linear moments, we have $\bFnorm{\invN(\bS_{3N}\tfrac{\partial\bS_{1N}}{\partial\lambda_{1}}\bE_{N}^{*})'\bM_{N}\bE_{N}^{*}}=\Op(\frac{\elln^{3/2}}{n})$, and
	$\bFnorm{\invN \bX_{N}^{*\prime}\bM_{N}\bE_{N}^{*}}=\Op\left(\sqrt{\frac{\elln}{n(T-1)}}\right)$.
	Finally, from Lemma \ref{lem:VBV} in the supplement, for $j=1,...,\lpn$, we have
	\begin{flalign}\label{eq:order of VV}
		\bFnorm{\invN\brL_{N}^{*\prime}\JL{N}\PL{Nj}\JL{N}\brL_{N}^{*}}=\Op(\elln^{-2\vsinf}), 
		\ \text{and} \ 
		\bFnorm{\invN\brL_{N}^{*\prime}\JL{N}\PL{Nj}\JL{N}\bE_{N}^{*}}=\Op(\elln^{-\vsinf})=\op(\elln^{1-\vsinf}).
	\end{flalign}
	Thus, the bias of the quadratic moments is
	%\begin{flalign}%\label{eq:bias of quad moments}
	$\Op(\elln^{-2\vsinf}\cdot\elln)=\Op(\elln^{1-2\vsinf})=\op(\elln^{3/2-\vsinf}).$
	%\end{flalign}

	\subsubsection*{Consistency of $\hat\pi_{gmm}$}
	We use the $C(\alpha)$ approach (see e.g. \citealp{jin2020asymptotically,chen2025spatial}) to derive the expansion of $\hat\pi_{gmm}$ and $\hat\lambda_{gmm}$
	from \eqref{eq:expansion of theta}.
	Define the projection matrix $M_{\lambda}=I_{\ell_m}-D_{\lambda}(D_{\lambda}'\Omega_{N}^{-1}D_{\lambda})^{-1}D_{\lambda}'\Omega_{N}^{-1}$ and $M_{\pi}=I_{\ell_m}-D_{\pi}(D_{\pi}'\Omega_{N}^{-1}D_{\pi})^{-1}D_{\pi}'\Omega_{N}^{-1}$ where $\ell_m=\lpn+\lqn$.  
	Then we have
	\begin{flalign}
		\label{eq:Fnorm_pi}
		\hat\pi_{gmm}-\pi_0=-[D_{\pi}'\Omega_{N}^{-1}M_{\lambda}D_{\pi}]^{-1}D_{\pi}'\Omega_{N}^{-1}M_{\lambda}\invN m_{N}(\theta_{0})+\op(1)
		\\
		\label{eq:Fnorm_lambda}
		\hat\lambda_{gmm}-\lambda_{0}=-[D_{\lambda}'\Omega_{N}^{-1}M_{\pi}D_{\lambda}]^{-1}D_{\lambda}'\Omega_{N}^{-1}M_{\pi}\invN m_{N}(\theta_{0})+\op(1)
	\end{flalign}
	Let $\tilde{D}_{\lambda}=\Omega_{N}^{-\half}D_{\lambda}$
	and 
	$\tilde{D}_{\pi}=\Omega_{N}^{-\half}D_{\pi}$,
	and define
	$\tilde{M}_{\lambda}=\Omega_{N}^{-\half}M_{\lambda}\Omega_{N}^{-\half}=I_{\ell_m}-\tilde{D}_{\lambda}(\tilde{D}_{\lambda}'\tilde{D}_{\lambda})^{-1}\tilde{D}_{\lambda}'$ and
	$\tilde{M}_{\pi}=\Omega_{N}^{-\half}M_{\pi}\Omega_{N}^{-\half}=I_{\ell_m}-\tilde{D}_{\pi}(\tilde{D}_{\pi}'\tilde{D}_{\pi})^{-1}\tilde{D}_{\pi}'$. Each $\tilde{M}_{\bullet}$ is the orthogonal projector and thus $\spnorm{\tilde{M}_{\bullet}}=1$ where the subscript $\bullet$ denotes $\pi$ or $\lambda$. Then, we have $M_{\bullet}\leq \spnorm{\Omega_{N}^{1/2}}\spnorm{\tilde{M}_{\bullet}}\spnorm{\Omega_{N}^{1/2}}=O(1)$. Define
	\begin{flalign}\label{eq:H_pi}
		K_{\pi}=-[D_{\pi}'\Omega_{N}^{-1}M_{\lambda}D_{\pi}]^{-1}D_{\pi}'\Omega_{N}^{-1}M_{\lambda}   
		\ \text{and} \
		K_{\lambda}=-[D_{\lambda}'\Omega_{N}^{-1}M_{\pi}D_{\lambda}]^{-1}D_{\lambda}'\Omega_{N}^{-1}M_{\pi},
	\end{flalign}
	We have $\spnorm{K_{\pi}}\leq\bFnorm{[D_{\lambda}'\Omega_{N}^{-1}M_{\pi}D_{\lambda}]^{-1}}_{\mathrm{sp}}\bFnorm{D_{\lambda}'\Omega_{N}^{-1}M_{\pi}}_{\mathrm{sp}}=O(1)$ by Assumption \ref{ass:id}.
	Thus,
	\begin{flalign}\label{eq:pi_og_lln}
		\begin{split}
			\hat\pi_{gmm}-\pi_0=K_{\pi}\invN m_{N}(\theta_{0})=K_{\pi}\invN m_{\Delta}+\Op(\Fnorm{\cJ_{1}})+\Op(\Fnorm{\cJ_{2}})
		\end{split}
	\end{flalign}
	where
	$\cJ_{i}$'s are in \eqref{eq:remainder term} for $i=1,2$, and $m_{\Delta}$ is in \eqref{eq:CLT_rv}.
	Therefore, combining \eqref{eq:order of Ylagr}-\eqref{eq:order of VV}, we have
	\begin{flalign}\label{eq:rate of pi}
		\hat\pi_{gmm}-\pi_0=\Op\Bigl(\elln^{3/2-\vsinf}+\frac{\elln^{3/2}}{n}+\frac{\elln^{3/2}\lnT}{n(T-1)}+\sqrt{\frac{\elln}{n(T-1)}}\Bigr)=\Op\Bigl(\elln^{3/2-\vsinf}+\frac{\elln^{3/2}}{n}+\sqrt{\frac{\elln}{n(T-1)}}\Bigr).
	\end{flalign}
	(\rn2)
	For the limiting distribution, define
	\begin{flalign}\label{eq:Var_pi}
		\Sigep_{\pi_0,gmm}=\lim_{n,T\to\infty}K_{\pi}\Omega_{N}K_{\pi}',
	\end{flalign}
	and then $\spnorm{\Sigep_{\pi_0,gmm}}\leq\spnorm{K_{\pi}}^2\spnorm{\Omega_{N}}=O(1)$.
	By \eqref{eq:pi_og_lln}, we have
	\begin{flalign}\label{eq:piog_clt}
		\begin{split}
			\sqrt{n(T-1)}\Sigep_{\pi_0,gmm}^{-1/2}(\hat\pi_{gmm}-\pi_0)&=\Sigep_{\pi_0,gmm}^{-1/2}K_{\pi}\invNsq m_{\Delta}+\Op\left(\sqrt{n(T-1)}\Bigl(\elln^{3/2-\vsinf}+\frac{\elln^{3/2}}{n}\Bigr)\right).
		\end{split}
	\end{flalign}
	It remain to show that $\sqrt{n(T-1)}\Sigep_{\pi_0,gmm}^{-1/2}(\hat\pi_{gmm}-\pi_0)\stackrel{d}{\to}N(0,I_{\ell_{\pi}})$.
	%=\sU_{N}+\Op\left(\sqrt{\frac{(T-1)\elln^{3}}{n}}+\sqrt{n(T-1)}\elln^{3/2-\vsinf}\right)
	The results are directly obtained from Lemma \ref{lem:CLT} and the Cram\'{e}r--Wold device.
	\QED
	\textit{Proof of Theorem \ref{thm:Var_pi}}.
	For (\rn1), 
	define $\Sigma_{\pi,N}=K_{\pi}\Omega_{N}K_{\pi}'$, then
	$\bm\Sigma_{\pi_0,gmm}=\lim_{n,T\to\infty}\Sigma_{\pi,N}$.
	It suffices to show (a) $\spnorm{\hat{\bm\Sigma}_{\pi,gmm}-\Sigma_{\pi,N}}=\op(1)$ and
	(b) $\spnorm{\Sigma_{\pi,N}-\bm\Sigma_{\pi_0,gmm}}=\op(1)$.
	
	For (a), note that
	\[
	\hat{\bm\Sigma}_{\pi,gmm}-\Sigma_{\pi,N}
	=
	(\hat{H}_{\pi}-K_\pi)\hat\Omega_{N}\hat{H}_{\pi}^{\prime}
	+
	K_\pi(\hat\Omega_{N}
	-\Omega_{N})\hat{H}_{\pi}^{\prime}
	+
	K_\pi\Omega_{N}(\hat{H}_{\pi}^{\prime}
	-K_{\pi}^{\prime}).
	\]
	Taking spectral norms and using submultiplicativity yields
	\[
	\spnorm{\hat{\bm\Sigma}_{\pi,gmm}-\Sigma_{\pi,N}}
	\leq
	\spnorm{\hat{H}_{\pi}-K_\pi}
	\spnorm{\hat\Omega_{N}}\spnorm{\hat{H}_{\pi}}
	+\spnorm{K_\pi}\spnorm{\hat\Omega_{N}-\Omega_{N}}\spnorm{\hat{H}_{\pi}}
	+\spnorm{K_\pi}\spnorm{\Omega_{N}}\spnorm{\hat{H}_{\pi}-K_\pi}.
	\]
	Under Assumption \ref{ass:id}, $K_\pi$ and $\Omega_{N}$ are $O(1)$.
	By Lemma \ref{lem:consistency}$, \spnorm{\hat{D}_{\pi}-D_{\pi}}=\op(1)$ and $\spnorm{\hat{D}_{\lambda}-D_{\lambda}}=\op(1)$.
	Therefore, $\spnorm{\hat{H}_{\pi}-K_\pi}=\op(1)$ since $\spnorm{\hat\Omega_{N}-\Omega_{N}}=\op(1)$ by Lemma \ref{lem:Var_moments},
	which implies $\spnorm{\hat{\bm\Sigma}_{\pi,gmm}-\Sigma_{\pi,N}}=\op(1)$.
	For (b), it holds by the definition $\bm\Sigma_{\pi_0,gmm}=\lim_{n,T\to\infty}K_\pi\Omega_{N}K_{\pi}^{\prime}$.
	Combining (a) and (b) completes the proof.
	\par 
	For (\rn2),
	denote $\Sigep_{b_\pi}=\min\Sigep_{\pi_0,gmm}=\min\mat{\lim_{n,T\to\infty}K_{\pi}'\Omega_{N}K_{\pi}}$.
	The equation \eqref{eq:Var_pi} can be equivalently written as $\Sigep_{\pi_0,gmm}=\te_{\pi}(\lim_{n,T\to\infty}D_{N}'\Omega_{N}^{-1}D_{N})^{-1}\te_{\pi}'$, where $\te_{\pi}$ is the $\ell_{\pi}\times\ell_{\theta}$ selection matrix that extracts the first $\ell_{\pi}$ components, i.e., $\te_{\pi}\theta=\pi$. 
	In the case of BGMME, $\bJ_{N} $ is replaced by $\JSL{N}$, and accordingly, $\bM_{N}=\JSL{N}\bQ_{N}(\bQ_{N}'\JSL{N}\bQ_{N})^{-1}\bQ_{N}'\JSL{N}$. For the quadratic moments, we have $\lpn=\ell_*=\ell_1+\ell_3$, and the components in $\bC_{\ell_*\times\dwy}$ and $\bC_{\ell_*\times\dwu}$, are
	$[\bC_{\ell_*\times\dwy}]_{jk}=
	\tr\bigl(\JSL{N}\bP_{Nj}^{s}\JSL{N}\bS_{3N}\tfrac{\bS_{1N}}{\partial\lambda_{1k}}\bS_{1N}^{-1}\bS_{3N}^{-1}\Sigep_{N}\bigr)$ 
	and 
	$[\bC_{\ell_*\times\dwu}]_{jk}=
	\tr\bigl(\JSL{N}\bP_{Nj}^{s}\JSL{N}\tfrac{\partial\bS_{3N}}{\partial\lambda_{3k}}\bS_{3N}^{-1}\Sigep_{N}\bigr)$, respectively.
	Note that 
	
	\[
	\begin{split}
		[\bC_{\ell_*\times\dwy}]_{jk}
		=&\tr\bigl((\Sigep_{N}\bS_{3N}\tfrac{\bS_{1N}}{\partial\lambda_{1k}}\bS_{1N}^{-1}\bS_{3N}^{-1})^{\prime}%
		\JSL{N}\bP_{Nj}^{s}\JSL{N}\bigr)\\
		=&\frac{1}{2}\tr\bigl([\JSL{N}]^{\frac{1}{2}}(\Sigep_{N}\bS_{3N}\tfrac{\bS_{1N}}{\partial\lambda_{1k}}\bS_{1N}^{-1}\bS_{3N}^{-1})^{s}[\JSL{N}]^{\frac{1}{2}}%
		[\JSL{N}]^{\frac{1}{2}}\bP_{Nj}^{s}[\JSL{N}]^{\frac{1}{2}}\bigr)\\
		=&\left\{\frac{1}{\sqrt{2}}\vecD'\Bigl([\JSL{N}]^{\frac{1}{2}}(\Sigep_{N}\bS_{3N}\tfrac{\bS_{1N}}{\partial\lambda_{1k}}\bS_{1N}^{-1}\bS_{3N}^{-1})^{s}[\JSL{N}]^{\frac{1}{2}}\Bigr)\right\}%
		\left\{\frac{1}{\sqrt{2}}\vecD\Bigl(
		[\JSL{N}]^{\frac{1}{2}}\bP_{Nj}^{s}[\JSL{N}]^{\frac{1}{2}}\Bigr)\right\}
		\\
		=&a_{1j}'a_{2k}
	\end{split}
	\]
	by using the symmetry of $\bP_{Nj}^{s}$ and $\Sigep_{N}$, together with the identity $\tr(A)=1/2\tr(A^s)$,  and $\tr(AB)=\vecD'(A)\vecD(B)$ for any conformable matrices $A$ and $B$, where $\vecD(\cdot)$ stacks all entries of a matrix into a column vector. 
	Similarly, 
	$
	[\bC_{\ell_*\times\dwu}]_{jk}
	=a_{3j}'a_{2k}
	$,
	where $a_{3j}=\frac{1}{\sqrt{2}}\vecD\bigl([\JSL{N}]^{\frac{1}{2}}(\Sigep_{N}\tfrac{\bS_{3N}}{\partial\lambda_{3k}}\bS_{3N}^{-1})^{s}[\JSL{N}]^{\frac{1}{2}}\bigr)$
	and the components in $\Varm_{N1}$ is
	\[
	[\Varm_{N1}]_{ij}=\tr(\JSL{N}\bP_{Ni}\JSL{N}\bP_{Nj}^{s}\JSL{N})=1/2\tr(%
	[\JSL{N}]^{\frac{1}{2}}\bP_{Nj}^{s}[\JSL{N}]^{\frac{1}{2}}
	[\JSL{N}]^{\frac{1}{2}}\bP_{Nj}^{s}[\JSL{N}]^{\frac{1}{2}})=a_{2j}'a_{2j}.
	\]
	Hence, $D_{N}'\Omega_{N}^{-1}D_{N}=(\E\Sigma_{\lin}+\E\Sigma_{\qua})$, where 
	\(
	\Sigma_{\lin}=\lim_{n,T\to\infty}\invN\bL_{N}^{*'}\bM_{N}\bL_{N}^{*}
	=\lim_{n,T\to\infty}\invN\bL_{N}^{*'}[\JSL{N}]^{1/2}[\JSL{N}]^{1/2}\bQ_{N}(\bQ_{N}'\JSL{N}\bQ_{N})^{-1}\bQ_{N}'[\JSL{N}]^{1/2}[\JSL{N}]^{1/2}\bL_{N}^{*}
	\leq \lim_{n,T\to\infty}\invN\bL_{N}^{*'}\JSL{N}\bL_{N}^{*}
	\)
	and
	\(
	\Sigma_{\qua}=\lim_{n,T\to\infty}\invN
	\kappa_\lambda'\kappa_\delta(\kappa_\delta'\kappa_\delta)^{-1}\kappa_\delta'\kappa_\lambda
	\leq
	\lim_{n,T\to\infty}\invN\kappa_{\lambda}'\kappa_{\lambda}
	\)
	by the Generalized Cauchy-Schwarz inequality, where  $\kappa_\lambda=(a_{11},...,a_{1\ell_1},a_{31},...,a_{3\ell_3})'$% 
	and 
	$\kappa_\delta=(a_{21},...,a_{2\ell_{*}})'$.
	Thus, 
	\begin{flalign}\label{eq:Var_bg_pi}
		\Sigep_{\pi_0,bgmm}=\Sigep_{b_\pi}=\te_{\pi}\left(\lim_{n,T\to\infty}\Sigep_{\theta}\right)^{-1}\te_{\pi}'
	\end{flalign}
	where 
	$\Sigep_{\theta}=\E\Bigl(\invN\bL_{N}^{*'}\JSL{N}\bL_{N}^{*}+\Diag(\kappa_{\theta}'\kappa_{\theta})\Bigr)$ with $\kappa_\theta= (\bzero_{\ell_\pi\times 1},a_{11},\dots, a_{1\ell_1}, \bzero_{\ell_{2}\times 1}, a_{31}, \dots, a_{3\ell_3})'$. 
	Then, we can evaluate $\hat\Sigep_{b_\pi}$ by $\pi_{bgmm}$. Its consistency follows by similar arguments to those presented in (\rn1).
	% By the partitioned inverse formula, this expression is algebraically equivalent to $\Sigep_{\pi_0,bgmm}=\left(\lim_{n,T\to\infty}\invN D_{\pi}'\Omega_{N}^{-1}M_{\lambda}D_{\pi}\right)^{-1}$.
	\QED
	\textit{Proof of Theorem \ref{thm:gm_gk}}.
	Because 
	$|\hat{g}_{k,gmm}(d)-g_{k0}(d)|\leq|\hat{g}_{k,gmm}(d)-\xi_{k0}(d)|+|g_{k0}(d)-\xi_{k0}(d)|$. 
	Consequently, taking the supremum over $d\in D$ yields
	\begin{flalign*}
		\sup_{d\in D}|\hat{g}_{k,gmm}(d)-g_{k0}(d)|\leq
		\sup_{d\in D}\spnorm{\bphi_{k}(d)}\Fnorm{\hat{\lambda}_{k,gmm}-\lambda_{k_0}}+\Op(\elln^{-\vsinf}).
	\end{flalign*}
	According to \eqref{eq:Fnorm_lambda} and \eqref{eq:H_pi}, we know that
	\begin{flalign}\label{eq:lambda_og_lln}
		\begin{split}
			\hat\lambda_{k,gmm}-\lambda_{k_0}=K_{\lambda}\invN(\cJ_{1}+\cJ_{2}),
		\end{split}
	\end{flalign}
	combining Assumption \ref{ass:sieve-basis}, we have
	\begin{flalign}\label{eq:hatg}
		\begin{split}
			\sup_{d\in D}|\hat{g}_{k,gmm}(d)-g_{k0}(d)|&=\Op\Bigl(\elln^{1/2}\Bigl(\elln^{3/2-\vsinf}+\frac{\elln^{3/2}}{n}+\sqrt{\frac{\elln}{n(T-1)}}\Bigr)+\elln^{-\vsinf}\Bigr)
			\\&=
			\Op\Bigl(\elln^{2-\vsinf}+\frac{\elln^{2}}{n}+\frac{\elln}{\sqrt{n(T-1)}}\Bigr).
		\end{split}
	\end{flalign}
	Next, let 
	\begin{flalign}\label{eq:Var_g}
		\Sigma_{g_{k0},gmm} = \lim_{n,T \to \infty} \frac{1}{\ell_{n}} \phi_{k}'(d) \Sigma_{\pi_{0},gmm} \phi_{k}(d). 
	\end{flalign}
	Using Assumption \ref{ass:sieve-basis}(\rn2), we bound the spectral norm as:
	\begin{flalign*}
		\spnorm{\Sigep_{g_{k0},gmm}}\leq \dfrac{1}{\elln}\spnorm{\bphi_{k}}^2\spnorm{\Sigep_{\pi_0,gmm}}=\elln^{-1}\cdot O(\elln)=O(1).
	\end{flalign*}
	Combing equations \eqref{eq:lambda_og_lln} and \eqref{eq:hatg}, the sufficient condition for the CLT is
	\begin{flalign*}
		\sqrt{\frac{(T-1)\elln^{3}}{n}}+\sqrt{n(T-1)}\elln^{3/2-\vsinf}\to 0 
	\end{flalign*}
	as $(n,T)\to\infty$.
	\QED
	
	%\begin{remark}
	The sufficient rate conditions differ across estimands. 
	Theorem \ref{thm:gm_pi}(i) establishes consistency of the finite-dimensional parameter $\hat\pi_{gmm}$ using the expansion in \eqref{eq:Fnorm_pi} evaluated at $\theta_{0}$, and thus only requires
	$\elln^{3/2-\vsinf}+\elln^{3/2}/n\to0$, together with $\sqrt{\tfrac{\elln}{n(T-1)}}\to0$ as it appears in Assumption \ref{ass:sieve-basis}(\rn1). 
	In contrast, Lemma \ref{lem:consistency}(ii) imposes the stronger condition 
	$\elln^{2-\vsinf}+\tfrac{\elln^{2}}{n}+\tfrac{\elln}{\sqrt{n(T-1)}}\to0$ 
	to guarantee uniform convergence of $\Theta_N=\Gamma\times\Lambda_N$, which is needed for joint consistency of $\hat\theta=(\hat\pi',\hat\lambda')'$. 
	Theorem \ref{thm:gm_gk}(i) requires the same stronger condition to obtain uniform consistency of $\hat g_{k,gmm}(d)$ over $d\in D$. 
	The gap reflects the additional difficulty of controlling stochastic fluctuations uniformly over the growing-dimensional sieve space.

		\newgeometry{top=1.8cm, bottom=1.8cm, left=1.7cm, right=1.7cm}
		\setcounter{assumption}{0}
		\setcounter{lemma}{0}
		\setcounter{footnote}{0}
		\setcounter{table}{0}
		\setcounter{figure}{0}
		\setcounter{equation}{0}
		\setcounter{section}{0}
		\setcounter{definition}{0}
		\setcounter{notation}{0}
		\setcounter{page}{1}
		\renewcommand{\thefigure}{S\arabic{figure}}
		\renewcommand{\theHfigure}{S\arabic{figure}}
		\renewcommand{\thetable}{S\arabic{table}}
		\renewcommand{\theHtable}{S\arabic{table}}
		\renewcommand{\thepage}{S\arabic{page}}
		\renewcommand{\thesection}{S\arabic{section}}
		\renewcommand{\theHsection}{S\arabic{section}}
		\renewcommand{\theequation}{S\arabic{equation}}%
		\renewcommand{\theHequation}{S\arabic{equation}}
		\renewcommand{\thelemma}{S\arabic{lemma}}
		\renewcommand{\theHlemma}{S\arabic{lemma}}
		\renewcommand{\theassumption}{S\arabic{assumption}}
		\renewcommand{\theHassumption}{S\arabic{assumption}}
		\renewcommand{\thefootnote}{S\arabic{footnote}}
		\renewcommand{\theHfootnote}{S\arabic{footnote}}
		\renewcommand{\theremark}{S\arabic{remark}}
		\renewcommand{\theHremark}{S\arabic{remark}}
		\renewcommand{\thedefinition}{S\arabic{definition}}
		\renewcommand{\theHdefinition}{S\arabic{definition}}
		\renewcommand{\thenotation}{S\arabic{notation}}
		\renewcommand{\theHnotation}{S\arabic{notation}}
	\section*{Supplementart material}	
	\begin{center}
			{\LARGE \textbf{A Supplement to}\\ \titlename}
		\end{center}
		% \tableofcontents
		\setstretch{1.2}
		\section{Additional lemmas}
		\begin{lemma}\label{lem:normleq}
			Let $A$ and $B$ be $n \times n$ real matrices, then $\Fnorm{AB}\leq\spnorm{A}\Fnorm{B}$ or $\Fnorm{AB}\leq\spnorm{B}\Fnorm{A}$.
		\end{lemma}
		\begin{proof}
			This is a standard matrix norm inequality.
		\end{proof}
		
		\begin{lemma}\label{lem:boundB}
			(\rn1) For any time-varying square matrix $\Bt$,  if
			$\sup\limits_{t}\|\Bt\|_{1}=\sup\limits_{t}\|\Bt\|_{\infty}=O(1)$, then $\sup\limits_{t}\spnorm{\Bt}=O(1)$.
			\\
			(\rn2) For any time-varying square random matrix $\Bt$, depending on an index $n$. If 
			$\sup_{t}\colnorm{\Bt}=\Op(1)$
			and
			$\sup_{t} \rownorm{\Bt} = \Op1)$,
			then $\sup_{t} \spnorm{\Bt} = \Op(1)$.
		\end{lemma}
		\begin{proof}
			(\rn1) 
			We use the standard matrix norm inequality $\spnorm{\Bt} \leq \sqrt{\colnorm{\Bt} \rownorm{\Bt}}$. Taking the supremum over $t$ yields:
			$\sup_{t} \spnorm{\Bt} \leq \sup_{t} \sqrt{\colnorm{\Bt} \rownorm{\Bt}} \leq \sqrt{\left( \sup_{t} \colnorm{\Bt} \right) \left( \sup_{t} \rownorm{\Bt} \right)}=O(1).$
			\\
			(\rn2)
			According to the proof of (\rn1),
			let $X = \sup_{t} \colnorm{\Bt}$ and $Y = \sup_{t} \rownorm{\Bt}$. Then, $Z = XY = O_p(1)$. Since norms are non-negative, $Z \ge 0$. Therefore, $\sqrt{Z} = \sqrt{XY} = O_p(1)$. Hence, $0 \leq \sup_{t} \spnorm{\Bt} \leq \sqrt{XY}= O_p(1)$.
		\end{proof}
		
		\begin{lemma}\label{lem:e-I} 
			Let the $n \times n$ matrix $A$ with its typical element $a_{ij}$ satisfy 
			$\sup_{i,j}|a_{ij}|=\Op(d^{-\nu})$ and
			$\sup_{i}\suma{j=1}{n}|a_{ij}|=\Op(d^{-\nu})$ for some $\nu > 0$. Suppose that $n^{1/2}d^{-\nu} \to 0$ as $n \to\infty$, then:
			$$
			\Fnorm{e^A - I}=\Op(\Fnorm{A})=\Op(n^{1/2}d^{-\nu}).
			$$
		\end{lemma}
		\begin{proof}\sloppy
			First, it is easy to conclude that $\Fnorm{A}=\Op(n^{1/2}d^{-\nu})$ since $\suma{j=1}{n}|a_{ij}|^2 \leq \left(\sup_{k}|a_{ik}| \right) \left(\suma{j=1}{n}|a_{ij}|\right)\leq\Op(d^{-2\nu})$.  Therefore, $\Fnorm{A}^2=\Op(n d^{-2\nu})$.
			Second,
			Using the triangle inequality for the Frobenius norm and the submultiplicative property, we get:
			$\Fnorm{e^A - I}\leq\Fnorm{A}+\frac{\Fnorm{A^2}}{2!}+\frac{\Fnorm{A^3}}{3!}+\cdots$.
			Since $\Fnorm{A^k} \leq \Fnorm{A}^k$, this becomes:
			$\Fnorm{e^A-I}\leq\Fnorm{A}+\frac{\Fnorm{A}^2}{2!}+\frac{\Fnorm{A}^3}{3!}+\cdots=e^{\Fnorm{A}}-1$.
			Since $\Fnorm{A^k} \leq \Fnorm{A}^k$, this becomes:
			$\Fnorm{e^A - I}\leq\Fnorm{A}+\frac{\Fnorm{A}^2}{2!} + \frac{\Fnorm{A}^3}{3!}+\cdots=e^{\Fnorm{A}}-1$.
			Let $X_p=\Fnorm{A}$, 
			we have $e^{X_p} - 1 = X_p + \op(1)$. 
			Combining above all, we conclude that $\Fnorm{e^A - I}=\Op(\Fnorm{A})$.
		\end{proof}
		Lemma \ref{lem:e-I} shows that the MESS-type has the same asymptotic behavior as that of the SAR-type.
		\begin{lemma}\label{lem:hTt}
			When $T \to\infty$,
			$\suma{t=1}{T-1}\hTt^2=O(T)$ and $\suma{t=1}{T-1}(1-\hTt^2)=O(\ln{T})$, where $\hTt=\sqrt{\frac{T-t}{T-t+1}}$.
		\end{lemma}
		\begin{proof}
			Note that $\hTt^2=1-\frac{1}{T-t+1}$, and we have $\suma{t=1}{T-1}\hTt^2=O(T)-O(\ln{T})=O(T)$ since $\suma{t=1}{T-1}\frac{1}{T-t+1}=O(\ln{T})$.
		\end{proof}
		
		\begin{lemma}\label{lem:Rr}
			% Let $\elln^{-\approg}=\max\limits_{i=1,2,3}\ell_{k}^{-\approgk}$.
			Under Assumption \ref{ass:appro_G} and $\ell_{k}^{-\approgk}+n^{1/2}\ell_{k}^{-\approgk} \to 0$ as $n \to\infty$:\\
			(\rn1) $e^{\Delta_{k}}-I_{n}$ satisfies \propref{propzero}{$\Op(\ell_{k}^{-\varsigma_{k}})$}, where $\Delta_{k}$ is defined in Notation \ref{notation:sup}(\rn2). Also, $R_{k}$ and $S_{k}R_{k}$ satisfy \propref{propzero}{$\Op(\ell_{k}^{-\varsigma_{k}})$}, and $B_{k}$ and $S_{k}$ satisfy \ref{propUB};
			\\
			(\rn2)
			$H_k$ satisfies \propref{propzero}{$\Op(\errorD{-}{3})$}, where $H_{k}=R_{k}B_{k}^{-1}$;
			\\
			(\rn3) 
			$\E\Fnorm{r_{kt}^{*}}=O(n^{1/2}\ell_{k}^{-\approgk}h_{Tt})$ where $h_{Tt}=\sqrt{\frac{T-t}{T-t+1}}$, $r_{kt}^{*}=h_{Tt}(r_{kt}-\frac{1}{T-t}\suma{h=t+1}{T}r_{kt})$ and $r_{kt}$ is defined in equation \eqref{eq:Vt}.
		\end{lemma}
		\begin{proof}
			For (\rn1), first, from Lemma \ref{lem:e-I}, we can derive that $\Fnorm{e^{\Delta_{k}}-I_{n}}=\Op(\Fnorm{\Delta_{k}})=\Op(n^{1/2}\ell_{k}^{-\approgk})$.  
			Second, $\Vert S_{k}\Vert_{\rc}\leq e^{\suma{\pk=1}{\ell_{k}}|\lambda_{\pk}|\Vert\varPhi_{k\pk}\Vert_{\rc}}=\Op(1)$ from Assumption \ref{ass:appro_G} (\rn1) and (\rn2) and thus $\spnorm{S_{k}}=\Op(1)$.
			Third, since
			$R_{k}=S_{k}(e^{\Delta_{k}}-I_{n})$, by Lemmas \ref{lem:normleq} and \ref{lem:e-I}, we have
			\[
			\Fnorm{R_{\gk}}\leq\spnorm{S_{k}}\Fnorm{e^{\Delta_{k}}-I_{n}}=\Op(1)\Op(\Fnorm{\Delta_{k}})=\Op(n^{1/2}\ell_{k}^{-\approgk}).
			\]
			The results of $\Fnorm{R_{k}}_{\rc}=\Op(\ell_{k}^{-\approgk})$ since $\Fnorm{e^{\Delta_{k}}-I_{n}}_{\rc}=\Op(\Fnorm{\Delta_{k}}_{\rc})=\Op(\ell_{k}^{-\approgk})$.
			Fourth, we have 
			\[
			\Fnorm{S_{k}R_{\gk}}\leq\spnorm{S_{k}}\Fnorm{R_{\gk}}=\Op(n^{1/2}\ell_{k}^{-\approgk}).
			\]
			For (\rn2), first, combining \eqref{eq:sdpd_np} and \eqref{eq:matrix function}, we have
			\begin{flalign}\label{eq:Ut, Et and RS}
				U_{t}=B_{3}^{-1}E_{t}=(R_{3}+S_{3})^{-1}E_{t}
			\end{flalign}
			When $k=3$, we have $r_{3t}=R_{3}U_{t}=H_{3}E_{t}$, where $H_{3}=R_{3}B_{3}^{-1}$.
			Since $B_{k}$, $S_{k}$ and $R_{k}$ are invertible, we have 
			$B_{3}^{-1}=(I_{n}+S_{3}^{-1}R_{3})^{-1}S_{3}^{-1}$.
			Then, we have $\spnorm{S_{k}^{-1}R_{k}}=\Op(n^{1/2}\ell_{k}^{-\approgk})$ which implies $\ell_{k}^{\approgk}=o(n^{1/2})$. Thus, by the Neumann series, we have $\spnorm{(I_{n}+S_{3}^{-1}R_{3})^{-1}}=O(1)$, $\spnorm{B_{3}^{-1}}=O(1)$ and $\spnorm{H_{3}}=\Op(\ell_{k}^{-\approwu})$.
			For (\rn3), consider $k=3$, and $r_{3t}^{*}=h_{Tt}(r_{3t}-\frac{1}{T-t}\suma{h=t+1}{T}r_{3t}).$
			% From the \textit{proof of Theorem 1} in \cite{pinkse2002spatial}, we have
			Note that
			\begin{flalign*}
				\begin{split}
					\E\Fnorm{r_{3t}}^2&
					=\E\suma{i=1}{n}r_{3t,i}^2=\E\suma{i=1}{n}(\suma{j=1}{n}h_{3,ij}\epsilon_{jt})^2
					=\E\suma{i=1}{n}\suma{j=1}{n}\suma{k=1}{n}h_{3,ij}h_{3,ik}\epsilon_{jt}\epsilon_{kt}
					\\
					&=\E\suma{i=1}{n}\suma{j=1}{n}h_{3,ij}^2\epsilon_{jt}^2
					\leq\sqrt{\E\left[\suma{i=1}{n}(\sum_{j=1}^{n}h_{3,ij}^2)^2\right]}\sqrt{\E\left[\suma{j=1}{n}\epsilon_{j}^4\right]}\\
					&=O(\sqrt{n}\errorD{-2}{3})\cdot O(\sqrt{n})
					=O(n\errorD{-2}{3}),
				\end{split}  
			\end{flalign*}
			by Cauchy-Schwartz inequality and the fact that $\suma{j=1}{n}h_{3,ij}^2\leq(\suma{j=1}{n}|h_{ij}|)^2\leq \rownorm{H_{3}}^2=\Op(\errorD{-2}{3})$.
			Thus,
			$\E\Fnorm{r_{3t}}\leq\sqrt{\E\Fnorm{r_{3t}}^2}=O(n^{1/2}\ell_{k}^{-\approwu})$. So we have
			$\E\Fnorm{\frac{1}{T-t}\suma{h=t+1}{T}r_{3t}}\leq\frac{1}{T-t}\suma{h=t+1}{T}\E\Fnorm{r_{3t}}=O(n^{1/2}\ell_{k}^{-\approwu}).$
			Finally, $\E\Fnorm{r_{3t}^{*}}=O(n^{1/2}\ell_{3}^{-\approwu}h_{Tt})$. Similarly, we have
			$\E\Fnorm{r_{1t}^{*}}=O(n^{1/2}\ell_{1}^{-\approwy}h_{Tt})$ and $\E\Fnorm{r_{2t}^{*}}=O(n^{1/2}\ell_{2}^{-\approwyl}h_{Tt})$ as $\rho(A)<1$.
		\end{proof}
		%In the following part, we denote the true estimand as $\theta_{0}=(\gamma_0',\lambda_{0}')'$.
		
		\begin{lemma}\label{lem:EPE}
			Under Assumption \ref{ass:sieve-basis}(\rn1), the covariance of $\mNL(\theta_{0})$ and $\mNQ(\theta_{0})$ is asymptotically zero as $(n,T)\to\infty$.
		\end{lemma}
		\begin{proof}
			Similar to the Proof of Theorem \ref{thm:gm_pi}.
		\end{proof}
		
		\begin{lemma}\label{lem:norm of M}
			Under Assumption \ref{ass:IV}, $\Fnorm{M_{t}}_{\rc}=\Op\left(\mn\right)$, and $\Fnorm{M_{t}}=\Op\left(\sqrt{\mn}\right)$.
		\end{lemma}
		\begin{proof}
			Note that the entries of $M_{t}$ are
			\[
			m_{ijt}=\tfrac{1}{n}\Bigl[J_{n}Q_{t}\Bigr]_{i}'\Bigl(\tfrac{1}{n}Q_{nt}'J_{n}\Sigma_{t}J_{n}Q_{nt}\Bigr)^{-1}\Bigl[J_{n}Q_{t}\Bigr]_{j}
			\]
			and thus
			\[
			|m_{ijt}|=\Op\Bigl(\tfrac{1}{n}\bFnorm{\bigl[J_{n}Q_{t}\bigr]_{i}}\bFnorm{\bigl[J_{n}Q_{t}\bigr]_{j}}\Bigr)
			=\Op\Bigl(\frac{\elln}{n}\Bigr)
			\]
			uniformly in $i,j$ for each $t$.
			Similarly, we also observe that
			\[
			\sum_{j=1}^{n}m_{ijt}^2=\Op\Bigl(\frac{\elln}{n}\Bigr).
			\]
			uniformly in $i$ for each $t$.
			So we have
			$\rcnorm{M_{t}}=\Op(\elln)$,  $\spnorm{M_{t}}=\Op(\elln)$
			and thus $\Fnorm{M_{t}}=\Op(\sqrt{\elln})$.
		\end{proof}
		Lemma \ref{lem:norm of M} indicates that the linear moments may be subject to various moment issues due to the cross-sectional dimension, which can be viewed as an extension of those presented in \cite{lee2014efficient}.
		
		\begin{lemma}\label{lem:VBV}
			For $n\times n$ time-varying matrix $\cB_{t}$, suppose $\cB_{t}$ satisfies \propref{propzero}{$\Op(b)$} and $Q_{jt}$ is the $j$-th column components of the IV $Q_{t}$, under Assumptions \ref{ass:nT}, \ref{ass:appro_G} and \ref{ass:IV}:\\
			(\rn1) $\E\Fnorm{r_{kt}^{*}}^2=O(n^2\hTt^4\errorD{-4}{k})$;
			\\
			(\rn2)
			$\E|r_{kt}^{*\prime}{\Bt}r_{kt}^{*}|=O\left(n\hTt^2\ell_{k}^{-2\approgk}b\right)$;\\
			(\rn3)
			$\E|r_{kt}^{*\prime}\Bt E_{t}^{*}|=O(n\hTt^2\ell_{k}^{-\approgk}b)$.\\
			(\rn4)
			$\E|Q_{jt}'\Bt(r_{kt}^{*}+E_{t}^{*})|=O(\sqrt{n}\hTt\ell_{k}^{-\varsigma_{k}}b)$.
		\end{lemma}
		\begin{proof}
			For (\rn1)
			$\E\Fnorm{r_{kt}}^2=\E\left[\sum_{i}r_{kt,i}^2\right]^2=\E\left[\sum_{i}(\sum_{j}h_{3,ij}\epsilon_{j})^2\right]^2=\E(E_{t}'H_{3}'H_{3}E_{t})^2$. Note that $H_{3}'H_{3}$ satisfies \propref{propzero}{$\Op(h^2)$} with $h=\errorD{-}{3}$, and we have $\E\Fnorm{r_{kt}}^2=O(n^2\hTt^4\errorD{-4}{k})$.
			For (\rn2), taking $k=3$ as an example, according to Lemmas \ref{lem:normleq} and \ref{lem:Rr}, we know that the dominant term is $\hTt\E|r_{3t}^{\prime}{\Bt}r_{3t}|$. According to the Cauchy-Schwarz inequality, 
			$\E|r_{kt}^{\prime}{\Bt}r_{3t}|\leq\sqrt{\E\Fnorm{r_{kt}^{\prime}{\Bt}r_{3t}}^2}=\sqrt{\E\Fnorm{E_{t}^{\prime}H_{3}'\Bt H_{3}E_{t}}^2}$, also, $H_{3}'\Bt H_{3}$ satisfies \propref{propzero}{$\Op(b\errorD{-2}{3})$}. Thus, $\E\Fnorm{r_{3t}^{\prime}{\Bt}r_{3t}}^2=O\left(n^2\errorD{-4}{3}b^2\right)$.
			The result of (\rn3)  can be proved similarly to (\rn2). For (\rn4), we know that $Q_{t}$ can be expressed as a linear combination of $Y_{t-1}^{(*,-1)}$ and $X_{t}^{*}$. Since that $Y_{t-1}^{(*,-1)}$ and $X_{t}^{*}$ are both uncorrelated with $E_{t}^{*}$, the desired result follows immediately by arguments analogous to those used for (\rn2) and (\rn3). 
		\end{proof}
		%Lemma \ref{lem:VBV} is related to the bias term of the moment conditions involving quadratic and linear moments, where $\Bt$ is related to $P_{jt}$ and $M_{t}$, for $j=1,...,\lpn$ and $k=1,2,3$. 
		%%%%%%%%%%%%%%%%%%%%%%%%%%%%%%%%%%%%%%%%%%%%%%%%%%%%%%%%%%
		\begin{lemma}\label{lem:bgmm_trans}
			Under Assumption \ref{ass:nT}, $\Var(\MSL{t}\Sigma_{t}^{-1/2}E_{t}^{*})=\MSL{t}+o(1)$.    
		\end{lemma}
		\begin{proof}
			For the case of \ref{var:homo} and \ref{var:hetei}, we have $\Var(\MSL{t}\Sigma_{t}^{-1/2}E_{t}^{*})=\MSL{t}$. Hence, we only need to investigate the case of \ref{var:hetet}. Denote $H_{t}=\MSL{t}\Sigma_{t}^{-1/2}E_{t}^{*}$.
			To prove that $\Var(H_t)=\MSL{t}+o(1)$ under the case of \ref{var:hetet} where $\Sigma_t=\sigma_t^2 I_n$, we first observe that the scaling matrix simplifies to $\Sigma_t^{-1/2}=\sigma_t^{-1} I_n$ and the projection matrix reduces to the standard centering matrix $M(\Sigma_t)=I_n- (\sigma_t^{-1}l_n)(n \sigma_t^{-2})^{-1}(\sigma_t^{-1}l_n)'=J_n$. 
			Since the variance of the forward orthogonal deviation error $E_t^*$ is dominated by the current period variance $\sigma_t^2 I_n$ for large $T$, specifically, $\Var(E_t^*) = \sigma_t^2 I_n + O((T-t)^{-1})$, the variance of the transformed error becomes 
			$\Var(H_t)=J_n(\sigma_t^{-1} I_n)[\sigma_t^2 I_n + o(1)](\sigma_t^{-1} I_n)J_n 
			=J_n(I_n+o(1))J_n=J_n+o(1)$. Given that $M(\Sigma_t)=J_n$ in this specification, it follows that $\Var(H_t)=\MSL{t}+\op(1)$ as $T \to \infty$.
		\end{proof}
		
		\begin{lemma}\label{lem:norm of the moments}
			Under Assumptions \ref{ass:disterbance}-\ref{ass:IV}, $\E\Fnorm{m_{N}(\theta_{0})}=o(1)$.
		\end{lemma}
		\begin{proof}
			
			Observe that
			\begin{flalign*}
				\begin{split}
					\E\Fnorm{m_{N}(\theta_{0})}&
					\leq\invN\Bigl(\E\Fnorm{\sum_{j=1}^{\lpn}\bV_{N}^{*\prime}\bJ_{N}\bP_{Nj}\bJ_{N}\bV_{N}^{*}}+\E\Fnorm{\sum_{j=1}^{\lqn}\bQ_{Nj}'\bJ_{N}\mathbf{V}_{N}^{*}}\Bigr)
					\\&
					=O\left((\elln^{-2\vsinf}+\elln^{-\vsinf})\cdot\elln^{1/2}\right)=o(1)
				\end{split}
			\end{flalign*}
			where the last two relations follow by the same argument used to bound $\cA_{1N}$ and $\cA_{2N}$ in the proof of Lemma \ref{lem:CLT}.
		\end{proof}

		\begin{lemma}\label{lem:foc of mess}
			(\rn1) For the MESS-type, when $j=1,...,\ell_{k}$ and $k=1,2,3$, under Assumption \ref{ass:sieve-basis} and $S_{k}$'s satisfy \ref{propUB}, we have
			\begin{flalign*}
				\spnorm{\tfrac{\partial S_{k}}{\partial\lambda_{kj_0}}-\varPhi_{kj}S_{k}}=O(\spnorm{\varPhi_{kj}})=O(1).
			\end{flalign*}
			(\rn2) $\spnorm{\frac{\partial S_{k}}{\partial\lambda_{k_0}}}=O(\sqrt{\elln})$, where $\frac{\partial S_{k}}{\partial\lambda_{k_0}}=\left[\frac{\partial S_{k}}{\partial\lambda_{k_{10}}},...,\frac{\partial S_{k}}{\partial\lambda_{{k\ell_{k}}_0}}\right]$ for the true estimand $\lambda_{kj_0}$.
		\end{lemma}
		\begin{proof}
			Since $\varPhi_{kj}$ does not commute with $S_{k}$, the Fr\'{e}chet derivative of the matrix exponential gives
			\begin{flalign*}
				\tfrac{\partial S_{k}}{\partial\lambda_{kj_0}}=\int_0^1 e^{(1-t) \Xi_k}\varPhi_{kj} e^{t\Xi_k} d s
				%=S_k\left(\int_0^1 e^{-t\Xi_k} \Phi_{j}e^{t\Xi_k}dt\right)
			\end{flalign*}
			Because the norm of an integral is less than the integral of the norm, we have
			\begin{flalign*}
				\begin{split}
					\spnorm{\tfrac{\partial S_{k}}{\partial\lambda_{kj_0}}}&\leq \left(\int_0^1\spnorm{e^{(1-t) \Xi_k}\varPhi_{kj} e^{t\Xi_k}}dt\right)\leq \left(\int_0^1\spnorm{e^{(1-t)\Xi_k}} \spnorm{\varPhi_{kj}}\spnorm{e^{t\Xi_k}}dt\right)
					\\
					&\leq
					\spnorm{\varPhi_{kj}}\left(\int_0^1e^{(1-t)\spnorm{\Xi_k}} e^{t\spnorm{\Xi_k}}dt\right)\leq \spnorm{\varPhi_{kj}}\cdot e^{\spnorm{\Xi_k}}=O(\spnorm{\varPhi_{kj}})=O(1).
				\end{split}
			\end{flalign*}
			from Assumption \ref{ass:sieve-basis}.
			Thus, we have
			\begin{flalign*}
				\spnorm{\tfrac{\partial S_{k}}{\partial\lambda_{kj_0}}-\varPhi_{kj}S_{k}}\leq\spnorm{\tfrac{\partial S_{k}}{\partial\lambda_{kj_0}}}+\spnorm{\varPhi_{kj}S_{k}}=O(\spnorm{\varPhi_{kj}}).
			\end{flalign*}
		\end{proof}
		\section{Estimation of heteroskedastic variances}\label{sec:estimate_variance}
		First, we explain the transformation operator $\JSL{t}$ derived in BGMME, which can be motivated from the approximated log-likelihood function
		\begin{flalign*}
			L_{nT}(\theta,\mathbf{c}_{n},\bm\alpha_{T},\Sigma_{t})=-\frac{nT}{2}\ln(2\pi)-\frac{nT}{2}\ln|\Sigma_{t}(\theta)|+T\bigl(\ln|(S_{1}(\lambda)|+\ln|S_{3}(\lambda)|\bigr)-\sum_{t=1}^{T}V_{t}^{c\prime}(\theta)\Sigma_{t}^{-1}(\theta)V_{t}^{c}(\theta).
		\end{flalign*}
		The use of the approximate likelihood relies on the negligibility of $r_{t}$, which in turn permits the replacement of the true $g_{k0}$ with asymptotically negligible cost. 
		Concentrating out $\bm\alpha_{T}$ by the first-order condition, we have
		\begin{flalign}
			L_{nT}(\theta,\mathbf{c}_{n},\Sigma_{t})=-\frac{nT}{2}\ln(2\pi)-\frac{nT}{2}\ln|\Sigma_{t}(\theta)|+T\bigl(\ln|(S_{1}(\lambda)|+\ln|S_{3}(\lambda)|\bigr)-\sum_{t=1}^{T}V_{t}^{c\prime}(\theta)\JSL{t}V_{t}^{c}(\theta),
		\end{flalign}
		where $V_{t}^{c}(\theta)=S_{3}(\lambda)S_{1}(\lambda)Y_{t}-S_{3}\bigl((\gamma I_{n}+S_{2})Y_{t-1}+X_{t}\beta+\mathbf{c}_{n}\bigr)$ and $\JSL{t}=\Sigma_{t}^{-1}-\Sigma_{t}^{-1}l_{n}(l_{n}'\Sigma_{t}^{-1}l_{n})^{-1}l_{n}'\Sigma_{t}^{-1}$. Thus, $\JSL{t}$ enables the best moment conditions to mimic the score of the likelihood function, and thus it can provide the moment conditions more efficiently than the counterparts relying on the operator $J_{n}$.
		The BGMME has two main advantages over MLE. First, when the model is SAR-type, it avoids evaluating the Jacobian determinant. Second, the BGMME is subject only to approximation bias from sieves, whereas MLE can suffer additional bias from the incidental-parameters problem, as documented in the dynamic panel data literature.
		
		\section{Additional results}\label{sec:mmc}
		In this section, we present comprehensive results from our Monte Carlo experiments. While the main text focused on MESS-type specifications, here we provide the corresponding finite-sample performance for SAR-type models. 
		Tables are presented on the following pages.
		%%%%%%%%%%%%%% references %%%%%%%%%%%%%%%%%
		\setlength{\bibsep}{0.6ex}
		\begin{spacing}{0.8}\footnotesize
			\normalem
			\bibliographystyle{apalike}
			\bibliography{semiW}
		\end{spacing}
		%%%%%%%%%%%%%% Additional MC
		\newpage
		\begin{table}[htbp]
			\centering \small \stla{0cm} \stl{0.8mm} 
			\caption{Finite sample performance of estimators of $\pi$ for the MESS-type, $\Var(\epsilon_{it})=\sigma_i^2$. Additional results.} 
			\begin{spacing}{1}%\hspace*{-5mm}
				\begin{tabular}{lcccccclcccccc}
					\toprule
					\multicolumn{7}{c}{$(n,T,\elln)=(100,10,2)$}                         & \multicolumn{7}{c}{$(n,T,\elln)=(100,10,[n^{1/5}]+2)$} \\
					\cmidrule{2-7}\cmidrule{9-14}
					\multirow{2}[4]{*}{} & \multicolumn{2}{c}{2SLS} & \multicolumn{2}{c}{OGMM} & \multicolumn{2}{c}{BGMM} & \multirow{2}[4]{*}{} & \multicolumn{2}{c}{2SLS} & \multicolumn{2}{c}{OGMM} & \multicolumn{2}{c}{BGMM} \\
					\cmidrule{2-7}\cmidrule{9-14}          & $\gamma$ & $\beta$  & $\gamma$ & $\beta$  & $\gamma$ & $\beta$  &       & $\gamma$ & $\beta$  & $\gamma$ & $\beta$  & $\gamma$ & $\beta$ \\
					\cmidrule{2-7}\cmidrule{9-14}
					Bias  & -0.0539  & 0.0062  & -0.0544  & 0.0049  & -0.0562  & 0.0044  & & -0.0508  & 0.0047  & -0.0517  & 0.0036  & -0.0524  & -0.0016  \\
					ESD   & 0.0393  & 0.0381  & 0.0388  & 0.0345  & 0.0435  & 0.0386  &  & 0.0381  & 0.0350  & 0.0374  & 0.0339  & 0.0398  & 0.0375  \\
					RMSE  & 0.0667  & 0.0386  & 0.0668  & 0.0348  & 0.0711  & 0.0388  &  & 0.0635  & 0.0353  & 0.0638  & 0.0341  & 0.0658  & 0.0375  \\   
					CP    & 0.7800  & 0.9190  & 0.7770  & 0.9240  & 0.7260  & 0.8750  &  & 0.7890  & 0.8950  & 0.7940  & 0.9180  & 0.7220  & 0.8840  \\
					\cmidrule{2-7}\cmidrule{9-14}
					\multicolumn{7}{c}{$(n,T,\elln)=(100,25,2)$}                         & \multicolumn{7}{c}{$(n,T,\elln)=(100,25,[n^{1/5}]+2)$} \\
					\cmidrule{2-7}\cmidrule{9-14}
					& \multicolumn{1}{c}{2SLS} &       & \multicolumn{1}{c}{OGMM} &       & \multicolumn{1}{c}{BGMM} &       & \multirow{2}[4]{*}{} & \multicolumn{2}{c}{2SLS} & \multicolumn{2}{c}{OGMM} & \multicolumn{2}{c}{BGMM} \\
					\cmidrule{2-7}\cmidrule{9-14}          &$\gamma$ & $\beta$  & $\gamma$ & $\beta$  & $\gamma$ & $\beta$  &       & $\gamma$ & $\beta$  & $\gamma$ & $\beta$  & $\gamma$ & $\beta$ \\
					\cmidrule{2-7}\cmidrule{9-14}
					Bias  & -0.0445  & 0.0095  & -0.0434  & 0.0080  & -0.0429  & 0.0077  &  & -0.0441  & 0.0074  & -0.0433  & 0.0060  & -0.0428  & 0.0008  \\
					ESD   & 0.0246  & 0.0243  & 0.0223  & 0.0240  & 0.0211  & 0.0214  &  & 0.0231  & 0.0237  & 0.0204  & 0.0219  & 0.0201  & 0.0205  \\
					RMSE  & 0.0508  & 0.0261  & 0.0488  & 0.0253  & 0.0478  & 0.0227  &  & 0.0498  & 0.0248  & 0.0479  & 0.0227  & 0.0473  & 0.0205  \\
					CP    & 0.8970  & 0.9090  & 0.9000  & 0.9270  & 0.9070  & 0.8340  &  & 0.8960  & 0.9010  & 0.9010  & 0.9140  & 0.8990  & 0.9060  \\   
					\cmidrule{2-7}\cmidrule{9-14}
					\multicolumn{7}{c}{$(n,T,\elln)=(200,10,2)$}                         & \multicolumn{7}{c}{$(n,T,\elln)=(200,10,[n^{1/5}]+2)$} \\
					\cmidrule{2-7}\cmidrule{9-14}
					\multirow{2}[4]{*}{} & \multicolumn{2}{c}{2SLS} & \multicolumn{2}{c}{OGMM} & \multicolumn{2}{c}{BGMM} & \multirow{2}[4]{*}{} & \multicolumn{2}{c}{2SLS} & \multicolumn{2}{c}{OGMM} & \multicolumn{2}{c}{BGMM} \\
					\cmidrule{2-7}\cmidrule{9-14}          & $\gamma$ & $\beta$  & $\gamma$ & $\beta$  & $\gamma$ & $\beta$  &       & $\gamma$ & $\beta$  & $\gamma$ & $\beta$  & $\gamma$ & $\beta$ \\
					\cmidrule{2-7}\cmidrule{9-14}
					Bias  & -0.0396  & 0.0071  & -0.0372  & 0.0067  & -0.0369  & 0.0066  & & -0.0349  & 0.0050  & -0.0343  & 0.0035  & -0.0339  & -0.0050  \\
					ESD   & 0.0320  & 0.0282  & 0.0260  & 0.0246  & 0.0256  & 0.0231  & & 0.0299  & 0.0278  & 0.0249  & 0.0232  & 0.0249  & 0.0222  \\
					RMSE  & 0.0509  & 0.0291  & 0.0454  & 0.0255  & 0.0449  & 0.0240  &  & 0.0460  & 0.0282  & 0.0424  & 0.0235  & 0.0421  & 0.0228  \\
					CP    & 0.8850  & 0.8820  & 0.8540  & 0.9020  & 0.8770  & 0.8520  &   & 0.8930  & 0.8950  & 0.8840  & 0.9060  & 0.8820  & 0.9030  \\   
					\cmidrule{2-7}\cmidrule{9-14}
					\multicolumn{7}{c}{$(n,T,\elln)=(200,25,2)$}                         & \multicolumn{7}{c}{$(n,T,\elln)=(200,25,[n^{1/4}])$}
					\\
					\cmidrule{2-7}\cmidrule{9-14}
					\multirow{2}[4]{*}{} & \multicolumn{2}{c}{2SLS} & \multicolumn{2}{c}{OGMM} & \multicolumn{2}{c}{BGMM} & \multirow{2}[4]{*}{} & \multicolumn{2}{c}{2SLS} & \multicolumn{2}{c}{OGMM} & \multicolumn{2}{c}{BGMM} \\
					\cmidrule{2-7}\cmidrule{9-14}           & $\gamma$ & $\beta$  & $\gamma$ & $\beta$  & $\gamma$ & $\beta$  &       & $\gamma$ & $\beta$  & $\gamma$ & $\beta$  & $\gamma$ & $\beta$ \\
					\cmidrule{2-7}\cmidrule{9-14}
					Bias  & -0.0199  & 0.0079  & -0.0124  & 0.0087  & -0.0123  & 0.0063  &   & -0.0155  & 0.0062  & -0.0109  & 0.0055  & -0.0103  & -0.0048  \\
					ESD   & 0.0138  & 0.0145  & 0.0135  & 0.0130  & 0.0132  & 0.0128  &  & 0.0135  & 0.0135  & 0.0133  & 0.0128  & 0.0130  & 0.0126  \\
					RMSE  & 0.0242  & 0.0165  & 0.0183  & 0.0156  & 0.0180  & 0.0143  &   & 0.0206  & 0.0149  & 0.0172  & 0.0139  & 0.0166  & 0.0135  \\
					CP    & 0.9140  & 0.8990  & 0.9160  & 0.9210  & 0.9250  & 0.9320  &    & 0.9310  & 0.9610  & 0.9680  & 0.9210  & 0.9700  & 0.9440  \\
					\cmidrule{2-7}\cmidrule{9-14}
					\multicolumn{7}{c}{$(n,T,\elln)=(400,10,2)$}                         & \multicolumn{7}{c}{$(n,T,\elln)=(400,10,[n^{1/5}]+2)$} \\
					\cmidrule{2-7}\cmidrule{9-14}
					\multirow{2}[4]{*}{} & \multicolumn{2}{c}{2SLS} & \multicolumn{2}{c}{OGMM} & \multicolumn{2}{c}{BGMM} & \multirow{2}[4]{*}{} & \multicolumn{2}{c}{2SLS} & \multicolumn{2}{c}{OGMM} & \multicolumn{2}{c}{BGMM} \\
					\cmidrule{2-7}\cmidrule{9-14}          & $\gamma$ & $\beta$  & $\gamma$ & $\beta$  & $\gamma$ & $\beta$  &       & $\gamma$ & $\beta$  & $\gamma$ & $\beta$  & $\gamma$ & $\beta$ \\
					\cmidrule{2-7}\cmidrule{9-14}
					Bias  & -0.0173  & -0.0055  & -0.0161  & 0.0031  & -0.0153  & 0.0010  &  & -0.0151  & 0.0023  & -0.0131  & 0.0002  & -0.0126  & -0.0010  \\
					ESD   & 0.0241  & 0.0233  & 0.0177  & 0.0158  & 0.0177  & 0.0157  &   & 0.0241  & 0.0226  & 0.0170  & 0.0155  & 0.0170  & 0.0154  \\
					RMSE  & 0.0297  & 0.0239  & 0.0239  & 0.0161  & 0.0234  & 0.0157  &    & 0.0284  & 0.0227  & 0.0215  & 0.0155  & 0.0212  & 0.0154  \\
					CP    & 0.9470  & 0.9140  & 0.9610  & 0.9180  & 0.9520  & 0.9230  &   & 0.9260  & 0.9200  & 0.9450  & 0.9240  & 0.9180  & 0.9420  \\
					\cmidrule{2-7}\cmidrule{9-14}
					\multicolumn{7}{c}{$(n,T,\elln)=(400,25,2)$}                         & \multicolumn{7}{c}{$(n,T,\elln)=(400,25,[n^{1/5}]+2)$}
					\\
					\cmidrule{2-7}\cmidrule{9-14}
					\multirow{2}[4]{*}{} & \multicolumn{2}{c}{2SLS} & \multicolumn{2}{c}{OGMM} & \multicolumn{2}{c}{BGMM} & \multirow{2}[4]{*}{} & \multicolumn{2}{c}{2SLS} & \multicolumn{2}{c}{OGMM} & \multicolumn{2}{c}{BGMM} \\
					\cmidrule{2-7}\cmidrule{9-14}           & $\gamma$ & $\beta$  & $\gamma$ & $\beta$  & $\gamma$ & $\beta$  &       & $\gamma$ & $\beta$  & $\gamma$ & $\beta$  & $\gamma$ & $\beta$ \\
					\cmidrule{2-7}\cmidrule{9-14}
					Bias  & -0.0142  & -0.0095  & -0.0125  & 0.0031  & -0.0122  & 0.0010  &    & -0.0133  & 0.0029  & -0.0113  & 0.0010  & -0.0071  & 0.0008  \\
					ESD   & 0.0152  & 0.0155  & 0.0095  & 0.0100  & 0.0095  & 0.0099  &     & 0.0140  & 0.0121  & 0.0091  & 0.0091  & 0.0091  & 0.0090  \\
					RMSE  & 0.0208  & 0.0182  & 0.0157  & 0.0105  & 0.0155  & 0.0100  &    & 0.0193  & 0.0124  & 0.0145  & 0.0092  & 0.0115  & 0.0090  \\
					CP    & 0.9430  & 0.9510  & 0.9530  & 0.9180  & 0.9430  & 0.9640  &    & 0.9570  & 0.9490  & 0.9310  & 0.9520  & 0.9420  & 0.9400  \\
					\bottomrule
				\end{tabular}%
				\hspace*{-1cm}
				\begin{tablenotes}
					\footnotesize
					\item \textbf{Note:} The true parameters are set to $\pi_0=(0.3,1)'$, and $\bard_0=15\%$. The results are based on 1,000 Monte Carlo replications. Bias denotes the mean bias of the estimates, ESD denotes the standard deviation, RMSE denotes the root mean squared error, and CP denotes the 95\% coverage probability. The 2SLS refers to the two-stage least square estimator, OGMM refers to the feasible optimal GMM estimator, and BGMM refers to the feasible best GMM estimator. 
				\end{tablenotes}
			\end{spacing}
			\label{sim:hei_robust}%
		\end{table}%

		% \begin{landscape}
			% \begin{table}[H]
				\begin{sidewaystable}[htbp]
					\centering\footnotesize\stla{0cm} \stl{0.8mm} \vspace*{-5mm}
					\caption{Finite sample performance of estimators of $G_{1}, G_{2}$ and $G_{3}$ for the MESS-type, $\Var(\epsilon_{it})=\sigma_i^2$. Additional results.}
					\begin{spacing}{0.5}\hspace*{-5mm}
						\begin{tabular}{lrrrrrrrrrlrrrrrrrrr}
							\toprule
							\multirow{4}[3]{*}{}	& \multicolumn{9}{c}{$(n,T,\elln)=(100,10,2)$} &   & \multicolumn{9}{c}{$(n,T,\elln)=(100,10,[n^{1/5}]+2)$} \\
							\cmidrule{2-20}     
							\multirow{2}[4]{*}{} & \multicolumn{3}{c}{2SLS} & \multicolumn{3}{c}{OGMM} & \multicolumn{3}{c}{BGMM} & \multirow{2}[4]{*}{} & \multicolumn{3}{c}{2SLS} & \multicolumn{3}{c}{OGMM} & \multicolumn{3}{c}{BGMM} \\        
							\cmidrule{2-10}\cmidrule{12-20}        & \multicolumn{1}{c}{$\tilde{g}_1$} & \multicolumn{1}{c}{$\tilde{g}_2$} & \multicolumn{1}{c}{$\tilde{g}_3$} & \multicolumn{1}{c}{$\tilde{g}_{1}$} & \multicolumn{1}{c}{$\tilde{g}_{2}$} & \multicolumn{1}{c}{$\tilde{g}_{3}$} & \multicolumn{1}{c}{$\tilde{g}_{1}$} & \multicolumn{1}{c}{$\tilde{g}_{2}$} & \multicolumn{1}{c}{$\tilde{g}_{3}$} &       & \multicolumn{1}{c}{$\tilde{g}_{1}$} & \multicolumn{1}{c}{$\tilde{g}_{2}$} & \multicolumn{1}{c}{$\tilde{g}_{3}$} & \multicolumn{1}{c}{$\tilde{g}_{1}$} & \multicolumn{1}{c}{$\tilde{g}_{2}$} & \multicolumn{1}{c}{$\tilde{g}_{3}$} & \multicolumn{1}{c}{$\tilde{g}_{1}$} & \multicolumn{1}{c}{$\tilde{g}_{2}$} & \multicolumn{1}{c}{$\tilde{g}_{3}$} \\
							\midrule
							MAE   & 0.0339  & 0.0352  & 0.0530  & 0.0338  & 0.0351  & 0.0514  & 0.0340  & 0.0352  & 0.0505  & MAE   & 0.0323  & 0.0342  & 0.0509  & 0.0322  & 0.0341  & 0.0499  & 0.0327  & 0.0343  & 0.0489  \\
							Bias  & -0.0204  & -0.0257  & -0.0464  & -0.0201  & -0.0259  & -0.0457  & -0.0197  & -0.0258  & -0.0451  & Bias  & -0.0187  & -0.0255  & -0.0450  & -0.0185  & -0.0257  & -0.0446  & -0.0180  & -0.0255  & -0.0441  \\
							RMSE  & 0.0228  & 0.0288  & 0.0531  & 0.0225  & 0.0289  & 0.0517  & 0.0224  & 0.0289  & 0.0506  & RMSE  & 0.0209  & 0.0284  & 0.0509  & 0.0207  & 0.0285  & 0.0500  & 0.0204  & 0.0284  & 0.0491  \\
							\midrule
							& \multicolumn{9}{c}{$(n,T,\elln)=(100,25,2)$} &   & \multicolumn{9}{c}{$(n,T,\elln)=(100,25,[n^{1/5}]+2)$} \\
							\cmidrule{2-20}     
							\multirow{2}[4]{*}{} & \multicolumn{3}{c}{2SLS} & \multicolumn{3}{c}{OGMM} & \multicolumn{3}{c}{BGMM} & \multirow{2}[4]{*}{} & \multicolumn{3}{c}{2SLS} & \multicolumn{3}{c}{OGMM} & \multicolumn{3}{c}{BGMM} \\       
							\cmidrule{2-10}\cmidrule{12-20}          & \multicolumn{1}{c}{$\tilde{g}_{1}$} & \multicolumn{1}{c}{$\tilde{g}_{2}$} & \multicolumn{1}{c}{$\tilde{g}_{3}$} & \multicolumn{1}{c}{$\tilde{g}_{1}$} & \multicolumn{1}{c}{$\tilde{g}_{2}$} & \multicolumn{1}{c}{$\tilde{g}_{3}$} & \multicolumn{1}{c}{$\tilde{g}_{1}$} & \multicolumn{1}{c}{$\tilde{g}_{2}$} & \multicolumn{1}{c}{$\tilde{g}_{3}$} &       & \multicolumn{1}{c}{$\tilde{g}_{1}$} & \multicolumn{1}{c}{$\tilde{g}_{2}$} & \multicolumn{1}{c}{$\tilde{g}_{3}$} & \multicolumn{1}{c}{$\tilde{g}_{1}$} & \multicolumn{1}{c}{$\tilde{g}_{2}$} & \multicolumn{1}{c}{$\tilde{g}_{3}$} & \multicolumn{1}{c}{$\tilde{g}_{1}$} & \multicolumn{1}{c}{$\tilde{g}_{2}$} & \multicolumn{1}{c}{$\tilde{g}_{3}$} \\
							\midrule
							MAE   & 0.0199  & 0.0209  & 0.0375  & 0.0199  & 0.0209  & 0.0368  & 0.0198  & 0.0209  & 0.0361  & MAE   & 0.0194  & 0.0209  & 0.0322  & 0.0192  & 0.0208  & 0.0317  & 0.0192  & 0.0208  & 0.0311  \\
							Bias  & -0.0159  & -0.0180  & -0.0352  & -0.0156  & -0.0181  & -0.0346  & -0.0153  & -0.0180  & -0.0342  & Bias  & -0.0128  & -0.0162  & -0.0305  & -0.0125  & -0.0163  & -0.0302  & -0.0121  & -0.0162  & -0.0298  \\
							RMSE  & 0.0164  & 0.0191  & 0.0370  & 0.0160  & 0.0191  & 0.0362  & 0.0158  & 0.0191  & 0.0357  & RMSE  & 0.0137  & 0.0175  & 0.0326  & 0.0133  & 0.0175  & 0.0321  & 0.0130  & 0.0175  & 0.0316  \\
							\midrule
							& \multicolumn{9}{c}{$(n,T,\elln)=(200,10,2)$} &   & \multicolumn{9}{c}{$(n,T,\elln)=(200,10,[n^{1/5}]+2)$} \\
							\cmidrule{2-20}     
							\multirow{2}[4]{*}{} & \multicolumn{3}{c}{2SLS} & \multicolumn{3}{c}{OGMM} & \multicolumn{3}{c}{BGMM} & \multirow{2}[4]{*}{} & \multicolumn{3}{c}{2SLS} & \multicolumn{3}{c}{OGMM} & \multicolumn{3}{c}{BGMM} \\       
							\cmidrule{2-10}\cmidrule{12-20}        & \multicolumn{1}{c}{$\tilde{g}_1$} & \multicolumn{1}{c}{$\tilde{g}_2$} & \multicolumn{1}{c}{$\tilde{g}_3$} & \multicolumn{1}{c}{$\tilde{g}_{1}$} & \multicolumn{1}{c}{$\tilde{g}_{2}$} & \multicolumn{1}{c}{$\tilde{g}_{3}$} & \multicolumn{1}{c}{$\tilde{g}_{1}$} & \multicolumn{1}{c}{$\tilde{g}_{2}$} & \multicolumn{1}{c}{$\tilde{g}_{3}$} &       & \multicolumn{1}{c}{$\tilde{g}_{1}$} & \multicolumn{1}{c}{$\tilde{g}_{2}$} & \multicolumn{1}{c}{$\tilde{g}_{3}$} & \multicolumn{1}{c}{$\tilde{g}_{1}$} & \multicolumn{1}{c}{$\tilde{g}_{2}$} & \multicolumn{1}{c}{$\tilde{g}_{3}$} & \multicolumn{1}{c}{$\tilde{g}_{1}$} & \multicolumn{1}{c}{$\tilde{g}_{2}$} & \multicolumn{1}{c}{$\tilde{g}_{3}$} \\
							\midrule
							MAE   & 0.0329  & 0.0342  & 0.0509  & 0.0329  & 0.0341  & 0.0499  & 0.0331  & 0.0343  & 0.0489  & MAE   & 0.0323  & 0.0329  & 0.0488  & 0.0322  & 0.0328  & 0.0478  & 0.0327  & 0.0330  & 0.0467  \\
							Bias  & -0.0190  & -0.0255  & -0.0475  & -0.0185  & -0.0257  & -0.0467  & -0.0180  & -0.0255  & -0.0457  & Bias  & -0.0187  & -0.0240  & -0.0450  & -0.0185  & -0.0243  & -0.0446  & -0.0175  & -0.0247  & -0.0441  \\
							RMSE  & 0.0216  & 0.0287  & 0.0541  & 0.0211  & 0.0288  & 0.0526  & 0.0209  & 0.0287  & 0.0512  & RMSE  & 0.0196  & 0.0250  & 0.0472  & 0.0194  & 0.0252  & 0.0466  & 0.0186  & 0.0257  & 0.0459  \\
							\midrule
							& \multicolumn{9}{c}{$(n,T,\elln)=(200,25,2)$} &   & \multicolumn{9}{c}{$(n,T,\elln)=(200,25,[n^{1/5}]+2)$} \\
							\cmidrule{2-20}     
							\multirow{2}[4]{*}{} & \multicolumn{3}{c}{2SLS} & \multicolumn{3}{c}{OGMM} & \multicolumn{3}{c}{BGMM} & \multirow{2}[4]{*}{} & \multicolumn{3}{c}{2SLS} & \multicolumn{3}{c}{OGMM} & \multicolumn{3}{c}{BGMM} \\    
							\cmidrule{2-10}\cmidrule{12-20}          & \multicolumn{1}{c}{$\tilde{g}_{1}$} & \multicolumn{1}{c}{$\tilde{g}_{2}$} & \multicolumn{1}{c}{$\tilde{g}_{3}$} & \multicolumn{1}{c}{$\tilde{g}_{1}$} & \multicolumn{1}{c}{$\tilde{g}_{2}$} & \multicolumn{1}{c}{$\tilde{g}_{3}$} & \multicolumn{1}{c}{$\tilde{g}_{1}$} & \multicolumn{1}{c}{$\tilde{g}_{2}$} & \multicolumn{1}{c}{$\tilde{g}_{3}$} &       & \multicolumn{1}{c}{$\tilde{g}_{1}$} & \multicolumn{1}{c}{$\tilde{g}_{2}$} & \multicolumn{1}{c}{$\tilde{g}_{3}$} & \multicolumn{1}{c}{$\tilde{g}_{1}$} & \multicolumn{1}{c}{$\tilde{g}_{2}$} & \multicolumn{1}{c}{$\tilde{g}_{3}$} & \multicolumn{1}{c}{$\tilde{g}_{1}$} & \multicolumn{1}{c}{$\tilde{g}_{2}$} & \multicolumn{1}{c}{$\tilde{g}_{3}$} \\
							\midrule
							MAE   & 0.0094  & 0.0099  & 0.0198  & 0.0094  & 0.0099  & 0.0198  & 0.0093  & 0.0099  & 0.0194  & MAE   & 0.0092  & 0.0098  & 0.0166  & 0.0091  & 0.0098  & 0.0166  & 0.0092  & 0.0099  & 0.0164  \\
							Bias  & -0.0075  & -0.0084  & -0.0181  & -0.0074  & -0.0083  & -0.0181  & -0.0074  & -0.0083  & -0.0180  & Bias  & -0.0062  & -0.0076  & -0.0159  & -0.0062  & -0.0075  & -0.0159  & -0.0061  & -0.0075  & -0.0158  \\
							RMSE  & 0.0077  & 0.0090  & 0.0190  & 0.0076  & 0.0089  & 0.0190  & 0.0076  & 0.0089  & 0.0189  & RMSE  & 0.0068  & 0.0083  & 0.0173  & 0.0068  & 0.0082  & 0.0173  & 0.0067  & 0.0083  & 0.0170  \\ 
							\midrule
							& \multicolumn{9}{c}{$(n,T,\elln)=(400,10,2)$} &   & \multicolumn{9}{c}{$(n,T,\elln)=(400,10,[n^{1/5}]+2)$} \\
							\cmidrule{2-20}     
							\multirow{2}[4]{*}{} & \multicolumn{3}{c}{2SLS} & \multicolumn{3}{c}{OGMM} & \multicolumn{3}{c}{BGMM} & \multirow{2}[4]{*}{} & \multicolumn{3}{c}{2SLS} & \multicolumn{3}{c}{OGMM} & \multicolumn{3}{c}{BGMM} \\    
							\cmidrule{2-10}\cmidrule{12-20}        & \multicolumn{1}{c}{$\tilde{g}_1$} & \multicolumn{1}{c}{$\tilde{g}_2$} & \multicolumn{1}{c}{$\tilde{g}_3$} & \multicolumn{1}{c}{$\tilde{g}_{1}$} & \multicolumn{1}{c}{$\tilde{g}_{2}$} & \multicolumn{1}{c}{$\tilde{g}_{3}$} & \multicolumn{1}{c}{$\tilde{g}_{1}$} & \multicolumn{1}{c}{$\tilde{g}_{2}$} & \multicolumn{1}{c}{$\tilde{g}_{3}$} &       & \multicolumn{1}{c}{$\tilde{g}_{1}$} & \multicolumn{1}{c}{$\tilde{g}_{2}$} & \multicolumn{1}{c}{$\tilde{g}_{3}$} & \multicolumn{1}{c}{$\tilde{g}_{1}$} & \multicolumn{1}{c}{$\tilde{g}_{2}$} & \multicolumn{1}{c}{$\tilde{g}_{3}$} & \multicolumn{1}{c}{$\tilde{g}_{1}$} & \multicolumn{1}{c}{$\tilde{g}_{2}$} & \multicolumn{1}{c}{$\tilde{g}_{3}$} \\
							\midrule
							MAE   & 0.0195  & 0.0200  & 0.0367  & 0.0195  & 0.0200  & 0.0360  & 0.0195  & 0.0200  & 0.0354  & MAE   & 0.0186  & 0.0196  & 0.0314  & 0.0184  & 0.0196  & 0.0309  & 0.0183  & 0.0197  & 0.0302  \\
							Bias  & -0.0158  & -0.0180  & -0.0357  & -0.0153  & -0.0181  & -0.0350  & -0.0148  & -0.0182  & -0.0344  & Bias  & -0.0126  & -0.0159  & -0.0313  & -0.0122  & -0.0160  & -0.0308  & -0.0115  & -0.0163  & -0.0301  \\
							RMSE  & 0.0160  & 0.0183  & 0.0363  & 0.0154  & 0.0184  & 0.0355  & 0.0149  & 0.0186  & 0.0349  & RMSE  & 0.0129  & 0.0163  & 0.0320  & 0.0125  & 0.0164  & 0.0314  & 0.0119  & 0.0167  & 0.0306  \\
							\midrule
							& \multicolumn{9}{c}{$(n,T,\elln)=(400,25,2)$} &   & \multicolumn{9}{c}{$(n,T,\elln)=(400,25,[n^{1/5}]+2)$} \\
							\cmidrule{2-20}     
							\multirow{2}[4]{*}{} & \multicolumn{3}{c}{2SLS} & \multicolumn{3}{c}{OGMM} & \multicolumn{3}{c}{BGMM} & \multirow{2}[4]{*}{} & \multicolumn{3}{c}{2SLS} & \multicolumn{3}{c}{OGMM} & \multicolumn{3}{c}{BGMM} \\    
							\cmidrule{2-10}\cmidrule{12-20}          & \multicolumn{1}{c}{$\tilde{g}_{1}$} & \multicolumn{1}{c}{$\tilde{g}_{2}$} & \multicolumn{1}{c}{$\tilde{g}_{3}$} & \multicolumn{1}{c}{$\tilde{g}_{1}$} & \multicolumn{1}{c}{$\tilde{g}_{2}$} & \multicolumn{1}{c}{$\tilde{g}_{3}$} & \multicolumn{1}{c}{$\tilde{g}_{1}$} & \multicolumn{1}{c}{$\tilde{g}_{2}$} & \multicolumn{1}{c}{$\tilde{g}_{3}$} &       & \multicolumn{1}{c}{$\tilde{g}_{1}$} & \multicolumn{1}{c}{$\tilde{g}_{2}$} & \multicolumn{1}{c}{$\tilde{g}_{3}$} & \multicolumn{1}{c}{$\tilde{g}_{1}$} & \multicolumn{1}{c}{$\tilde{g}_{2}$} & \multicolumn{1}{c}{$\tilde{g}_{3}$} & \multicolumn{1}{c}{$\tilde{g}_{1}$} & \multicolumn{1}{c}{$\tilde{g}_{2}$} & \multicolumn{1}{c}{$\tilde{g}_{3}$} \\
							\midrule
							MAE   & 0.0092  & 0.0095  & 0.0188  & 0.0092  & 0.0095  & 0.0187  & 0.0092  & 0.0096  & 0.0184  & MAE   & 0.0088  & 0.0095  & 0.0169  & 0.0088  & 0.0095  & 0.0169  & 0.0088  & 0.0096  & 0.0164  \\
							Bias  & -0.0075  & -0.0087  & -0.0181  & -0.0074  & -0.0086  & -0.0180  & -0.0073  & -0.0087  & -0.0178  & Bias  & -0.0064  & -0.0078  & -0.0169  & -0.0064  & -0.0078  & -0.0168  & -0.0060  & -0.0079  & -0.0164  \\
							RMSE  & 0.0076  & 0.0089  & 0.0183  & 0.0075  & 0.0088  & 0.0182  & 0.0074  & 0.0089  & 0.0180  & RMSE  & 0.0066  & 0.0081  & 0.0173  & 0.0066  & 0.0081  & 0.0172  & 0.0062  & 0.0082  & 0.0167  \\ 
							\bottomrule
						\end{tabular}%
						\hspace*{-1cm}
						\begin{tablenotes}
							\footnotesize
							\item \textbf{Note:} The $\tilde{g}_{k}$ extracts the column vectors composed of non-zero elements from the upper triangular submatrix of $G_{k}$, and $\bard_0=15\%$. The results are based on 1,000 Monte Carlo replications. MAE denotes the mean absolute error. Bias denotes the mean bias of the estimates, and RMSE denotes the root mean squared error. The 2SLS refers to the two-stage least square estimator, OGMM refers to the feasible optimal GMM estimator, and BGMM refers to the feasible best GMM estimator. 
						\end{tablenotes}
					\end{spacing}
					\label{sim:heteG_d0}%
					% \end{table}%
				% \end{landscape}
		\end{sidewaystable}
		%%%%%%%%%%%%%%%%%% SAR type
		\begin{table}[htbp]
			\centering \small \stla{0cm} \stl{0.8mm} 
			\caption{Finite sample performance of estimators of $\pi$ for the SAR-type, $\Var(\epsilon_{it})=\sigma_i^2$.}
			\begin{spacing}{1}%\hspace*{-5mm}
				\begin{tabular}{lcccccclcccccc}
					\toprule
					\multicolumn{7}{c}{$(n,T,\elln)=(100,10,2)$}                         & \multicolumn{7}{c}{$(n,T,\elln)=(100,10,[n^{1/5}]+2)$} \\
					\cmidrule{2-7}\cmidrule{9-14}
					\multirow{2}[4]{*}{} & \multicolumn{2}{c}{2SLS} & \multicolumn{2}{c}{OGMM} & \multicolumn{2}{c}{BGMM} & \multirow{2}[4]{*}{} & \multicolumn{2}{c}{2SLS} & \multicolumn{2}{c}{OGMM} & \multicolumn{2}{c}{BGMM} \\
					\cmidrule{2-7}\cmidrule{9-14}          & $\gamma$ & $\beta$  & $\gamma$ & $\beta$  & $\gamma$ & $\beta$  &       & $\gamma$ & $\beta$  & $\gamma$ & $\beta$  & $\gamma$ & $\beta$ \\
					\cmidrule{2-7}\cmidrule{9-14}
					Bias  & 0.0021  & -0.0024  & 0.0108  & -0.0015  & 0.0110  & -0.0018  &  & -0.0020  & -0.0041  & 0.0089  & -0.0029  & 0.0087  & -0.0027  \\
					ESD   & 0.0361  & 0.0292  & 0.0406  & 0.0287  & 0.0416  & 0.0294  &  & 0.0378  & 0.0288  & 0.0424  & 0.0283  & 0.0423  & 0.0292  \\
					RMSE  & 0.0362  & 0.0293  & 0.0420  & 0.0288  & 0.0430  & 0.0295  & & 0.0378  & 0.0291  & 0.0433  & 0.0285  & 0.0432  & 0.0294  \\
					CP    & 0.7970  & 0.9070  & 0.7950  & 0.9280  & 0.7850  & 0.9250  &  & 0.8930  & 0.9030  & 0.8830  & 0.9100  & 0.8720  & 0.8990  \\
					\cmidrule{2-7}\cmidrule{9-14}
					\multicolumn{7}{c}{$(n,T,\elln)=(100,25,2)$}                         & \multicolumn{7}{c}{$(n,T,\elln)=(100,25,[n^{1/5}]+2)$} \\
					\cmidrule{2-7}\cmidrule{9-14}
					& \multicolumn{1}{c}{2SLS} &       & \multicolumn{1}{c}{OGMM} &       & \multicolumn{1}{c}{BGMM} &       & \multirow{2}[4]{*}{} & \multicolumn{2}{c}{2SLS} & \multicolumn{2}{c}{OGMM} & \multicolumn{2}{c}{BGMM} \\
					\cmidrule{2-7}\cmidrule{9-14}          &$\gamma$ & $\beta$  & $\gamma$ & $\beta$  & $\gamma$ & $\beta$  &       & $\gamma$ & $\beta$  & $\gamma$ & $\beta$  & $\gamma$ & $\beta$ \\
					\cmidrule{2-7}\cmidrule{9-14}
					Bias  & 0.0163  & -0.0002  & 0.0163  & -0.0002  & 0.0149  & 0.0000  & & 0.0152  & 0.0012  & 0.0152  & 0.0012  & 0.0125  & 0.0002  \\
					ESD   & 0.0221  & 0.0175  & 0.0221  & 0.0163  & 0.0212  & 0.0163  & & 0.0202  & 0.0173  & 0.0202  & 0.0165  & 0.0196  & 0.0164  \\
					RMSE  & 0.0275  & 0.0175  & 0.0275  & 0.0163  & 0.0259  & 0.0163  & & 0.0253  & 0.0173  & 0.0253  & 0.0165  & 0.0232  & 0.0164  \\
					CP    & 0.8960  & 0.9070  & 0.8840  & 0.9270  & 0.8730  & 0.9160  &  & 0.8980  & 0.9140  & 0.8990  & 0.9180  & 0.8880  & 0.9070  \\
					\cmidrule{2-7}\cmidrule{9-14}
					\multicolumn{7}{c}{$(n,T,\elln)=(200,10,2)$}                         & \multicolumn{7}{c}{$(n,T,\elln)=(200,10,[n^{1/5}]+2)$} \\
					\cmidrule{2-7}\cmidrule{9-14}
					\multirow{2}[4]{*}{} & \multicolumn{2}{c}{2SLS} & \multicolumn{2}{c}{OGMM} & \multicolumn{2}{c}{BGMM} & \multirow{2}[4]{*}{} & \multicolumn{2}{c}{2SLS} & \multicolumn{2}{c}{OGMM} & \multicolumn{2}{c}{BGMM} \\
					\cmidrule{2-7}\cmidrule{9-14}          & $\gamma$ & $\beta$  & $\gamma$ & $\beta$  & $\gamma$ & $\beta$  &       & $\gamma$ & $\beta$  & $\gamma$ & $\beta$  & $\gamma$ & $\beta$ \\
					\cmidrule{2-7}\cmidrule{9-14}
					Bias  & 0.0088  & -0.0017  & 0.0086  & -0.0007  & 0.0035  & -0.0006  &  & 0.0057  & -0.0002  & 0.0054  & -0.0003  & -0.0017  & -0.0023  \\
					ESD   & 0.0286  & 0.0206  & 0.0283  & 0.0203  & 0.0270  & 0.0200  &  & 0.0270  & 0.0211  & 0.0266  & 0.0208  & 0.0252  & 0.0206  \\
					RMSE  & 0.0299  & 0.0207  & 0.0296  & 0.0203  & 0.0272  & 0.0200  &  & 0.0276  & 0.0211  & 0.0271  & 0.0208  & 0.0253  & 0.0207  \\
					CP    & 0.8000  & 0.8990  & 0.8060  & 0.9260  & 0.7960  & 0.9150  &  & 0.8430  & 0.8880  & 0.8140  & 0.9200  & 0.8040  & 0.9090  \\
					\cmidrule{2-7}\cmidrule{9-14}
					\multicolumn{7}{c}{$(n,T,\elln)=(200,25,2)$}                         & \multicolumn{7}{c}{$(n,T,\elln)=(200,25,[n^{1/5}]+2)$}
					\\
					\cmidrule{2-7}\cmidrule{9-14}
					\multirow{2}[4]{*}{} & \multicolumn{2}{c}{2SLS} & \multicolumn{2}{c}{OGMM} & \multicolumn{2}{c}{BGMM} & \multirow{2}[4]{*}{} & \multicolumn{2}{c}{2SLS} & \multicolumn{2}{c}{OGMM} & \multicolumn{2}{c}{BGMM} \\
					\cmidrule{2-7}\cmidrule{9-14}           & $\gamma$ & $\beta$  & $\gamma$ & $\beta$  & $\gamma$ & $\beta$  &       & $\gamma$ & $\beta$  & $\gamma$ & $\beta$  & $\gamma$ & $\beta$ \\
					\cmidrule{2-7}\cmidrule{9-14}
					Bias  & 0.0123  & 0.0019  & 0.0123  & 0.0019  & 0.0109  & -0.0002  &  & 0.0089  & 0.0017  & 0.0089  & 0.0017  & 0.0063  & -0.0006  \\
					ESD   & 0.0150  & 0.0118  & 0.0150  & 0.0115  & 0.0147  & 0.0115  &  & 0.0141  & 0.0120  & 0.0139  & 0.0115  & 0.0139  & 0.0115  \\
					RMSE  & 0.0194  & 0.0120  & 0.0194  & 0.0117  & 0.0183  & 0.0115  &   & 0.0167  & 0.0121  & 0.0165  & 0.0116  & 0.0153  & 0.0115  \\
					CP    & 0.9260  & 0.9400  & 0.9360  & 0.9330  & 0.9250  & 0.9220  &  & 0.9400  & 0.9350  & 0.9420  & 0.9420  & 0.9350  & 0.9400  \\
					\cmidrule{2-7}\cmidrule{9-14}
					\multicolumn{7}{c}{$(n,T,\elln)=(400,10,2)$}                         & \multicolumn{7}{c}{$(n,T,\elln)=(400,10,[n^{1/5}]+2)$} \\
					\cmidrule{2-7}\cmidrule{9-14}
					\multirow{2}[4]{*}{} & \multicolumn{2}{c}{2SLS} & \multicolumn{2}{c}{OGMM} & \multicolumn{2}{c}{BGMM} & \multirow{2}[4]{*}{} & \multicolumn{2}{c}{2SLS} & \multicolumn{2}{c}{OGMM} & \multicolumn{2}{c}{BGMM} \\
					\cmidrule{2-7}\cmidrule{9-14}          & $\gamma$ & $\beta$  & $\gamma$ & $\beta$  & $\gamma$ & $\beta$  &       & $\gamma$ & $\beta$  & $\gamma$ & $\beta$  & $\gamma$ & $\beta$ \\
					\cmidrule{2-7}\cmidrule{9-14}
					Bias  & 0.0063  & 0.0011  & 0.0062  & 0.0011  & 0.0033  & -0.0004  &  & 0.0050  & 0.0018  & 0.0049  & 0.0015  & 0.0009  & 0.0000  \\
					ESD   & 0.0198  & 0.0145  & 0.0198  & 0.0144  & 0.0193  & 0.0144  & & 0.0189  & 0.0146  & 0.0185  & 0.0145  & 0.0181  & 0.0144  \\
					RMSE  & 0.0208  & 0.0145  & 0.0207  & 0.0144  & 0.0196  & 0.0144  & & 0.0196  & 0.0147  & 0.0191  & 0.0146  & 0.0181  & 0.0144  \\
					CP    & 0.9120  & 0.9220  & 0.9270  & 0.9300  & 0.9190  & 0.9320  & & 0.9130  & 0.9300  & 0.9320  & 0.9360  & 0.9240  & 0.9330  \\
					\cmidrule{2-7}\cmidrule{9-14}
					\multicolumn{7}{c}{$(n,T,\elln)=(400,25,2)$}                         & \multicolumn{7}{c}{$(n,T,\elln)=(400,25,[n^{1/5}]+2)$}
					\\
					\cmidrule{2-7}\cmidrule{9-14}
					\multirow{2}[4]{*}{} & \multicolumn{2}{c}{2SLS} & \multicolumn{2}{c}{OGMM} & \multicolumn{2}{c}{BGMM} & \multirow{2}[4]{*}{} & \multicolumn{2}{c}{2SLS} & \multicolumn{2}{c}{OGMM} & \multicolumn{2}{c}{BGMM} \\
					\cmidrule{2-7}\cmidrule{9-14}           & $\gamma$ & $\beta$  & $\gamma$ & $\beta$  & $\gamma$ & $\beta$  &       & $\gamma$ & $\beta$  & $\gamma$ & $\beta$  & $\gamma$ & $\beta$ \\
					\cmidrule{2-7}\cmidrule{9-14}
					Bias  & 0.0073  & 0.0019  & 0.0073  & 0.0019  & 0.0070  & 0.0001  & & 0.0060  & 0.0018  & 0.0060  & 0.0018  & 0.0054  & -0.0007  \\
					ESD   & 0.0102  & 0.0081  & 0.0102  & 0.0081  & 0.0100  & 0.0081  & & 0.0095  & 0.0084  & 0.0095  & 0.0084  & 0.0095  & 0.0083  \\
					RMSE  & 0.0125  & 0.0083  & 0.0125  & 0.0083  & 0.0122  & 0.0081  &   & 0.0112  & 0.0086  & 0.0112  & 0.0086  & 0.0109  & 0.0083  \\
					CP    & 0.9250  & 0.9360  & 0.9320  & 0.9380  & 0.9250  & 0.9360  &  & 0.9340  & 0.9420  & 0.9400  & 0.9450  & 0.9410  & 0.9420  \\
					\bottomrule
				\end{tabular}%
				\hspace*{-1cm}
				\begin{tablenotes}
					\footnotesize
					\item \textbf{Note:} The true parameters are set to $\pi_0=(0.3,1)'$, and $\bard_0=10\%$. The results are based on 1,000 Monte Carlo replications. Bias denotes the mean bias of the estimates, ESD denotes the standard deviation, RMSE denotes the root mean squared error, and CP denotes the 95\% coverage probability. The 2SLS refers to the two-stage least square estimator, OGMM refers to the feasible optimal GMM estimator, and BGMM refers to the feasible best GMM estimator. 
				\end{tablenotes}
			\end{spacing}
			\label{sim:hei_sar}%
		\end{table}%
		
		\begin{table}[htbp]
			\centering \small \stla{0cm} \stl{0.8mm} 
			\caption{Estimated $\rho(\hat{A})$ for the SAR-type, $\Var(\epsilon_{it})=\sigma_i^2$, $\rho(A)=0.7333$.}
			\begin{spacing}{1}\hspace*{-0.8cm}
				\begin{tabular}{lcccccclcccccc}
					\toprule
					\multicolumn{7}{c}{$(n,T)=(100,10)$}                         & \multicolumn{7}{c}{$(n,T,\elln)=(100,25)$} \\
					\cmidrule{2-7}\cmidrule{9-14}
					\multirow{2}[4]{*}{} & \multicolumn{3}{c}{$\elln=2$} & \multicolumn{3}{c}{$\elln=[n^{1/5}]+2$} & \multirow{2}[4]{*}{} & \multicolumn{3}{c}{$\elln=2$} & \multicolumn{3}{c}{$\elln=[n^{1/5}]+2$} \\
					\cmidrule{2-7}\cmidrule{9-14}          & 2SLS &OGMM & BGMM & 2SLS &OGMM & BGMM  &       & 2SLS &OGMM & BGMM & 2SLS &OGMM & BGMM \\
					\cmidrule{2-7}\cmidrule{9-14}    Mean  & 0.6065  & 0.6044  & 0.6096  & 0.6096  & 0.6118  & 0.6121  &       & 0.6485  & 0.6274  & 0.6282  & 0.6566  & 0.6464  & 0.6454  \\
					% Bias  & -0.1268  & -0.1289  & -0.1237  & -0.1237  & -0.1215  & -0.1212  &       & -0.0849  & -0.1060  & -0.1051  & -0.0768  & -0.0870  & -0.0879  \\
					ESD   & 0.1834  & 0.1786  & 0.1785  & 0.1724  & 0.1742  & 0.1735  &       & 0.1353  & 0.1420  & 0.1408  & 0.1327  & 0.1365  & 0.1355  \\
					RMSE  & 0.2230  & 0.2202  & 0.2172  & 0.2122  & 0.2124  & 0.2117  &       & 0.1597  & 0.1772  & 0.1757  & 0.1533  & 0.1618  & 0.1615  \\
					\cmidrule{2-7}\cmidrule{9-14}
					\multicolumn{7}{c}{$(n,T)=(200,10)$}                         & \multicolumn{7}{c}{$(n,T,\elln)=(200,25)$} \\
					\cmidrule{2-7}\cmidrule{9-14}
					\multirow{2}[4]{*}{} & \multicolumn{3}{c}{$\elln=2$} & \multicolumn{3}{c}{$\elln=[n^{1/5}]+2$} & \multirow{2}[4]{*}{} & \multicolumn{3}{c}{$\elln=2$} & \multicolumn{3}{c}{$\elln=[n^{1/5}]+2$} \\
					\cmidrule{2-7}\cmidrule{9-14}          & 2SLS &OGMM & BGMM & 2SLS &OGMM & BGMM  &       & 2SLS &OGMM & BGMM & 2SLS &OGMM & BGMM \\
					\cmidrule{2-7}\cmidrule{9-14}
					Mean  & 0.6319  & 0.6355  & 0.6614  & 0.6368  & 0.6349  & 0.6355  &   & 0.6885  & 0.6740  & 0.6741  & 0.7007  & 0.6926  & 0.6924  \\
					% Bias  & -0.1014  & -0.0978  & -0.0719  & -0.0965  & -0.0984  & -0.0978  & Bias  & -0.0448  & -0.0593  & -0.0592  & -0.0326  & -0.0407  & -0.0409  \\
					ESD    & 0.1520  & 0.1542  & 0.1535  & 0.1508  & 0.1557  & 0.1535  &     & 0.0918  & 0.0827  & 0.0827  & 0.0898  & 0.0859  & 0.0862  \\
					RMSE  & 0.1827  & 0.1826  & 0.1695  & 0.1791  & 0.1842  & 0.1820  &   & 0.1022  & 0.1018  & 0.1017  & 0.0955  & 0.0951  & 0.0954  \\
					\cmidrule{2-7}\cmidrule{9-14}
					\multicolumn{7}{c}{$(n,T)=(400,10)$}                         & \multicolumn{7}{c}{$(n,T,\elln)=(400,25)$} \\
					\cmidrule{2-7}\cmidrule{9-14}
					\multirow{2}[4]{*}{} & \multicolumn{3}{c}{$\elln=2$} & \multicolumn{3}{c}{$\elln=[n^{1/5}]+2$} & \multirow{2}[4]{*}{} & \multicolumn{3}{c}{$\elln=2$} & \multicolumn{3}{c}{$\elln=[n^{1/5}]+2$} \\
					\cmidrule{2-7}\cmidrule{9-14}          & 2SLS &OGMM & BGMM & 2SLS &OGMM & BGMM  &       & 2SLS &OGMM & BGMM & 2SLS &OGMM & BGMM \\
					\cmidrule{2-7}\cmidrule{9-14}
					Mean  & 0.6818  & 0.6535  & 0.6535  & 0.6901  & 0.6619  & 0.6606  &   & 0.7489  & 0.7131  & 0.7131  & 0.7533  & 0.7534  & 0.7545  \\
					% Bias  & -0.0515  & -0.0799  & -0.0798  & -0.0432  & -0.0714  & -0.0728  & Bias  & 0.0156  & -0.0202  & -0.0202  & 0.0200  & 0.0201  & 0.0212  \\
					ESD    & 0.1237  & 0.1162  & 0.1162  & 0.1205  & 0.1160  & 0.1161  &     & 0.0855  & 0.0835  & 0.0835  & 0.0803  & 0.0806  & 0.0812  \\
					RMSE  & 0.1340  & 0.1410  & 0.1410  & 0.1280  & 0.1363  & 0.1370  &   & 0.0870  & 0.0859  & 0.0859  & 0.0828  & 0.0830  & 0.0839  \\
					\bottomrule
				\end{tabular}%
				\hspace*{-1cm}
				\begin{tablenotes}
					\footnotesize
					\item \textbf{Note:} The $\bard_0=10\%$, and the results are based on 1,000 Monte Carlo replications. Mean denotes the mean of the estimates. Bias denotes the mean bias of the estimates, ESD denotes the standard deviation, and RMSE denotes the root mean squared error. The 2SLS refers to the two-stage least square estimator, OGMM refers to the feasible optimal GMM estimator, and BGMM refers to the feasible best GMM estimator. 
				\end{tablenotes}
			\end{spacing}
			\label{sim:rhoG_sar}%
		\end{table}%
		
		% \begin{landscape}
			% \begin{table}[H]
				\begin{sidewaystable}[htbp]
					\centering\footnotesize\stla{0cm} \stl{0.8mm} \vspace*{-5mm}
					\caption{Finite sample performance of estimators of $G_{1}, G_{2}$ and $G_{3}$ for the SAR-type, $\Var(\epsilon_{it})=\sigma_i^2$.}
					\begin{spacing}{0.5}\hspace*{-5mm}
						\begin{tabular}{lrrrrrrrrrlrrrrrrrrr}
							\toprule
							\multirow{4}[3]{*}{}	& \multicolumn{9}{c}{$(n,T,\elln)=(100,10,2)$} &   & \multicolumn{9}{c}{$(n,T,\elln)=(100,10,[n^{1/5}]+2)$} \\
							\cmidrule{2-20}     
							\multirow{2}[4]{*}{} & \multicolumn{3}{c}{2SLS} & \multicolumn{3}{c}{OGMM} & \multicolumn{3}{c}{BGMM} & \multirow{2}[4]{*}{} & \multicolumn{3}{c}{2SLS} & \multicolumn{3}{c}{OGMM} & \multicolumn{3}{c}{BGMM} \\        
							\cmidrule{2-10}\cmidrule{12-20}        & \multicolumn{1}{c}{$\tilde{g}_1$} & \multicolumn{1}{c}{$\tilde{g}_2$} & \multicolumn{1}{c}{$\tilde{g}_3$} & \multicolumn{1}{c}{$\tilde{g}_{1}$} & \multicolumn{1}{c}{$\tilde{g}_{2}$} & \multicolumn{1}{c}{$\tilde{g}_{3}$} & \multicolumn{1}{c}{$\tilde{g}_{1}$} & \multicolumn{1}{c}{$\tilde{g}_{2}$} & \multicolumn{1}{c}{$\tilde{g}_{3}$} &       & \multicolumn{1}{c}{$\tilde{g}_{1}$} & \multicolumn{1}{c}{$\tilde{g}_{2}$} & \multicolumn{1}{c}{$\tilde{g}_{3}$} & \multicolumn{1}{c}{$\tilde{g}_{1}$} & \multicolumn{1}{c}{$\tilde{g}_{2}$} & \multicolumn{1}{c}{$\tilde{g}_{3}$} & \multicolumn{1}{c}{$\tilde{g}_{1}$} & \multicolumn{1}{c}{$\tilde{g}_{2}$} & \multicolumn{1}{c}{$\tilde{g}_{3}$} \\
							\midrule
							MAE   & 0.0158  & 0.0178  & -     & 0.0158  & 0.0190  & 0.0193  & 0.0159  & 0.0190  & 0.0195  & MAE   & 0.0119  & 0.0123  & -     & 0.0117  & 0.0124  & 0.0122  & 0.0116  & 0.0124  & 0.0123  \\
							Bias  & 0.0017  & -0.0039  & -     & -0.0060  & -0.0022  & 0.0016  & -0.0056  & -0.0023  & 0.0017  & Bias  & 0.0013  & -0.0012  & -     & -0.0070  & -0.0005  & 0.0012  & -0.0069  & -0.0005  & 0.0013  \\
							RMSE  & 0.0157  & 0.0193  & -     & 0.0165  & 0.0213  & 0.0206  & 0.0167  & 0.0212  & 0.0208  & RMSE  & 0.0081  & 0.0098  & -     & 0.0102  & 0.0099  & 0.0084  & 0.0101  & 0.0099  & 0.0084  \\
							\midrule
							& \multicolumn{9}{c}{$(n,T,\elln)=(100,25,2)$} &   & \multicolumn{9}{c}{$(n,T,\elln)=(100,25,[n^{1/5}]+2)$} \\
							\cmidrule{2-20}     
							\multirow{2}[4]{*}{} & \multicolumn{3}{c}{2SLS} & \multicolumn{3}{c}{OGMM} & \multicolumn{3}{c}{BGMM} & \multirow{2}[4]{*}{} & \multicolumn{3}{c}{2SLS} & \multicolumn{3}{c}{OGMM} & \multicolumn{3}{c}{BGMM} \\       
							\cmidrule{2-10}\cmidrule{12-20}          & \multicolumn{1}{c}{$\tilde{g}_{1}$} & \multicolumn{1}{c}{$\tilde{g}_{2}$} & \multicolumn{1}{c}{$\tilde{g}_{3}$} & \multicolumn{1}{c}{$\tilde{g}_{1}$} & \multicolumn{1}{c}{$\tilde{g}_{2}$} & \multicolumn{1}{c}{$\tilde{g}_{3}$} & \multicolumn{1}{c}{$\tilde{g}_{1}$} & \multicolumn{1}{c}{$\tilde{g}_{2}$} & \multicolumn{1}{c}{$\tilde{g}_{3}$} &       & \multicolumn{1}{c}{$\tilde{g}_{1}$} & \multicolumn{1}{c}{$\tilde{g}_{2}$} & \multicolumn{1}{c}{$\tilde{g}_{3}$} & \multicolumn{1}{c}{$\tilde{g}_{1}$} & \multicolumn{1}{c}{$\tilde{g}_{2}$} & \multicolumn{1}{c}{$\tilde{g}_{3}$} & \multicolumn{1}{c}{$\tilde{g}_{1}$} & \multicolumn{1}{c}{$\tilde{g}_{2}$} & \multicolumn{1}{c}{$\tilde{g}_{3}$} \\
							\midrule
							MAE   & 0.0124  & 0.0124  & -     & 0.0117  & 0.0129  & 0.0122  & 0.0117  & 0.0129  & 0.0122  & MAE   & 0.0119  & 0.0123  & -     & 0.0117  & 0.0124  & 0.0122  & 0.0116  & 0.0124  & 0.0123  \\
							Bias  & 0.0026  & -0.0014  & -     & -0.0056  & -0.0001  & 0.0018  & -0.0054  & -0.0001  & 0.0019  & Bias  & 0.0013  & -0.0012  & -     & -0.0070  & -0.0005  & 0.0012  & -0.0069  & -0.0005  & 0.0013  \\
							RMSE  & 0.0099  & 0.0110  & -     & 0.0103  & 0.0115  & 0.0097  & 0.0103  & 0.0115  & 0.0097  & RMSE  & 0.0081  & 0.0098  & -     & 0.0102  & 0.0099  & 0.0084  & 0.0101  & 0.0099  & 0.0084  \\
							\midrule
							& \multicolumn{9}{c}{$(n,T,\elln)=(200,10,2)$} &   & \multicolumn{9}{c}{$(n,T,\elln)=(200,10,[n^{1/5}]+2)$} \\
							\cmidrule{2-20}     
							\multirow{2}[4]{*}{} & \multicolumn{3}{c}{2SLS} & \multicolumn{3}{c}{OGMM} & \multicolumn{3}{c}{BGMM} & \multirow{2}[4]{*}{} & \multicolumn{3}{c}{2SLS} & \multicolumn{3}{c}{OGMM} & \multicolumn{3}{c}{BGMM} \\       
							\cmidrule{2-10}\cmidrule{12-20}        & \multicolumn{1}{c}{$\tilde{g}_1$} & \multicolumn{1}{c}{$\tilde{g}_2$} & \multicolumn{1}{c}{$\tilde{g}_3$} & \multicolumn{1}{c}{$\tilde{g}_{1}$} & \multicolumn{1}{c}{$\tilde{g}_{2}$} & \multicolumn{1}{c}{$\tilde{g}_{3}$} & \multicolumn{1}{c}{$\tilde{g}_{1}$} & \multicolumn{1}{c}{$\tilde{g}_{2}$} & \multicolumn{1}{c}{$\tilde{g}_{3}$} &       & \multicolumn{1}{c}{$\tilde{g}_{1}$} & \multicolumn{1}{c}{$\tilde{g}_{2}$} & \multicolumn{1}{c}{$\tilde{g}_{3}$} & \multicolumn{1}{c}{$\tilde{g}_{1}$} & \multicolumn{1}{c}{$\tilde{g}_{2}$} & \multicolumn{1}{c}{$\tilde{g}_{3}$} & \multicolumn{1}{c}{$\tilde{g}_{1}$} & \multicolumn{1}{c}{$\tilde{g}_{2}$} & \multicolumn{1}{c}{$\tilde{g}_{3}$} \\
							\midrule
							MAE   & 0.0082  & 0.0090  & -     & 0.0077  & 0.0094  & 0.0093  & 0.0078  & 0.0094  & 0.0092  & MAE   & 0.0074  & 0.0079  & -     & 0.0069  & 0.0083  & 0.0071  & 0.0069  & 0.0083  & 0.0071  \\
							Bias  & 0.0022  & -0.0026  & -     & -0.0015  & -0.0013  & 0.0001  & -0.0015  & -0.0013  & 0.0000  & Bias  & 0.0011  & -0.0025  & -     & -0.0030  & -0.0013  & -0.0004  & -0.0029  & -0.0013  & -0.0004  \\
							RMSE  & 0.0083  & 0.0100  & -     & 0.0081  & 0.0105  & 0.0100  & 0.0083  & 0.0105  & 0.0099  & RMSE  & 0.0061  & 0.0079  & -     & 0.0065  & 0.0084  & 0.0061  & 0.0065  & 0.0083  & 0.0061  \\
							\midrule
							& \multicolumn{9}{c}{$(n,T,\elln)=(200,25,2)$} &   & \multicolumn{9}{c}{$(n,T,\elln)=(200,25,[n^{1/5}]+2)$} \\
							\cmidrule{2-20}     
							\multirow{2}[4]{*}{} & \multicolumn{3}{c}{2SLS} & \multicolumn{3}{c}{OGMM} & \multicolumn{3}{c}{BGMM} & \multirow{2}[4]{*}{} & \multicolumn{3}{c}{2SLS} & \multicolumn{3}{c}{OGMM} & \multicolumn{3}{c}{BGMM} \\    
							\cmidrule{2-10}\cmidrule{12-20}          & \multicolumn{1}{c}{$\tilde{g}_{1}$} & \multicolumn{1}{c}{$\tilde{g}_{2}$} & \multicolumn{1}{c}{$\tilde{g}_{3}$} & \multicolumn{1}{c}{$\tilde{g}_{1}$} & \multicolumn{1}{c}{$\tilde{g}_{2}$} & \multicolumn{1}{c}{$\tilde{g}_{3}$} & \multicolumn{1}{c}{$\tilde{g}_{1}$} & \multicolumn{1}{c}{$\tilde{g}_{2}$} & \multicolumn{1}{c}{$\tilde{g}_{3}$} &       & \multicolumn{1}{c}{$\tilde{g}_{1}$} & \multicolumn{1}{c}{$\tilde{g}_{2}$} & \multicolumn{1}{c}{$\tilde{g}_{3}$} & \multicolumn{1}{c}{$\tilde{g}_{1}$} & \multicolumn{1}{c}{$\tilde{g}_{2}$} & \multicolumn{1}{c}{$\tilde{g}_{3}$} & \multicolumn{1}{c}{$\tilde{g}_{1}$} & \multicolumn{1}{c}{$\tilde{g}_{2}$} & \multicolumn{1}{c}{$\tilde{g}_{3}$} \\
							\midrule
							MAE   & 0.0064  & 0.0060  & -     & 0.0057  & 0.0059  & 0.0058  & 0.0057  & 0.0059  & 0.0058  & MAE   & 0.0063  & 0.0058  & -     & 0.0056  & 0.0059  & 0.0056  & 0.0056  & 0.0059  & 0.0056  \\
							Bias  & 0.0029  & -0.0018  & -     & -0.0004  & -0.0009  & 0.0000  & -0.0005  & -0.0009  & -0.0001  & Bias  & 0.0017  & -0.0017  & -     & -0.0011  & -0.0010  & -0.0008  & -0.0011  & -0.0010  & -0.0008  \\
							RMSE  & 0.0054  & 0.0055  & -     & 0.0047  & 0.0053  & 0.0048  & 0.0048  & 0.0053  & 0.0047  & RMSE  & 0.0039  & 0.0045  & -     & 0.0036  & 0.0044  & 0.0037  & 0.0036  & 0.0044  & 0.0037  \\
							\midrule
							& \multicolumn{9}{c}{$(n,T,\elln)=(400,10,2)$} &   & \multicolumn{9}{c}{$(n,T,\elln)=(400,10,[n^{1/5}]+2)$} \\
							\cmidrule{2-20}     
							\multirow{2}[4]{*}{} & \multicolumn{3}{c}{2SLS} & \multicolumn{3}{c}{OGMM} & \multicolumn{3}{c}{BGMM} & \multirow{2}[4]{*}{} & \multicolumn{3}{c}{2SLS} & \multicolumn{3}{c}{OGMM} & \multicolumn{3}{c}{BGMM} \\    
							\cmidrule{2-10}\cmidrule{12-20}        & \multicolumn{1}{c}{$\tilde{g}_1$} & \multicolumn{1}{c}{$\tilde{g}_2$} & \multicolumn{1}{c}{$\tilde{g}_3$} & \multicolumn{1}{c}{$\tilde{g}_{1}$} & \multicolumn{1}{c}{$\tilde{g}_{2}$} & \multicolumn{1}{c}{$\tilde{g}_{3}$} & \multicolumn{1}{c}{$\tilde{g}_{1}$} & \multicolumn{1}{c}{$\tilde{g}_{2}$} & \multicolumn{1}{c}{$\tilde{g}_{3}$} &       & \multicolumn{1}{c}{$\tilde{g}_{1}$} & \multicolumn{1}{c}{$\tilde{g}_{2}$} & \multicolumn{1}{c}{$\tilde{g}_{3}$} & \multicolumn{1}{c}{$\tilde{g}_{1}$} & \multicolumn{1}{c}{$\tilde{g}_{2}$} & \multicolumn{1}{c}{$\tilde{g}_{3}$} & \multicolumn{1}{c}{$\tilde{g}_{1}$} & \multicolumn{1}{c}{$\tilde{g}_{2}$} & \multicolumn{1}{c}{$\tilde{g}_{3}$} \\
							\midrule
							MAE   & 0.0042  & 0.0046  & -     & 0.0039  & 0.0048  & 0.0040  & 0.0039  & 0.0048  & 0.0040  & MAE   & 0.0039  & 0.0043  & -     & 0.0035  & 0.0044  & 0.0037  & 0.0035  & 0.0044  & 0.0037  \\
							Bias  & 0.0016  & -0.0004  & -     & 0.0003  & 0.0002  & 0.0007  & 0.0003  & 0.0002  & 0.0008  & Bias  & 0.0013  & -0.0005  & -     & -0.0001  & 0.0001  & 0.0004  & -0.0001  & 0.0001  & 0.0005  \\
							RMSE  & 0.0044  & 0.0052  & -     & 0.0040  & 0.0054  & 0.0045  & 0.0041  & 0.0054  & 0.0044  & RMSE  & 0.0035  & 0.0045  & -     & 0.0031  & 0.0046  & 0.0033  & 0.0031  & 0.0046  & 0.0033  \\ 
							\midrule
							& \multicolumn{9}{c}{$(n,T,\elln)=(400,25,2)$} &   & \multicolumn{9}{c}{$(n,T,\elln)=(400,25,[n^{1/5}]+2)$} \\
							\cmidrule{2-20}     
							\multirow{2}[4]{*}{} & \multicolumn{3}{c}{2SLS} & \multicolumn{3}{c}{OGMM} & \multicolumn{3}{c}{BGMM} & \multirow{2}[4]{*}{} & \multicolumn{3}{c}{2SLS} & \multicolumn{3}{c}{OGMM} & \multicolumn{3}{c}{BGMM} \\    
							\cmidrule{2-10}\cmidrule{12-20}          & \multicolumn{1}{c}{$\tilde{g}_{1}$} & \multicolumn{1}{c}{$\tilde{g}_{2}$} & \multicolumn{1}{c}{$\tilde{g}_{3}$} & \multicolumn{1}{c}{$\tilde{g}_{1}$} & \multicolumn{1}{c}{$\tilde{g}_{2}$} & \multicolumn{1}{c}{$\tilde{g}_{3}$} & \multicolumn{1}{c}{$\tilde{g}_{1}$} & \multicolumn{1}{c}{$\tilde{g}_{2}$} & \multicolumn{1}{c}{$\tilde{g}_{3}$} &       & \multicolumn{1}{c}{$\tilde{g}_{1}$} & \multicolumn{1}{c}{$\tilde{g}_{2}$} & \multicolumn{1}{c}{$\tilde{g}_{3}$} & \multicolumn{1}{c}{$\tilde{g}_{1}$} & \multicolumn{1}{c}{$\tilde{g}_{2}$} & \multicolumn{1}{c}{$\tilde{g}_{3}$} & \multicolumn{1}{c}{$\tilde{g}_{1}$} & \multicolumn{1}{c}{$\tilde{g}_{2}$} & \multicolumn{1}{c}{$\tilde{g}_{3}$} \\
							\midrule
							MAE   & 0.0035  & 0.0031  & -     & 0.0032  & 0.0031  & 0.0027  & 0.0032  & 0.0031  & 0.0027  & MAE   & 0.0034  & 0.0032  & -     & 0.0029  & 0.0031  & 0.0025  & 0.0029  & 0.0031  & 0.0025  \\
							Bias  & 0.0022  & 0.0001  & -     & 0.0016  & 0.0002  & 0.0005  & 0.0016  & 0.0002  & 0.0005  & Bias  & 0.0017  & 0.0003  & -     & 0.0008  & 0.0003  & 0.0000  & 0.0008  & 0.0003  & 0.0000  \\
							RMSE  & 0.0033  & 0.0030  & -     & 0.0028  & 0.0030  & 0.0022  & 0.0028  & 0.0030  & 0.0022  & RMSE  & 0.0026  & 0.0026  & -     & 0.0017  & 0.0026  & 0.0015  & 0.0017  & 0.0026  & 0.0015  \\ 
							\bottomrule
						\end{tabular}%
						\hspace*{-1cm}
						\begin{tablenotes}
							\footnotesize
							\item \textbf{Note:} The $\tilde{g}_{k}$ extracts the column vectors composed of non-zero elements from the upper triangular submatrix of $G_{k}$, and $\bard_0=10\%$. The results are based on 1,000 Monte Carlo replications. MAE denotes the mean absolute error. Bias denotes the mean bias of the estimates, and RMSE denotes the root mean squared error. The 2SLS refers to the two-stage least square estimator, OGMM refers to the feasible optimal GMM estimator, and BGMM refers to the feasible best GMM estimator. 
						\end{tablenotes}
					\end{spacing}
					\label{sim:heteG_sar}%
				\end{sidewaystable}
\end{document}